\documentclass[article, onecolumn]{IEEEtran}

\usepackage{setspace}
\usepackage{amsmath}
\usepackage{mdwmath}
\usepackage{mdwtab}
\usepackage{amssymb}
\usepackage{graphicx}
\usepackage{amsthm}
\usepackage{latexsym}
\usepackage{amsfonts}
\usepackage{amscd}
\usepackage{verbatim}
\usepackage{graphicx}
\usepackage{young}
\usepackage[vcentermath]{youngtab}
\usepackage{setspace}
\usepackage{tikz}
\usepackage[thinlines]{easytable}
\usetikzlibrary{decorations.pathreplacing}
\usepackage[mathscr]{euscript}
\usetikzlibrary{arrows, calc}
\usepackage{chngcntr}
\counterwithout{figure}{table}
\newcommand\numberthis{\addtocounter{equation}{1}\tag{\theequation}}
\theoremstyle{plain}
\newtheorem{thm}{Theorem}[section]
\newtheorem{prop}[thm]{Proposition}
\newtheorem{cor}[thm]{Corollary}

\newtheorem{example}[thm]{Example}
\newtheorem{theorem}{Theorem}[section]
\newtheorem{proposition}[thm]{Proposition}

\newtheorem{lemma}[thm]{Lemma}

\theoremstyle{definition}
\newtheorem*{definition}{Definition}

\newtheorem{remark}[thm]{Remark}

\newenvironment{proofoutline}
{\proof[Proof outline]}
\begin{document}
\title{Ulam Sphere Size Analysis for \\Permutation 
and Multipermutation Codes \\Correcting Translocation Errors
}

\author{Justin~Kong%,~\IEEEmembership{Student Member,~IEEE,}

\thanks{J. Kong is with the 
Department of Mathematics and Informatics, 
Graduate School of Science, 
Chiba University 
1-33 Yayoi-cho, Inage-ku, Chiba City 
Chiba Pref., 263-0022 JAPAN, email: jkong@math.s.chiba-u.ac.jp }}

\maketitle

\begin{abstract}
Permutation and multipermutation codes in the Ulam metric have 
been suggested for use in non-volatile memory storage systems 
such as flash memory devices. 
In this paper we introduce a new method to calculate permutation 
sphere sizes in the Ulam metric using Young Tableaux and prove 
the non-existence of non-trivial perfect permutation codes in the 
Ulam metric. We then extend the study to multipermutations, providing 
tight upper and lower bounds on multipermutation Ulam sphere 
sizes and resulting upper and lower bounds on the maximal 
size of multipermutation codes in the Ulam metric. 
\end{abstract}

%\begin{IEEEkeywords}
%Ulam metric, permutation codes, flash memory, 
%error control coding, multipermutations
%\end{IEEEkeywords}

% \fi
\IEEEpeerreviewmaketitle

\section{Introduction}
%A permutation code $C$ is a subset of the symmetric group 
%$\mathbb{S}_n$, equipped with a distance metric.  
The history of permutation codes 
dates as far back as the 1960's and 70's, with Slepian, 
Berger, Blake, and others 
\cite{Berger, Blake, Slepian}. 
However, the application of permutation codes and multipermutation codes 
for use in non-volatile memory storage 
systems such as flash memory has received attention in 
the coding theory literature in recent years \cite{Barg, Jiang, Jiangtwo,  Lim, Vinck, Wadayama}.
One of the main distance metrics in the literature has been the 
Kendall-$\tau$ metric, which is suitable for correction of the 
type of error expected to occur in flash memory devices
\cite{Gad, Hagiwara, Jiang, Jiangtwo, Wang}.
Errors occur in these devices when the electric cell charges 
storing information leak over time or there is an overshoot of charge 
level in the rewriting process.  For relatively small 
leak or overshoot errors the Kendtall-$\tau$ metric is 
appropriate. However, it may not be well-suited for large errors 
within a single cell. 

In 2013, Farnoud et al. proposed permutation codes 
using the Ulam metric 
 \cite{Farnoud}. They showed that the use of the Ulam metric 
 would allow a large leakage or overshoot 
error within a single cell to be viewed as a single error. 
Subsequent papers expounded on the use of Ulam metric 
in multipermutation codes and bounds on the size of permutation codes 
in the Ulam metric \cite{Farnoudtwo, Gologlu}. 
Meanwhile, Buzaglo et al. discovered the existence 
of a perfect permutation code under the cyclic Kendall-$\tau$
metric, and proved the non-existence of perfect permutation 
codes under the Kendall-$\tau$ metric for certain parameters
\cite{Buzaglo}. 
However, the possibility of perfect permutation codes in the 
Ulam metric had not previously 
been considered. Exploring this possibility requires first
understanding the sizes of Ulam permutation spheres, 
of which only limited research exists. Even
less is known about the size of multipermutation Ulam spheres.

In this paper we consider four main questions. 
Their answers are the main contributions of this paper. 
The first question is: How can permutation Ulam sphere 
sizes be calculated? One answer to this question 
is to use Young Tableaux and the RSK-Correspondence 
(Theorem \ref{klambda}). 
The second question is: Do perfect Ulam permutation 
codes exists? The answer to this question is that 
nontrivial perfect Ulam permutation codes do not exist
(Theorem \ref{nonexistence}). 
These two questions are closely related to each other 
since perfect Ulam permutation code sizes are 
characterized by Ulam sphere sizes. 
These main results are summarized in Tables 
\ref{permspheresizes} and \ref{permcodelimits} 
on the following page. 
Notation appearing on the tables is defined in subsequent sections.

The discussion is then extended to multipermutation codes, 
where we consider the third question. 
%The third question has two parts:  
%How are are the Ulam metrics related for permutations and 
%multipermutations, and how can the multipermutation Ulam 
%metric be simplified? 
%Lemmas \ref{n-l} and \ref{translocations} address the third question. 
%Finally, the fourth question also has two parts: 
The third question is: 
How can multipermutation Ulam sphere sizes be calculated? 
Theorem \ref{klambda} and Theorem \ref{ballcalc} 
show how to calculate sphere sizes 
for certain parameters.
Finally, the fourth question is: 
What is the maximum possible Ulam multipermutation code size? 
Lemmas \ref{upperbound}, \ref{perfect bound}, and 
\ref{G-V bound} 
(as well as Lemmas \ref{bound1} and \ref{bound2} for the
special binary case) provide
new upper and lower bounds on the maximal code size. 
These main results are summarized in Tables 
\ref{multipermspheresizes} and \ref{multipermcodesizes}.
Notation appearing on the tables is defined in subsequent sections.

\begin{table}[h!]
\caption{(New Results 1) Permutation Ulam sphere sizes }
\label{permspheresizes}
\begin{small}
\begin{center}
\begin{TAB}(r,3cm,.8cm)[5pt]{|c|c|}{|c|c|c|}
\\
\textbf{Permutation Ulam Sphere Size Formulas} & \textbf{Reference}
\\
$|S(\sigma,t)|\;=\;|S(e,t)|\;=\;\underset{\lambda\in\Lambda}{\sum}(f^{\lambda})^{2}$ & Theorem \ref{klambda} \\
$|S(\sigma,1)|\;=\;1+(n-1)^{2}$ & Proposition \ref{one sphere} \\
\end{TAB}
\end{center} 
\end{small}
\end{table}

\begin{table}[h!]
\begin{small}
\caption{(New Results 2) Theoretical limit on maximum Ulam permutation code size}
\label{permcodelimits}
\begin{center}
\begin{TAB}(r,3cm,.8cm)[5pt]{|c|c|}{|c|c|}
\\
\textbf{Theorem on perfect Ulam permutation codes} & \textbf{Reference}
\\
Nontrivial perfect $t$-error correcting permutation codes do not exist & Theorem \ref{nonexistence} \\
\end{TAB}
\end{center} 
\end{small}
\end{table}

\begin{table}[h!]
\caption{(New Results 3) Multipermutation Ulam Sphere Sizes}
\label{multipermspheresizes}
 \begin{small}
 \begin{center}
\begin{TAB}(r,3cm,.85cm)[3pt]{|c|c|}{|c|c|c|c|c|c|}
\\
\textbf{ Multipermutation Ulam Sphere Size Formulas and Bounds} & \textbf{Reference} \\
$|S(\mathbf{m}_{e}^{r},t)|\;=\;\underset{\lambda\in\Lambda}{\sum(f^{\lambda})(K_{r}^{\lambda})}$ & Theorem \ref{klambda} \\
$|S(\mathbf{m}_{\sigma}^{r},1)|\;=\;1+(n-1)^{2}
-|SD(\mathbf{m}_{\sigma}^{r})|-|AD(\mathbf{m}_{\sigma}^{r})|$ & Theorem \ref{ballcalc} \\
%\textbf{Non-Binary Case: } 
$1+(n-1)(n/r-1)
\;=\;
|S(\mathbf{m}_{e}^r,1)|
\;\le\;
|S(\mathbf{m}_{\sigma}^{r},1)| $ 
& Lemma \ref{esphere} and Theorem \ref{klambda}\\
%\textbf{Binary Case: } 
%$1+(n-1)
%\;=\; 
%|S(\mathbf{m}_e^r,1)| 
%\;\le\;|S(\mathbf{m}_{\sigma}^{r},1)|\;<\; U(r)$ 
%& Lemma \ref{esphere} \\
\textbf{Non-Binary Case: }
$ |S(\mathbf{m}_\sigma^r,1)|
\;\le\;
|S(\mathbf{m}_\omega^r,1)| 
\;=\;
1+(n-1)^{2}-(r-1)n$ 
& Lemma \ref{omegasphere} \\
\textbf{Binary Case: } 
$|S(\mathbf{m}_{\sigma}^{r},1)|\;<\; U(r) $
& Corollary \ref{spherebound}
\end{TAB}
\end{center}
\end{small}
\end{table}

\begin{table}[h!]
\caption{(New Results 4) Theoretical limits on maximum Ulam multipermutation code size}
\label{multipermcodesizes}
\begin{small}
\begin{center}
\begin{TAB}(r,3cm,.85cm)[3pt]{|c|c|c|}{|c|c|c|c|c|c|}
\\
\textbf{Ulam Multipermutation Code Max Size Bounds} & 
\textbf{Value}& 
\textbf{Reference}
\\
1-error correcting code upper bound
& $|C|  \;\;\le\;\; \frac{n!}{(r!)^{n/r}\left(1+(n-1)(n/r-1)\right)}$
& Lemma \ref{upperbound} \\
\textbf{Non-Binary Case: } Perfect 1-error correcting code lower bound
& $\frac{n!}{(r!)^{n/r}  ((1+(n-1)^2) - (r-1)n)} \;\;\le\;\; |C|$
& Lemma \ref{perfect bound} \\
\textbf{Binary Case: } Perfect 1-error correcting code lower bound
& $ \frac{n!}{(r!)^2(U(r))} \;\;\le\;\; |C|$
&Lemma \ref{bound1}\\
\textbf{Non-Binary Case: } MPC$_\circ(n,r,d)$ lower bound
& $ \frac{n!}{(r!)^{n/r} (1 + (n-1)^2 - (r-1)n )^{d-1}  } 
\;\;\le\;\; |C|$
& Lemma \ref{G-V bound} \\
\textbf{Binary Case: } 
MPC$_\circ(n,r,d)$ lower bound
& $\frac{n!}{(r!)^2(U(r))^{d-1}} \;\;\le\;\; |C|$
&Lemma \ref{bound2} \\
\end{TAB}
\end{center} 
\end{small}
\end{table}

The organization is as follows: 
Section \ref{prelims} defines notation and basic concepts 
used throughout the paper.
Sections \ref{sphere size} and \ref{perfect codes} 
focus primarily on permutations, although many 
results apply to multipermutations as well. 
Section \ref{sphere size} introduces 
a method of calculating Ulam sphere sizes using Young tableaux 
and the RSK-Correspondence.  
Section \ref{perfect codes} 
focuses on proving the non-existence 
of nontrivial perfect permutation codes. 

The remaining sections focus on multipermutations. 
Section \ref{upper} discusses how to calculate $r$-regular 
multipermutation Ulam sphere sizes. 
Section \ref{minmax} discusses minimum and maximum 
sphere sizes and provides new upper and lower bounds 
on maximal multipermutation code size. Included in the last
two subsections of Section \ref{minmax} is an explanation of how 
to determine the maximum sphere size in the binary case, 
which presents unique challenges. 
Finally Section \ref{conclusion} gives concluding remarks.

%%%%%%%%%%%%%%%%PRELIMS
\section{Preliminaries and Notation}\label{prelims}
In this paper we utilize the following notation and definitions, 
generally following conventions established in \cite{Farnoud} and \cite{Farnoudtwo}. 
%\textbf{Bold} faced text denotes a definition, while 
%\textit{italicized} text is merely for emphasis. 
Throughout the paper we will assume that $n$ and $r$ are  
positive integers, with $r$ dividing $n$. The symbol 
$[n]$ denotes the set of integers $ \{1, 2, \dots, n\}$. 
The symbol $\mathbb{S}_n$ 
stands for the set of permutations (automorphisms) on $[n]$, i.e., the 
symmetric group of order $n!$.  For a permutation $\sigma \in \mathbb{S}_n,$ we
use the notation $\sigma = [\sigma(1), \sigma(2), \dots, \sigma(n)]$, where 
for all $i \in [n],$ $\sigma(i)$ is the image of $i$ under $\sigma.$  
With some abuse of notation, we may also use $\sigma$ to refer 
to the sequence $(\sigma(1), \sigma(2), \dots, \sigma(n)) \in [n]^n$. 
Given two permutations $\sigma, \pi \in \mathbb{S}_n$, the product 
$\sigma \pi$ is defined by $(\sigma \pi) (i) = \sigma(\pi(i))$. 
In other words, we define multiplication of permutations by composition, e.g.,
$[2,1,5,4,3][5,1,4,2,3] = [3, 2, 4, 1, 5].$
The identity permutation $[1,2, \dots, n] \in \mathbb{S}_n$ is denoted by $e$.

An $r$-{regular multiset} is a multiset such that 
each of its element appears exactly $r$ times 
(i.e., each element is repeated $r$ times).  
%For the remainder of the paper, we assume that 
%Given a postive integer $r$ dividing $n$, 
%we write {\fontfamily {cmss} \selectfont M}$(n,r)$ 
%to denote the $r$-regular multiset whose elements are 
%precisely the set $[n/r]$. 
For example, 
%{\fontfamily {cmss}\selectfont M}$(6,2)$ = 
$\{1,1,2,2,3,3\}$ is a $2$-regular multiset.  
A \textbf{multipermutation} is an ordered tuple whose 
entries exactly correspond to the elements of a multiset, 
and in the instance of an $r$-regular multiset, we call the 
multipermutation an \textbf{$r$-regular multipermutation}.  
For example, $(3, 2, 2, 1, 3, 1)\in [3]^6$ is a $2$-regular multipermutation 
of $\{1,1,2,2,3,3\}$.
Following the work of \cite{Farnoudtwo}, and because 
$r$-regular multipermutations result in the largest 
potential code space \cite{Gad}, in this paper we only consider 
$r$-regular multipermutations. Hence for the remainder of this paper 
``multipermutation" will always refer to an $r$-regular 
multipermutation. 

\begin{definition}[$\mathbf{m}_\sigma^r$]
Given %an $r$-regular multiset
%{\fontfamily {cmss}\selectfont M}$(n,r)$, 
%for each 
$\sigma \in \mathbb{S}_n$ we 
define a corresponding $r$-regular 
multipermutation $\mathbf{m}_\sigma^r$
%{\fontfamily {cmss} \selectfont m}$_\sigma^r$
as follows: 
for all $i \in [n]$ and $j \in [n/r],$ 
\[\mathbf{m}_\sigma^r(i):= j \text{ if and only if } (j-1)r + 1 \le \sigma(i) \le jr,\]
and $\mathbf{m}_\sigma^r :=
(\mathbf{m}_\sigma^r(1), \mathbf{m}_\sigma^r(2), \dots, \mathbf{m}_\sigma^r(n)) \in 
[n/r]^n$.
\end{definition}

As an example of $\mathbf{m}_\sigma^r$, let $n = 6$, 
$r = 2$, and $\sigma = [1,5,2,4,3,6].$ Then 
$\mathbf{m}_\sigma^r = (1,3,1,2,2,3).$ 
Note that this definition differs slightly from the 
correspondence defined in \cite{Farnoudtwo}, which was 
defined in terms of the inverse permutation.  This is so that 
certain properties (Remarks \ref{n-l} and \ref{translocations})
of the Ulam metric for permutations (the case when $r = 1$) will also
hold for general multipermutations. 
Notice that $\mathbf{m}^1_\sigma = (\sigma(1), \dots \sigma(n)) \in [n]^n$,  
so based on our abuse of notation described in the first paragraph of 
this section, we may denote 
$\mathbf{m}^1_\sigma$ simply by $\sigma$. 
In other words, whenever $r=1$, $r$-regular multipermutations 
reduce to permutations, or more accurately their associated 
sequences. 

With the correspondence above, we may define 
an equivalence relation between elements of $\mathbb{S}_n$. 
For permutations $\sigma, \pi \in \mathbb{S}_n$, %and $r$ dividing $n,$
we say that $\sigma \equiv_r \pi$ if and only if 
$\mathbf{m}_\sigma^r = \mathbf{m}_\pi^r$. 
The equivalence class $R_r(\sigma)$ of $\sigma \in \mathbb{S}_n$ 
is defined by $R_r(\sigma) := \{\pi \in \mathbb{S}_n \;:\; \pi \equiv_r \sigma\}.$ 
Note that if $r=1$, then $R_r(\sigma)$ is simply the singleton $\{\sigma\}$. 
For a subset $S \subseteq \mathbb{S}_n,$ define 
$\mathcal{M}_r(S) := \{\mathbf{m}_\sigma^r \; : \; \sigma \in S\},$
i.e. the set of $r$-regular multipermutations corresponding
to elements of $S$. 
When $r=1$, we may identify $\mathcal{M}_r(S)$ simply by $S$.

%Recall that the Ulam metric $\mathrm{d}_\circ(\sigma,\pi)$ 
%between permutations $\sigma,\pi \in \mathbb{S}_n$ was 
%defined in terms of longest common subsequences: 
% $\mathrm{d}_\circ(\sigma, \pi) := n - \ell(\sigma, \pi).$ 
%Recall also that the Ulam distance $\mathrm{d}_\circ(\sigma, \pi)$ 
% between $\sigma, \pi \in \mathbb{S}_n$ is equivalent to the 
% minimum number of translocations needed to transform 
% $\sigma$ into $\pi.$ 
We next define the $r$-regular Ulam distance. 
For the definition, it is first necessary to define 
$\ell(\mathbf{x}, \mathbf{y})$. 
Given sequences $\mathbf{x}, \mathbf{y} \in \mathbb{Z}^n$, 
then $\ell(\mathbf{x}, \mathbf{y})$ denotes 
 the length of the longest common 
subsequence of $\mathbf{x}$ and $\mathbf{y}$ 
(not to be confused with the longest common substring). 
More precisely, $\ell(\mathbf{x}, \mathbf{y})$ is the largest integer 
$k \in \mathbb{Z}_{>0}$ 
such that there exists a sequence $(a_1, a_2, \dots, a_k)$ 
where for all $l \in [k]$, we have $a_l = \mathbf{x}(i_l) = \mathbf{y}(j_l)$ 
with  $1\le i_1 < i_2 < \dots < i_k \le n$ and $1 \le j_1 < j_2 < \dots < j_k\le n$. 
For example, $\ell((3,1,2,1,2,3), (1,1,2,2,3,3)) = 4$, since 
$(1,1,2,3)$ is a common subsequence of both $(3,1,2,1,2,3)$ and 
$(1,1,2,2,3,3)$ and its length is $4$. It does not 
matter that other equally long common subsequences exist (e.g. $(1,2,2,3)$), 
as long as there do not exist any longer common subsequences. 
If $\sigma \in \mathbb{S}_n$, then $\ell(\sigma,e)$ is the 
length of the longest increasing subsequence of $\sigma$, which we denote 
simply by $\ell(\sigma)$. Similarly, for an $r$-regular multipermutation 
$\mathbf{m}_\sigma^r$, we denote the length of the longest non-decreasing 
subsequence $\ell(\mathbf{m}_\sigma^r,\mathbf{m}_e^r)$ of 
$\mathbf{m}_\sigma^r$ simply by $\ell(\mathbf{m}_\sigma^r)$.

%\begin{definition}[$\mathrm{d}_\circ(\sigma,\pi)$, Ulam distance]
%%%Given $n \in \mathbb{Z}_{>0}$ and $\sigma, \pi \in \mathbb{S}_n,$ define
%\[
%\mathrm{d}_\circ(\sigma,\pi) := n - \ell(\sigma,\pi).
%\]
%We call $\mathrm{d}_\circ(\sigma,\pi)$ the \textbf{Ulam distance} between 
%$\sigma$ and $\pi$.  
%\end{definition}

%The $r$-regular Ulam distance for 
% multipermutations is defined in terms of the Ulam 
% distance for permutations. 
 
\begin{definition}[$\mathrm{d}_\circ(\mathbf{m}_\sigma^r,\mathbf{m}_\pi^r)$, $r$-regular Ulam distance]
Let %$n, r \in \mathbb{Z}_{>0}$, $r | n$, and 
$\mathbf{m}_\sigma^r, \mathbf{m}_\pi^r \in 
\mathcal{M}_r(\mathbb{S}_n)$. Define 
\[
\mathrm{d}_\circ(\mathbf{m}^r_\sigma, \mathbf{m}^r_\pi) := 
\underset{\sigma' \in R_r(\sigma), \pi' \in R_r(\pi) }{\min}
\mathrm{d}_\circ(\sigma', \pi'),
\]
where 
$
\mathrm{d}_\circ(\sigma,\pi) := n - \ell(\sigma,\pi).
$
We call $\mathrm{d}_\circ(\mathbf{m}_\sigma^r,\mathbf{m}_\pi^r)$ the 
\textbf{$r$-regular Ulam distance} 
between $\mathbf{m}_\sigma^r$ and $\mathbf{m}_\pi^r$. 
In the case when $r=1$, we may simply say the 
\textbf{Ulam distance} between $\sigma$ and $\pi$ and 
use the notation $\mathrm{d}_\circ(\sigma,\pi)$.

\end{definition}

The definition of $r$-regular Ulam distance above 
follows the convention of \cite{Farnoudtwo}, 
defining the distance in terms of equivalence 
classes comprised of permutations, although our notation differs. 
However, it is convenient to think of the distance instead in 
terms of the multipermutations themselves.  
A simple argument shows that the $r$-regular Ulam 
distance between multipermutations 
$\mathbf{m}_\sigma^r$ and $\mathbf{m}_\pi^r$ 
is equal to $n$ minus the length of their longest 
common subsequence. The details of the argument 
can be found in the appendices. 

\begin{remark}\label{n-l}
Let %$n,r \in \mathbb{Z}_{>0}$, $r \vert n$, and 
 $\mathbf{m}_\sigma^r, \mathbf{m}_\pi^r \in 
 \mathcal{M}_r(\mathbb{S}_n)$. Then 
\begin{align*}
\mathrm{d}_\circ(\mathbf{m}_\sigma^r,\mathbf{m}_\pi^r) =
n - \ell(\mathbf{m}_\sigma^r, \mathbf{m}_\pi^r).
\end{align*} 
\end{remark}

Viewed this way, 
it is easily verified that the $r$-regular Ulam distance 
$\mathrm{d}_\circ(\mathbf{m}_\sigma^r,\mathbf{m}_\pi^r)$ is a proper 
metric between the multipermutations 
$\mathbf{m}_\sigma^r$ and $\mathbf{m}_\pi^r$. 
Additionally, it is known that in the permutation case, the 
case when $r=1$, that the Ulam distance can be characterized 
in terms of a specific type of permutation known as 
 translocations \cite{Farnoud, Gologlu}. 
 We can show a similar relationship for multipermutations. 
 We define translocations 
 below and then give the relationship between the 
 Ulam distance and translocations. 

\begin{definition}[$\phi(i,j)$, translocation]
Given %$n \in \mathbb{Z}_{>0}$ and 
distinct $i,j\in[n],$ 
define $\phi(i,j) \in \mathbb{S}_n$  as follows: 

%If $i < j,$ then $\phi(i,j) :=  [1,2,\dots i-1,i+1,i+2, \dots, j, i, j+1,\dots, n],$ 
%and if $i > j,$ then 
%$\phi(i,j) := [1,2,\dots j-1, i, j, j+1, \dots, i-1, i+1, \dots n].$

\[
\phi (i,j) :=
\begin{cases}
  [1,2,\dots i-1,i+1,i+2, \dots, j, i, j+1,\dots, n]  
~ ~ ~ \text{if } i < j  \\
  [1,2,\dots j-1, i, j, j+1, \dots, i-1, i+1, \dots, n]   
\hfill  \text{if } i > j
  \end{cases}
  \]

If $i = j,$ then define $\phi(i,j) := e$. 
We refer to $\phi(i,j)$ as a \textbf{translocation}, and if we do not 
specify the indexes $i$ and $j$ we may denote a translocation 
simply by $\phi$.  
\end{definition}

Intuitively, a translocation is the permutation that results in a 
delete/insertion operation.  
More specifically, given $\sigma \in \mathbb{S}_n$ 
and the translocation $\phi(i,j) \in \mathbb{S}_n,$ 
the product $\sigma \phi(i,j)$ is the result of 
deleting $\sigma(i)$ from the $i$th position of $\sigma$,
then shifting all positions between the $i$th and $j$th position  
by one (left if $i < j$ and right if $i > j$), and finally 
reinserting $\sigma(i)$ into the new $j$th position.  
The top half of Figure \ref{translocation} illustrates the permutation 
$\sigma = [6,2,8,5,4,1,3,9,7]$ 
(or its related $3$-regular multipermutation 
$\mathbf{m}_\sigma^3 = (2,1,3,2,2,1,1,3,3)$) 
represented physically by relative cell charge levels 
and the effect of multiplying $\sigma$ (or $\mathbf{m}_\sigma^3$) on the right  
by the translocation $\phi(1,9)$. 
The bottom half of Figure \ref{translocation} illustrates the same 
$\sigma$ (or $\mathbf{m}_\sigma^3$) and the effect of $\phi(7,4)$.  
Notice that multiplying by $\phi(1,9)$ corresponds 
to the error that occurs when the highest (1st) ranked cell 
suffers charge leakage that results in it being the lowest (9th) ranked cell. 
Multiplying by $\phi(7,4)$ corresponds to the error that occurs when 
the 7th highest cell is overfilled so that it is the 4th highest cell.

It is well-known that $\mathrm{d}_\circ(\sigma,\pi)$ equals 
the minimum number of translocations needed to transform 
$\sigma$ into $\pi$ \cite{Farnoud, Gologlu}.  That is,  
$\mathrm{d}_\circ(\sigma,\pi) = \min \{k \in \mathbb{Z}_{\ge 0} \; : \;
\text{ there exists } \phi_1, \phi_2, \dots, \phi_k$ such that 
$\sigma \phi_1 \phi_2 \cdots \phi_k = \pi\}$.  
By applying Remark \ref{n-l}, it is also a simple matter 
to prove that an analogous relationship holds for 
the $r$-regular Ulam distance. First, it is necessary to 
define multiplication between multipermutations and 
permutations. 

 The following definition is our own. We define 
the product $\mathbf{m}_\sigma^r \cdot \pi$ 
as $\mathbf{m}_\sigma^r \cdot \pi := 
\mathbf{m}^r_{\sigma \pi}$.  
Technically speaking, this can be seen a right group action of the set 
$\mathbb{S}_n$ of permutations on the set $\mathcal{M}_r(\mathbb{S}_n)$.
Since it is possible for different permutations to 
correspond to the same multipermutation, 
we should clarify that 
 $\mathbf{m}_\sigma^r = \mathbf{m}_\tau^r$ 
implies 
$\mathbf{m}_{\sigma \pi}^r = \mathbf{m}_{\tau \pi}^r.$ 
Indeed this is true because if 
 $\mathbf{m}_\sigma^r = \mathbf{m}_\tau^r$ then 
for all $i \in [n]$ we have 
$\mathbf{m}_\sigma^r(i) = \mathbf{m}_\tau^r(i)$, which implies 
for $j := \mathbf{m}_\sigma^r(i) $ that 
$(j-1)r + 1 \le \sigma(i) \le jr$ and $(j-1)r +1 \le \tau(i) \le jr$. 
This in turn implies that 
$(j-1)r + 1 \le \sigma\pi(\pi(i)) \le jr$ and 
$(j-1)r + 1 \le \tau\pi(\pi(i)) \le jr$, which means 
$\mathbf{m}_{\sigma\pi}(\pi(i)) = \mathbf{m}_{\tau\pi}(\pi(i))$. 
Intuitively speaking, the same corresponding elements of the 
sequences $\sigma$ and $\tau$ still correspond (with a different index) 
after being multiplied on the right by $\pi.$ 
Hence $\mathbf{m}_{\sigma\pi}^r = \mathbf{m}_{\tau\pi}^r$, 
or by our notation $\mathbf{m}_\sigma^r \cdot \pi 
= \mathbf{m}_\tau^r\cdot \pi.$

 If two multipermutations $\mathbf{m}_\sigma^r$ 
and $\mathbf{m}_\pi^r$ have a common subsequence of 
length $k$, then $\mathbf{m}_\sigma^r$ can 
be transformed into $\mathbf{m}_\pi^r$ with $n-k$ (but 
no fewer) delete/insert operations.  As with permutations, 
delete/insert operations correspond to applying 
(multiplying on the right) a translocation.  
Hence by Remark \ref{n-l}  we can 
state the following remark about the $r$-regular 
Ulam distance.  The details of the proof
can be found in the appendices. 

\begin{remark}\label{translocations}
Let %$n,r \in \mathbb{Z}_{>0}$, $r \vert n$, and 
 $\mathbf{m}_\sigma^r, \mathbf{m}_\pi^r \in 
 \mathcal{M}_r(\mathbb{S}_n)$. Then 
\[
\mathrm{d}_\circ(\mathbf{m}_\sigma^r,\mathbf{m}_\pi^r) = 
\min \{ k \in \mathbb{Z}_{\ge0} \;:\; \text{ there exists } \; (\phi_1, \phi_2, \dots, \phi_k)\; 
\text{ such that }
\; \mathbf{m}_\sigma^r \cdot \phi_1  \phi_2 \cdots \phi_k = \mathbf{m}_\pi^r \}.
\]
\end{remark}

%A quick example makes the explanation above clearer: 
% if $n = 6$ and $r=2$ with permutations 
% $\sigma := [1,2,3,4,5,6]$ and $\pi := [2,1,4,3,5,6]$, 
% then $\mathbf{m}_\sigma^r = \mathbf{m}_\pi^r 
% = (1,1,2,2,3,3).$ 
% Here, if we take any permutation 
% $\tau \in \mathbb{S}_6$ and multiply both 
% $\sigma$ and $\pi$ on the right side by $\tau$, then 
% in the resulting permutations $\sigma\pi$ and 
% $\tau\pi,$ we will still have the same 
% corresponding values. 

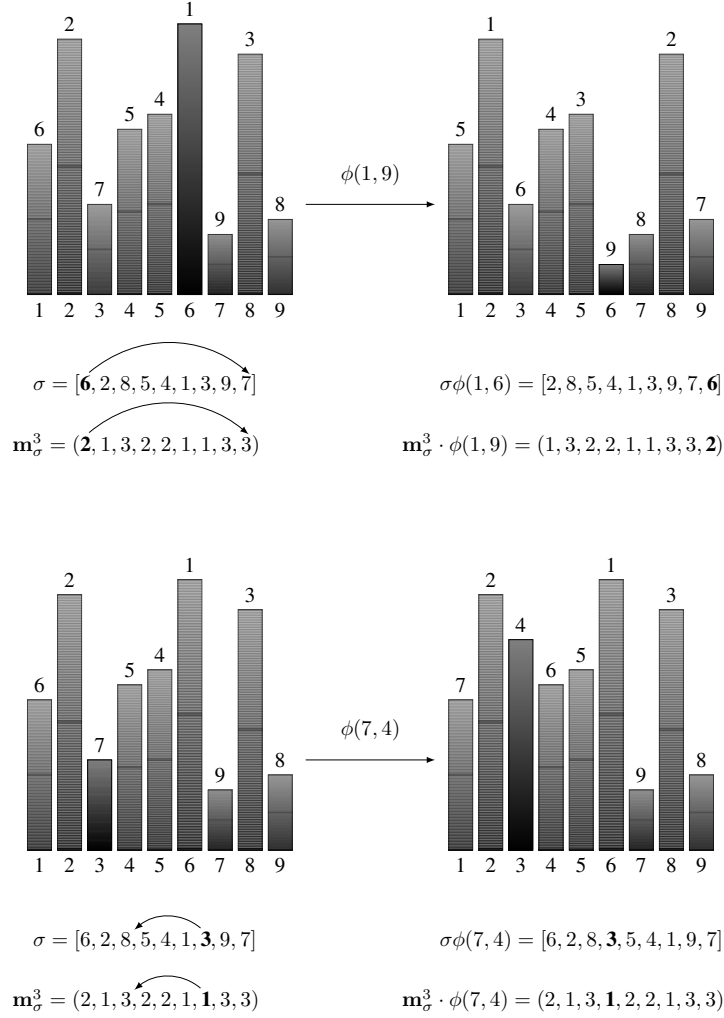
\begin{figure}[h]\caption{Translocation illustration}\label{translocation}
\[
\scalebox{.8}{

	\begin{tikzpicture}[node distance=1cm]
\node [draw, shade, bottom color = black, opacity = 0.6,
          minimum width=.4cm, minimum height=2.5cm] (a1) at (0,-4.75){};%%%
\node [draw, shade,  bottom color = black, opacity = 0.6,
           minimum width=.4cm, minimum height=4.25cm] (b1) at (.5,-(3 + 7/8){};%
\node [draw, shade,  bottom color = black, opacity = 0.6,
           minimum width=.4cm, minimum height=1.5cm] (c1) at (1,-5.25) {};
\node [draw, shade,  bottom color = black, opacity = 0.6,
           minimum width=.4cm, minimum height=2.75cm] (d1) at (1.5,-4-5/8){};%%%
\node [draw, shade,  bottom color = black, opacity = 0.6,
            minimum width=.4cm, minimum height=3.00cm] (e1) at (2,-4.5){};%%%
\node [draw, shade,  bottom color = black, opacity = 1,
           minimum width=.4cm, minimum height=4.5cm] (f1) at (2.5,-3.75) {};
\node [draw, shade,  bottom color = black, opacity = 0.6,
           minimum width=.4cm, minimum height=1cm] (g1) at (3, -5.5) {};
\node [draw, shade, bottom color = black, opacity = 0.6,
           minimum width=.4cm, minimum height=4cm] (h1) at (3.5, -4) {};
\node [draw, shade,  bottom color = black, opacity = 0.6,
           minimum width=.4cm, minimum height=1.25cm] (i1) at (4, -(5+3/8) {};

%%original rank labels on top 1/4 increments
\node (Ha1) at (0,-3.25){6};%%%
\node (Hb1) at (.5,-1.5) {2};%
\node (Hc1) at (1, -4.25) {7};
\node (Hd1) at (1.5, -3) {5};
\node (He1) at (2, -2.75){4};
\node (Hf1) at (2.5, -1.25) {1};
\node (Hg1) at (3, -4.75) {9};
\node (Hh1) at (3.5, -1.75) {3};
\node (Hi1) at (4, -4.5) {8};

%%original position labels on bottom
\node (La1) at (0,-6.25){1};
\node (Lb1) at (.5, -6.25){2};
\node (Lc1) at (1,-6.25) {3};
\node (Ld1) at (1.5, -6.25) {4};
\node (Le1) at (2, -6.25) {5};
\node (Lf1) at (2.5, -6.25) {6};
\node (Lg1) at (3, -6.25) {7};
\node (Lh1) at (3.5, -6.25) {8};
\node (Li1) at (4,-6.25) {9};

%%sigma and msigma 
\node (sigma) at (1.75, -7.5) {$\sigma = [\textbf{6},2,8,5,4,1,3,9,7]$};
\node (msigma) at (1.6,-8.5) 
                      {$\mathbf{m}_\sigma^3 = (\textbf{2},1,3,2,2,1,1,3,3)$};

%%%arrow in between
\node (f1) at (4.3,-4.5) {};
\node (f2) at (6.7, -4.5) {};
\node (f3) at (5.5, -4){\textcolor{black}{$\phi(1,9)$}};
\draw [-latex] (f1) to (f2);

%%right side bars
\node [draw, shade,  bottom color = black, opacity = 0.6,
           minimum width=.4cm, minimum height=2.5cm] (a1) at (7,-4.75){};%%%
\node [draw, shade,  bottom color = black, opacity = 0.6,
           minimum width=.4cm, minimum height=4.25cm] (b1) at (7.5,-(3 + 7/8){};%
\node [draw, shade,  bottom color = black, opacity = 0.6,
           minimum width=.4cm, minimum height=1.5cm] (c1) at (8,-5.25) {};
\node [draw, shade,  bottom color = black, opacity = 0.6,
          minimum width=.4cm, minimum height=2.75cm] (d1) at (8.5,-4-5/8){};%%%
\node [draw, shade,  bottom color = black, opacity = 0.6,
          minimum width=.4cm, minimum height=3.00cm] (e1) at (9,-4.5){};%%%
\node [draw, shade, bottom color = black, opacity = 1,  
           minimum width=.4cm, minimum height=.5cm] (e1) at (9.5,-5.75) {};
\node [draw, shade,  bottom color = black, opacity = 0.6,
           minimum width=.4cm, minimum height=1cm] (g1) at (10, -5.5) {};
\node [draw, shade,  bottom color = black, opacity = 0.6,
           minimum width=.4cm, minimum height=4cm] (h1) at (10.5, -4) {};
\node [draw, shade,  bottom color = black, opacity = 0.6,
           minimum width=.4cm, minimum height=1.25cm] (i1) at (11, -(5+3/8) {};

%%original rank labels on top 1/4 increments
\node (Ha1) at (7,-3.25){5};%%%
\node (Hb1) at (7.5,-1.5) {1};%
\node (Hc1) at (8, -4.25) {6};
\node (Hd1) at (8.5, -3) {4};
\node (He1) at (9, -2.75){3};
\node (Hf1) at (9.5, -5.25) {9};
\node (Hg1) at (10, -4.75) {8};
\node (Hh1) at (10.5, -1.75) {2};
\node (Hi1) at (11, -4.5) {7};

%%original position labels on bottom
\node (La1) at (7,-6.25){1};
\node (Lb1) at (7.5, -6.25){2};
\node (Lc1) at (8,-6.25) {3};
\node (Ld1) at (8.5, -6.25) {4};
\node (Le1) at (9, -6.25) {5};
\node (Lf1) at (9.5, -6.25) {6};
\node (Lg1) at (10, -6.25) {7};
\node (Lh1) at (10.5, -6.25) {8};
\node (Li1) at (11,-6.25) {9};

%%%%%%
\node (sigmaphi) at (9, -7.5) 
{$\sigma \textcolor{black}{\phi(1,6)} = [2,8,5,4,1,3,9,7,\textbf{6}] $};
\node (msigmaphi) at (8.7,-8.5) {$\mathbf{m}_\sigma^3 \cdot
\textcolor{black}{\phi(1,9)} 
 = (1,3,2,2,1,1,3,3,\textbf{2})$};

\node (A) at (.7,-7.4) {};
\node (B) at (3.62, -7.43) {}; 
\draw [-latex, color = black] (A) to [out=40,in=140] (B);

\node (C) at (.7,-8.4) {};
\node (D) at (3.6, -8.43) {}; 
\draw [-latex, color = black] (C) to [out=40,in=140] (D);

\end{tikzpicture}

}
\]
\vspace{.5cm}
\[
\scalebox{0.8}{

	\begin{tikzpicture}[node distance=1cm]
\node [draw, shade, bottom color = black, opacity = 0.6,
          minimum width=.4cm, minimum height=2.5cm] (a1) at (0,-4.75){};%%%
\node [draw, shade,  bottom color = black, opacity = 0.6,
           minimum width=.4cm, minimum height=4.25cm] (b1) at (.5,-(3 + 7/8){};%
\node [draw, shade,  bottom color = black, opacity = 1,
           minimum width=.4cm, minimum height=1.5cm] (c1) at (1,-5.25) {};
\node [draw, shade,  bottom color = black, opacity = 0.6,
           minimum width=.4cm, minimum height=2.75cm] (d1) at (1.5,-4-5/8){};%%%
\node [draw, shade,  bottom color = black, opacity = 0.6,
            minimum width=.4cm, minimum height=3.00cm] (e1) at (2,-4.5){};%%%
\node [draw, shade,  bottom color = black, opacity = 0.6,
           minimum width=.4cm, minimum height=4.5cm] (f1) at (2.5,-3.75) {};
\node [draw, shade,  bottom color = black, opacity = 0.6,
           minimum width=.4cm, minimum height=1cm] (g1) at (3, -5.5) {};
\node [draw, shade, bottom color = black, opacity = 0.6,
           minimum width=.4cm, minimum height=4cm] (h1) at (3.5, -4) {};
\node [draw, shade,  bottom color = black, opacity = 0.6,
           minimum width=.4cm, minimum height=1.25cm] (i1) at (4, -(5+3/8) {};

%%original rank labels on top 1/4 increments
\node (Ha1) at (0,-3.25){6};%%%
\node (Hb1) at (.5,-1.5) {2};%
\node (Hc1) at (1, -4.25) {7};
\node (Hd1) at (1.5, -3) {5};
\node (He1) at (2, -2.75){4};
\node (Hf1) at (2.5, -1.25) {1};
\node (Hg1) at (3, -4.75) {9};
\node (Hh1) at (3.5, -1.75) {3};
\node (Hi1) at (4, -4.5) {8};

%%original position labels on bottom
\node (La1) at (0,-6.25){1};
\node (Lb1) at (.5, -6.25){2};
\node (Lc1) at (1,-6.25) {3};
\node (Ld1) at (1.5, -6.25) {4};
\node (Le1) at (2, -6.25) {5};
\node (Lf1) at (2.5, -6.25) {6};
\node (Lg1) at (3, -6.25) {7};
\node (Lh1) at (3.5, -6.25) {8};
\node (Li1) at (4,-6.25) {9};

%%sigma and msigma 
\node (sigma) at (1.75, -7.5) {$\sigma = [6,2,8,5,4,1,\textbf{3},9,7]$};
\node (msigma) at (1.6,-8.5) 
                      {$\mathbf{m}_\sigma^3 = (2,1,3,2,2,1,\textbf{1},3,3)$};

%%%arrow in between
\node (f1) at (4.3,-4.5) {};
\node (f2) at (6.7, -4.5) {};
\node (f3) at (5.5, -4){\textcolor{black}{$\phi(7,4)$}};
\draw [-latex] (f1) to (f2);

%%right side bars
\node [draw, shade,  bottom color = black, opacity = 0.6,
           minimum width=.4cm, minimum height=2.5cm] (a1) at (7,-4.75){};%%%
\node [draw, shade,  bottom color = black, opacity = 0.6,
           minimum width=.4cm, minimum height=4.25cm] (b1) at (7.5,-(3 + 7/8){};%
\node [draw, shade,  bottom color = black, opacity = 1,
           minimum width=.4cm, minimum height=3.5cm] (c1) at (8,-4.25) {};
\node [draw, shade,  bottom color = black, opacity = 0.6,
          minimum width=.4cm, minimum height=2.75cm] (d1) at (8.5,-4-5/8){};%%%
\node [draw, shade,  bottom color = black, opacity = 0.6,
          minimum width=.4cm, minimum height=3.00cm] (e1) at (9,-4.5){};%%%
\node [draw, shade,  bottom color = black, opacity = 0.6,
           minimum width=.4cm, minimum height=4.5cm] (f1) at (9.5,-3.75) {};
\node [draw, shade,  bottom color = black, opacity = 0.6,
           minimum width=.4cm, minimum height=1cm] (g1) at (10, -5.5) {};
\node [draw, shade,  bottom color = black, opacity = 0.6,
           minimum width=.4cm, minimum height=4cm] (h1) at (10.5, -4) {};
\node [draw, shade,  bottom color = black, opacity = 0.6,
           minimum width=.4cm, minimum height=1.25cm] (i1) at (11, -(5+3/8) {};

%%original rank labels on top 1/4 increments
\node (Ha1) at (7,-3.25){7};%%%
\node (Hb1) at (7.5,-1.5) {2};%
\node (Hc1) at (8, -2.25) {4};
\node (Hd1) at (8.5, -3) {6};
\node (He1) at (9, -2.75){5};
\node (Hf1) at (9.5, -1.25) {1};
\node (Hg1) at (10, -4.75) {9};
\node (Hh1) at (10.5, -1.75) {3};
\node (Hi1) at (11, -4.5) {8};

%%original position labels on bottom
\node (La1) at (7,-6.25){1};
\node (Lb1) at (7.5, -6.25){2};
\node (Lc1) at (8,-6.25) {3};
\node (Ld1) at (8.5, -6.25) {4};
\node (Le1) at (9, -6.25) {5};
\node (Lf1) at (9.5, -6.25) {6};
\node (Lg1) at (10, -6.25) {7};
\node (Lh1) at (10.5, -6.25) {8};
\node (Li1) at (11,-6.25) {9};

%%%%%%
\node (sigmaphi) at (9, -7.5) 
{$\sigma \textcolor{black}{\phi(7,4)} = [6,2,8,\textbf{3},5,4,1,9,7] $};
\node (msigmaphi) at (8.7,-8.5) {$\mathbf{m}_\sigma^3 \cdot
\textcolor{black}{\phi(7,4)} 
 = (2,1,3,\textbf{1},2,2,1,3,3)$};

\node (A) at (2.8,-7.4) {};
\node (B) at (1.45, -7.43) {}; 
\draw [-latex, color = black] (A) to [out=140,in=40] (B);

\node (C) at (2.8,-8.4) {};
\node (D) at (1.45, -8.43) {}; 
\draw [-latex, color = black] (C) to [out=140,in=40] (D);

\end{tikzpicture}   
	
	}
\]
\end{figure}

We now define the notions of a
multipermutation code and 
an $r$-regular Ulam sphere.

\begin{definition}
[$r$-regular multipermutation code, MPC$(n,r)$, MPC$(n,r,d)$]
Recall that $n,r \in \mathbb{Z}_{>0}$ with $r \vert n$.
 An \textbf{$r$-regular multipermutation code} 
 (or simply a \textbf{multipermutation code}) is a subset 
$C \subseteq \mathcal{M}_r(\mathbb{S}_n)$.
Such a code is denoted by MPC$(n,r)$, and 
we say that $C$ is an MPC$(n,r)$.  
If $C$ is an 
MPC$(n,r)$ such that 
$\underset{
\mathbf{m}_\sigma^r, \mathbf{m}_\pi^r \in C, 
\mathbf{m}_\sigma^r \ne \mathbf{m}_\pi^r}{\min} 
\mathrm{d}_\circ(\mathbf{m}_\sigma^r,\mathbf{m}_\pi^r)=d$, 
then we call $C$ an MPC$_\circ(n,r,d)$.
We refer to any $1$-regular multipermutation code 
simply as a \textbf{permutation code}. 
\end{definition}

Our definition of multipermutation codes is in terms 
of multipermutations, i.e. ordered tuples, rather than in 
terms of permutations, i.e. automorphisms. 
This differs slightly from \cite{Farnoudtwo}, 
where multipermutation codes were defined as subsets of 
$\mathbb{S}_n$ with the requirement that the entire 
equivalence class of each element in a code 
was a subset of the code. 
Next, we define $r$-regular Ulam spheres.

\begin{definition}[$S(\mathbf{m}_\sigma^r,t)$, 
$r$-regular multipermutation Ulam sphere] 
Let %$n,r \in \mathbb{Z}_{>0}$, $r \vert n$, 
$t \in \mathbb{Z}_{\ge0}$, 
and $\mathbf{m}_\sigma^r \in \mathcal{M}_r(\mathbb{S}_n).$ Define 
\[
S(\mathbf{m}_\sigma^r, t) := \{
\mathbf{m}_\pi^r \in \mathcal{M}_r(\mathbb{S}_n) 
\;:\; \mathrm{d}_\circ(\mathbf{m}_\sigma^r,\mathbf{m}_\pi^r) \le t\}
\]

We call $S(\mathbf{m}_\sigma^r, t)$ 
the \textbf{$r$-regular multipermutation Ulam sphere}, 
(or simply the \textbf{multipermutation Ulam sphere})
centered at $\mathbf{m}_\sigma^r$ of radius $t$.
We refer to any $1$-regular multipermutation Ulam sphere
as a \textbf{permutation Ulam sphere} and use the simplified notation 
$S(\sigma,t)$ instead of 
$S(\mathbf{m}_\sigma^r,t)$.
\end{definition}

By Remark \ref{n-l}, 
$S(\mathbf{m}_\sigma^r,t) = 
\{\mathbf{m}_\pi^r
 \in \mathcal{M}_r(\mathbb{S}_n) \;:\; 
 n - \ell(\mathbf{m}_\sigma^r, \mathbf{m}_\pi^r)  \le t\}$. 
%It should be noted, however, that the notation $\mathbf{m}_\pi^r$
% is a bit misleading because given 
%  $\mathbf{m}_\pi^r \in \mathcal{M}(\mathbb{S}_n),$ 
% we cannot uniquely determine $\pi.$  
The $r$-regular Ulam sphere definition can also be viewed 
in terms of translocations. Remark \ref{translocations} implies 
that $S(\mathbf{m}_\sigma^r, t)$ is equivalent to 
 $\{\mathbf{m}_\pi^r 
 \in \mathcal{M}_r(\mathbb{S}_n) \;:\; 
\text{ there exists } k \in \{0,1,\dots,t\} \text{ and } (\phi_1, \dots, \phi_k) \text{ such that } 
\mathbf{m}_\sigma^r \cdot \phi_1 \cdots \phi_k  = \mathbf{m}_\pi^r\}$.  
This is the set of all multipermutations reachable by applying $t$ 
translocations to the center multipermutation $\mathbf{m}_\sigma^r$.

It is well-known that an 
MPC$_\circ(n,r,d)$ code is $t$-error 
correcting if and only if $d \ge 2t+1$ 
\cite{Farnoudtwo}. 
This is because if the distance between two codewords is 
greater or equal to $2t + 1$, then after $t$ or fewer errors 
(multiplication by $t$ or fewer translocations), the resulting 
multipermutation remains closer to the original multipermutation 
than any other multipermutation.
We finish this section by defining perfect 
$t$-error correcting codes. 

\begin{definition}[perfect code] 
Let $C \subseteq  \mathcal{M}_r(\mathbb{S}_n)$ be an MPC$(n,r)$. 
Then $C$ is a perfect $t$-error 
correcting code if and only if for all 
$\mathbf{m}_\sigma^r \in \mathcal{M}_r(\mathbb{S}_n),$ there exists a unique 
$\mathbf{m}_c^r \in \mathcal{M}_r (C)$ 
such that $\mathbf{m}_\sigma^r \in S(\mathbf{m}_c^r, t).$ 
We call such $C$ a \textbf{perfect $t$-error correcting } 
$\mathsf{MPC}(n,r)$, or simply a 
\textbf{perfect code} if the context is clear.  
A permutation code that is perfect is called 
a \textbf{perfect permutation code}. 
\end{definition} 

A perfect MPC$(n,r)$ partitions  
$\mathcal{M}_r(\mathbb{S}_n)$. 
This means the spheres centered at codewords 
fill the space without overlapping. 
A perfect code $C \subseteq \mathcal{M}_r(\mathbb{S}_n)$ is said to be 
\textbf{trivial}
if either (1) $C = \mathcal{M}_r(\mathbb{S}_n)$ 
(occurring when $t = 0$); or 
(2) $\vert C \vert = 1$ (occurring when $t = n-r$).

%%%%%%%%%%%%%%  3RD SECTION?, SPHERE CALCULATION
\section{Permutation Ulam Sphere Size}\label{sphere size}

%\subsection{Sphere with Radius of One}
This section focuses on the first of four main questions: 
how can we calculate the sizes of permutation Ulam spheres? 
The answer to this question, in the form of Theorem \ref{klambda}, 
 is the first main result of this paper. 
The theorem is actually stated in terms of multipermutations, 
making it also a partial answer the the third main question of this paper concerning 
how to calculate multipermutation Ulam sphere sizes. 
However, unlike permutations, in the case of multipermutations 
sphere sizes may depend upon the choice of center, 
limiting the applicability of the theorem for multipermutation Ulam spheres. 
The proof of the theorem is provided after a necessary 
lemma is recalled and notation used in the theorem is clarified. 
 
 \begin{theorem} \label{klambda}
Let %$r,n \in \mathbb{Z}_{>0},$ $r \vert n$, 
$t \in \{0,1,2\dots, n-r\}$, and 
$\Lambda := \{\lambda \vdash n \; : \; \lambda_1 \ge n-t\}.$
%Let $\Lambda$ be the 
%set of all partitions $\lambda \vdash n$ such that 
%$\lambda_1 \ge n-t.$
Then 
\begin{align}\label{neweqn1}
|S (\mathbf{m}_e^r, t)| = \underset{\lambda \in \Lambda}{\sum} 
(f^{\lambda})(K^\lambda_r).
\end{align}
\end{theorem}
% \begin{corollary} \label{flambda applied}
%Let $n\in \mathbb{Z}_{>0}$, 
%$t \in \{0,1,\dots,n-1\}$,  and $\Lambda := \{\lambda \vdash n \; : \; 
%\lambda_1 \ge n -t \}$. 
%\[
%\text{Then }
%\vert S (e, t)\vert = \underset{\lambda \in \Lambda}{\sum} (f^{\lambda})^2.
%\]
%\end{corollary}

Although this section is primarily concerned with permutation 
Ulam sphere sizes, many of the results hold for multipermutation 
Ulam spheres as well, and lemmas and propositions in this section 
are stated with as much generality as possible. 
In the case of permutation codes, 
perfect codes and sphere sizes are related as follows: 
a perfect $t$-error correcting permutation code $C \subseteq 
\mathbb{S}_n$, 
if it exists, will have cardinality $|C| = {n!}/{S(c,t)},$ where $c \in C$.  
Hence one of the first questions that may be considered in exploring 
the possibility of a perfect code (the second of four main questions) 
is the feasibility of a code of such size.  
As noted in \cite{Farnoud}, for any $\sigma \in \mathbb{S}_n$,
we have $\vert S(\sigma,t)\vert = \vert S(e,t)\vert$. 
Hence calculation of permutation Ulam sphere sizes can be 
reduced to the case when the identity is the center. 

One way to calculate permutation Ulam sphere sizes 
centered at $e$ is to use 
Young tableaux and the RSK-Correspondence.
It is first necessary to introduce some basic notation and 
definitions regarding Young diagrams and Young tableaux. 
 Additional information on the subject 
can be found in \cite{Schensted}, \cite{Fulton} and  \cite{Stanley}.

A \textbf{Young diagram} is a left-justified collection of cells with a
(weakly) decreasing number of cells in each row below.  Listing the 
number of cells in each row gives a partition 
$\lambda = (\lambda_1, \lambda_2, \dots, \lambda_k)$ 
 of $n,$ where $n$ is the total
number of cells in the Young diagram.  
The notation $\lambda \vdash n$ is used to mean $\lambda$ is a partition of $n$.
Because the partition $\lambda \vdash n$ 
defines a unique Young diagram and vice versa, 
 a Young diagram may be referred to  by its associated 
 partition $\lambda \vdash n$.  
 For example, the partition $\lambda := (4,3,3,2) \vdash 12$ has the 
 corresponding Young diagram pictured on the left side of 
 Figure \ref{Young}.

 \begin{figure}[h]
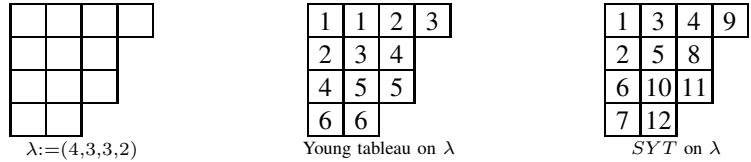
\caption{Young diagram and SYT}\label{Young}
 \[
 \underset{\lambda := (4,3,3,2)}
{\begin{Young}
& & & \cr
   & & \cr 
  & & \cr 
  & \cr
\end{Young}}
\hspace{2cm}
 \underset{\text{Young tableau on $\lambda$}}
{\begin{Young}
1&1 &2 &3 \cr
 2& 3&4 \cr 
 4 &5 &5 \cr 
 6&6 \cr
\end{Young}}
\hspace{2cm}
\underset{SYT \text{ on } \lambda}{
\begin{Young}
1 & 3 & 4 & 9 \cr
2 & 5 & 8 \cr
6 &10 & 11\cr
7 &12 \cr
\end{Young}}
\]
\end{figure} 

A \textbf{Young tableau} is a filling of a Young diagram 
$\lambda \vdash n$  with the following 
two qualities: (1) cell values are weakly increasing across each row; and 
(2) cell values are strictly increasing down each column. 
One possible Younge tableau is pictured in the center of 
Figure \ref{Young}.
A \textbf{standard Young tableau}, abbreviated by $SYT$, is 
a filling of a Young diagram $\lambda \vdash n$  with the following 
three qualities: (1) cell values are strictly increasing across each row; 
(2) cell values are strictly increasing down each column; and 
(3) each of the integers $1$ through $n$ appears exactly once. 
One possible $SYT$ on $\lambda := (4,3,3,1)$ is pictured in the right 
side of Figure \ref{Young}.

%In the statement of the lemma as well as subsequent appearances, 
%$P$ and $Q$ are the $SYT$ associated with $\sigma$ by the 
%RSK-correspondence. 

Among other things, the famous RSK-correspondence 
(\cite{Fulton, Stanley}) 
provides a bijection between $r$-regular multipermutations
 $\mathbf{m}_\sigma^r$ and ordered pairs $(P,Q)$ on the same 
 Young diagram $\lambda \vdash n,$ where 
 $P$ is a Young tableau whose members come from 
$\mathbf{m}_\sigma^r$ and $Q$ is a SYT.
The next lemma, a stronger form 
of which appears in \cite{Fulton}, is 
an application of the RSK-correspondence.

\begin{lemma}\label{length lemma}
Let %$n \in \mathbb{Z}_{>0},$ 
$\mathbf{m}_\sigma^r \in 
\mathcal{M}_r(\mathbb{S}_n)$ and 
let $P$ and $Q$, both on $\lambda \vdash n$, 
be the pair of Young tableaux associated with 
$\mathbf{m}_\sigma^r$ by the RSK-correspondence. 
Then %the number of columns in $P$ is equal to 
\[
\lambda_1 \;=\;
\ell(\mathbf{m}_\sigma^r).
\]
%the length of the longest non-decreasing subsequence of 
%$\mathbf{m}_\sigma^r$.
\end{lemma}

In words, the above lemma says that $\lambda_1$, the number of 
columns in the $P$ (or equivalently $Q$) associated with 
$\mathbf{m}_\sigma^r$ by the RSK-correspondence, 
is equal to $\ell(\mathbf{m}_\sigma^r)$, the length of the longest 
non-decreasing subsequence of $\mathbf{m}_\sigma^r$.
The lemma implies that 
for all $k \in [n]$, the size of the set 
$\{\mathbf{m}_\sigma^r \in \mathcal{M}_r(\mathbb{S}_n) \;\vert\; 
\ell(\mathbf{m}^r_\sigma) = k\}$ is equal to the 
sum of the number of ordered pairs 
$(P,Q)$ on each Young diagram $\lambda \vdash n$. 
%such that 
%$\lambda_1 = k,$ where $P$ is a tableaux whose members come from 
%$\mathbf{m}_\sigma^r$ and $Q$ is a standard 
%tableaux.  
Following conventional notation (\cite{Fulton, Stanley}), 
$f^\lambda$ denotes the number of $SYT$ on 
$\lambda \vdash n$.
We denote by $K^\lambda_r$ 
(our own notation) the 
number of Young tableaux on $\lambda \vdash n$ such that 
each $i \in [n/r]$ appears exactly $r$ times. 
We are now able to prove
Theorem \ref{klambda}, which states the relationship between 
$|S(\mathbf{m}_e^r, t)|$, $f^{\lambda}$, and $K^\lambda_r$. \\~\\
%As mentioned earlier, this provides an answer to the first main 
%question of this paper and a partial answer to the third main 
%question of this paper concerning the calculation of multipermutation 
%Ulam spheres. However, as we will see later, multipermutation 
%sphere sizes differ depending upon the choice of center, 
%which limits the applicability of the theorem in the multipermutation case. \\~\\
%\begin{proof}
\textit{Proof of Theorem \ref{klambda}:} \\
Assume %$r,n \in \mathbb{Z}_{>0},$ $r \vert n$, and 
$t \in \{0,1, \dots, n-1\}$, and
let $\Lambda := \{\lambda \vdash n \; : \; \lambda_1 \ge n-t\}.$
Furthermore, let 
$\Lambda^{(l)} := \{\lambda \vdash n \; : \; \lambda_1 = l\},$
the set of all partitions of $n$ 
having exactly $l$ columns.  
By the RSK-Correspondence and Lemma \ref{length lemma}, 
there is a bijection between the set 
$\{\mathbf{m}_\sigma^r \; : \; \ell(\mathbf{m}_\sigma^r) = l\}$ and 
the set of ordered pairs $(P,Q)$ where both $P$ and $Q$ have 
exactly $l$ columns.  
This implies that 
$\#\{\mathbf{m}_\sigma^r \; : \; \ell(\mathbf{m}_\sigma^r) = l\} 
\;=\; 
\underset{\lambda \in \Lambda^{(l)}}{\sum} (f^{\lambda})(K^\lambda_r)$
(here $\#A$ is an alternate notation for the cardinality of a set that 
we prefer for conditionally defined sets).
By Remark \ref{n-l}, 
$|S (\mathbf{m}_e^r, t)| 
\;=\; 
\#\{ \mathbf{m}_\sigma \; :  \; \mathrm{d}_\circ(\mathbf{m}_e^r,\mathbf{m}_\sigma^r) \le t\}
\;=\; \#\{ \mathbf{m}_\sigma \; : \; \ell(\mathbf{m}_\sigma^r)  \ge n-t \} $. 
Hence it follows that 
$|S (\mathbf{m}_e^r, t)| 
\;=\;
 \underset{\lambda \in \Lambda}{\sum} (f^{\lambda})(K^\lambda_r).$
\hfill $\square$ \\
%\end{proof}

Because $K_1^\lambda$ is equivalent to $f^\lambda$ by definition, 
in the case of permutation Ulam spheres, equation (\ref{neweqn1})
simplifies to 
\begin{align}
\label{neweqn2}
\vert S (e, t)\vert = \underset{\lambda \in \Lambda}{\sum} (f^{\lambda})^2.
\end{align}
%Corollary \ref{flambda applied} is an immediate consequence of 
%Theorem \ref{klambda}. Corollary \ref{flambda applied} can be 
%applied with the well-known 
In both equation (\ref{neweqn1}) and (\ref{neweqn2}), 
the famous hook length formula, due to 
Frame, Robinson, and Thrall \cite{Frame, Fulton}, 
provides a way to calculate $f^\lambda$.  
Within the hook length formula,  the notation $(i,j) \in \lambda$
 is used to refer to the cell in the $i$th 
row and $j$th column of a Young diagram 
$\lambda \vdash n$.  The notation $h(i,j)$ denotes the
 \textbf{hook length} of $(i,j) \in \lambda$, i.e., 
the number of boxes below or to the right of $(i,j)$, including the box $(i,j)$ 
itself.  More formally, 
$h(i,j) := \# (\{ (i,j^*) \in \lambda \;:\; j^* \ge j\} \cup \{ (i^*,j)
 \in \lambda \;:\; i^* \ge i\})$.
 The \textbf{hook-length formula} is as follows:

\[
f^\lambda = \frac{n!}{\underset{(i,j)\in\lambda} \Pi h(i,j)}.
\]

Applying the hook length formula to Theorem \ref{klambda}, 
we may explicitly calculate Ulam permutation sphere sizes, 
as demonstrated in the following propositions.  
These propositions will be 
useful later to show the nonexistence of nontrivial 
$t$-error correcting perfect permutation 
codes for $t \in \{1,2,3\}$. 
Proposition %\ref{zero sphere} and 
 \ref{one sphere} 
 is stated in terms of general multipermutation Ulam spheres.

\begin{proposition}\label{one sphere}
%Let $n,r \in \mathbb{Z}_{>0},$ and $r \vert n$. Then 
$|S (\mathbf{m}_e^r, 1)|  = 1 + (n-1)(n/r-1).$
\end{proposition}

%%%%%%%%%%%%%%%%%%%%%%%%%%%%%%%%
\begin{proof}
First note that $|S (\mathbf{m}_e^r, 0)| =  |\{\mathbf{m}_e^r\}|= 1.$ 
  There is only one possible partition 
 $\lambda \vdash n$ such that $\lambda_1 = n-1,$
  namely $\lambda' := (n-1,1),$ with its Young diagram pictured below. 
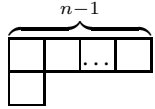
\[
\overbrace{
\begin{Young}
& & \dots &  \cr
 \cr
\end{Young}
}^{n-1}
\]

Therefore by Theorem \ref{klambda}, 
 $|S (\mathbf{m}_e^r, 1)| = 1 + (f^{\lambda'})(K^{\lambda'}_r).$  
Applying the hook length formula, we obtain 
$f^{\lambda'} = n-1.$  The value $K^{\lambda'}_r$ is characterized 
by possible fillings of row $2$ with the stipulation that each $i \in [n/r]$
 must appear 
exactly $r$ times in the diagram.  In this case, since there is only a single 
box in row $2$, the possible fillings are $i \in [n/r-1],$ 
each of which yields a unique Young tableau 
of the desired type.  Hence $K^{\lambda'}_r = n/r -1$, 
which implies that $|S (\mathbf{m}_e^r, 1)| = 1 + (n-1)(n/r-1).$
\end{proof}

%\begin{prop}\label{zero sphere}
%Let $n\in \mathbb{Z}_{>0}$ and $\sigma \in \mathbb{S}_n.$ 
%Then $\vert S(\sigma, 0)\vert = 1$. 
%\end{prop}
%\begin{proof}
%Although this is an obvious fact, we wish to consider why it is true 
%from the perspective of Lemma \ref{flambda applied}. 
%Note first that there is only one partition $\lambda \vdash n$ 
%such that $\lambda_1 = n$, namely 
%$\lambda := (n)$ with the associated Young diagram below.  
%\[
%\overbrace{
%\begin{Young}
%& & \dots &  \cr
%\end{Young}
%}^{n}
%\]
%It is clear that there is only one possible Young tableau on $\lambda$ 
%so that $(f^{\lambda}) = 1,$ and thus by Lemma \ref{flambda applied} 
%$\vert S(\sigma, 0) \vert$ = 1.
%\end{proof}

%\begin{prop} \label{one sphere}
%Let $n\in \mathbb{Z}_{>0}$ and 
%$\sigma \in \mathbb{S}_n$.  Then $\vert S(\sigma,1)\vert = 1+ (n-1)^2$.
%\end{proposition}
%\end{prop}
%\begin{proof}
%  There is only one possible partition 
% $\lambda \vdash n$ such that $\lambda_1 = n-1,$
%  namely $\lambda := (n-1,1),$ with its Young diagram pictured below. 
%\[
%\overbrace{
%\begin{Young}
%& & \dots &  \cr
% \cr
%\end{Young}
%}^{n-1}
%\]
%Therefore by Lemma \ref{flambda applied} and Proposition \ref{zero sphere}, 
%$\vert S(\sigma, 1)\vert = 1 + (f^{\lambda})^2$.
%Applying the hook length formula, we obtain 
%$(f^{\lambda})^2 = \big( {n!}/({(n)(n-2)!}) \big)^2 = (n-1)^2$, 
%which implies that $\vert S(\sigma,1) \vert = 1 + (n-1)^2.$  
%\end{proof}

%%%%%%%%%%%%%%INSERT FROM ISITA BEGIN
Setting $r=1$, Proposition \ref{one sphere} 
implies that $|S(e,1)| = 1 + (n-1)^2$. 
The next two propositions continue the same vein of reasoning, 
but focus on permutation Ulam spheres. 
Such individual cases could be considered indefinitely. 
In fact, a recurrence equation providing an alternative 
method of calculating permutation Ulam sphere sizes for reasonably 
small radii is also known \cite{Kobayashi}.
However, the following two propositions are the last instances of significance in this 
paper as their results will be necessary to prove the 
second main result of this paper.

\begin{prop}\label{two sphere}
Let %$n \in \mathbb{Z}_{>0}$, 
$n>3$ and $\sigma \in \mathbb{S}_n$.
Then 
\[\vert S(\sigma,2)\vert =
 1+ (n-1)^2 + \left( \frac{(n)(n-3)}{2}\right)^2 + \left(\frac{(n-1)(n-2)}{2}\right)^2.\]
 \end{prop}

\begin{proof}
Assume $n>3$ and $\sigma \in \mathbb{S}_n$. 
Note first that $\vert S(\sigma,2)\vert = \vert S(\sigma,1)\vert + 
\# \{\pi \in \mathbb{S}_n \;: \; \ell(\pi) = n-2 \}.$ 
The only partitions $\lambda \vdash n$ such that 
$\lambda_1 = n - 2$ are 
$\lambda^{(1)} := (n-2, 1, 1)$ and $\lambda^{(2)} := (n-2, 2)$, 
with their respective Young diagrams pictured below.  

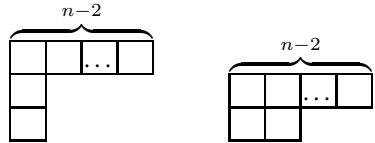
\[
\overbrace{
\begin{Young}
& & \dots &  \cr
 \cr
 \cr
\end{Young}
}^{n-2}
\; \; \; \; \; \; \; \; \; \; 
\overbrace{
\begin{Young}
& & \dots &  \cr
&  \cr
\end{Young}
}^{n-2}
\]
Using the hook length formula, 
 $f^{\lambda^{(1)}}$ and
 $f^{\lambda^{(2)}}$ may be calculated to yield: 
 $f^{\lambda^{(1)}} =  ((n)(n-3))/{2}$ 
 and  $f^{\lambda^{(2)}} = ((n-1)(n-2))/{2}$.
 Following the same reasoning as in 
Proposition \ref{one sphere} yields the desired result.  
\end{proof}

%%FINAL DRAFT EDIT

\begin{lemma}\label{3 sphere}
Let %$n\in \mathbb{Z}_{>0}$, 
$n>5$ and $\sigma \in \mathbb{S}_n$ 
  Then 
\begin{eqnarray*}
\vert S(\sigma,3)\vert &=&
 1+ (n-1)^2 + \left( \frac{(n)(n-3)}{2}\right)^2 
 + \left(\frac{(n-1)(n-2)}{2}\right)^2 \\
  &+& \left( \frac{(n)(n-1)(n-5)}{6}\right)^2
   + \left( \frac{(n)(n-2)(n-4)}{3}\right)^2
 + \left( \frac{(n-1)(n-2)(n-3)}{6}\right)^2. 
 \end{eqnarray*}
 
\end{lemma}
\begin{proof}
The proof is essentially the same as the proof
for Proposition \ref{two sphere}.  
In this case $\# \{\pi \in \mathbb{S}_n \; : \; \ell(\pi) = n-3\}$ 
can be calculated by considering the partitions $\lambda^{(1)} := (n-3,3),$ 
$\lambda^{(2)} := (n-3,2,1)$, and $\lambda^{(3)} := (n-3, 1, 1, 1),$ 
the only Young diagrams having $n-3$ columns.  
These Young diagrams are pictured below. 
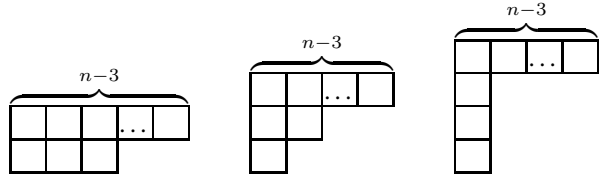
\[
\overbrace{
\begin{Young}
& & & \dots &  \cr
& & \cr
\end{Young}
}^{n-3}
\; \; \; \; \; \; \; \; 
\overbrace{
\begin{Young}
& & \dots &  \cr
&  \cr
\cr
\end{Young}
}^{n-3}
\; \; \; \; \; \; \; \; 
\overbrace{
\begin{Young}
& & \dots &  \cr
  \cr
 \cr
 \cr
\end{Young}
}^{n-3}
\]
Applying the hook length formula to
$\lambda^{(1)}$, 
$\lambda^{(2)}$, and $\lambda^{(3)}$ and adding 
the value from Proposition \ref{two sphere} yields the result. 

\end{proof}
%%%%%%%%%%%%%%%%%%%%%INSERT ISITA END 

%%%%%%%%%%%%%%%%% 4TH SECTION, SINGLE-ERROR
\section{Nonexistence of Nontrivial Perfect Ulam Permutation Codes}\label{perfect codes}
The previous section demonstrated how to calculate 
permutation Ulam sphere sizes. 
In this section, again focusing on permutation codes,  
we utilize sphere size calculations to 
prove the following theorem, establishing 
a theoretical limit on the maximum 
size of Ulam permutation codes. 
This is the second of four main contributions of this paper. 
The proof of the theorem can be found at the 
end of the current section. 

\begin{theorem}\label{nonexistence}
There do not exist any nontrivial perfect permutation codes in the Ulam metric.
\end{theorem}

In 2013, Farnoud et. al (\cite{Farnoud}) proved the following upper bound on 
the size of an Ulam permutation code $C \subseteq \mathbb{S}_n$ 
with minimum Ulam distance $d$ 
(i.e. $C$ is an MPC$(n,1,d))$. 
\begin{equation}\label{eqn1} 
|C|  \le (n - d + 1)!
\end{equation}

Hence one strategy to prove the non-existence of 
perfect permutation codes is to show that the size of a perfect 
code must necessarily be larger than the upper-bound 
given above.  Note that for equation (\ref{eqn1}) to make 
sense, $d$ must be less than 
or equal to $n-1$.  This is 
always true since the maximum Ulam distance 
between any two permutations in $\mathbb{S}_n$ is $n-1$, 
achieved when permutations are in reverse 
order of each other (e.g., $\mathrm{d}_\circ(e,[n,n-1,...,1]) = n-1).$

\begin{lemma}\label{1-error code}
There do not exist any (nontrivial) single-error correcting perfect permutation codes.
\end{lemma}
\begin{proof}
Assume that $C \subseteq \mathbb{S}_n$ is a perfect single-error 
correcting permutation code.
Recall that $C$ is trivial code if either $C = \mathbb{S}_n$ or if 
$\vert C \vert = 1$.
If $n \le 2,$ then for all $\sigma,\pi \in \mathbb{S}_n,$ we have 
$\pi \in S(\sigma, 1),$ which implies that $C$ is a trivial code.  
Thus we may assume that $n > 2.$ 

We proceed by contradiction. 
Since $C$ is a perfect single-error correcting
permutation code, 
$C$ is an MPC$(n, 1, d)$ with $3 \le d \le n-1$ 
and $|C| \;=\; {n!}/{\vert S(\sigma,1)\vert} \;=\; 
{n!}/({1+(n-1)^2})$ by Proposition \ref{one sphere}.
However, inequality (\ref{eqn1}) implies that the code size 
$|C| \le (n -2)!.$  Hence, it suffices to show that 
$|C| \;=\; {n!}/({1+(n-1)^2})\; > \;(n-2)!,$ which is true if and only if 
%$(n)(n-1) > 1 + (n-1)^2 \iff (n)(n-1) - (n-1)^2> 1 \iff
%(n-n+1)(n-1) > 1 \iff 
$n > 2.$    
\end{proof}

Similar arguments may also be applied to show that no 
nontrivial perfect $t$-error correcting codes 
exist for $t \in \{2,3\}$.  This is the subject of the next two lemmas. 
The remaining cases, when $t > 3$, are treated toward the end of this 
section.

\begin{lemma}\label{2-error code}
There do not exist any (nontrivial) perfect 2-error correcting permutation codes. 
\end{lemma}
\begin{proof}
Assume that $C$ is a perfect 2-error correcting permutation code. 
Similarly to the proof of Lemma \ref{1-error code}, if $n \le 3,$ then $C$ is a trivial code 
consisting of a single element, so we may assume $n >3$.  
Again we proceed by contradiction.  

Since $C \subseteq \mathbb{S}_n$ is a perfect 2-error correcting code, then 
$C$ is an MPC$(n, 1, d)$ code with $5 \le d \le n-1$ and 
Proposition \ref{two sphere} implies 
\[
|C| \;\;=\;\; \frac{n!}{\vert S(\sigma, 2)\vert} 
\;\;=\;\; \frac{n!}{1+ (n-1)^2 + \left( \frac{(n)(n-3)}{2}\right)^2 
+ \left(\frac{(n-1)(n-2)}{2}\right)^2}.
\]
By Inequality (\ref{eqn1}), $|C| \le (n-4)!$, so it suffices to prove 
that 
\[
%f(n) \;\; :=\;\; 
\frac{n!}{1+ (n-1)^2 
+ \left( \frac{(n)(n-3)}{2}\right)^2 
+ \left(\frac{(n-1)(n-2)}{2}\right)^2}
- (n-4)! \;\;>\;\; 0,
\]  
which is easily shown by elementary methods to be true for $n > 3$.
%Here $f'(n) \;=\; 2n^3 - 9n^2 + 9n -1,$ and $f''(n) \;=\; 6n^2 - 18n + 9,$ which is 
%positive for all $n \;>\; ({3+\sqrt3})/{2} \;\approx\; 2.37.$  Both $f'(4)$ and $f(4)$ are 
%strictly greater than $0$, which in turn implies that for all integer values of $n > 3,$ 
%it must be true that $f(n) > 0.$  
\end{proof}

\begin{lemma}\label{3-error code}
There do not exist any (nontrivial) perfect 3-error correcting codes.  
\end{lemma}
\begin{proofoutline}
Assume that $C\subseteq  \mathbb{S}_n$ is a perfect 3-error
correcting code. 
Similarly to the proof of Lemmas 
\ref{1-error code} and \ref{2-error code}, if 
$n \le 7$, then $C$ is a trivial code, so we may 
assume that $n > 7$. 
%Then for all distinct $c_1, c_2 \in C,$ it must be the 
%case that $\mathrm{d}_\circ(c_1,c_2) \ge 7.$ However, if $n \le 7$, 
%then for all $\sigma, \pi \in \mathbb{S}_n$, we have 
%$\mathrm{d}_\circ(\pi,\sigma) \le 6.$  Hence since $C$ was 
%assumed to be nontrivial we may assume that $n > 7.$ 
The remainder of the proof follows the
 same reasoning as the proof for Lemma
\ref{2-error code}, utilizing the sphere size calculated in 
Proposition \ref{3 sphere}. 
\hfill $\square$
\end{proofoutline}

For small values of $t$, explicit sphere calculations 
work well for showing the non-existence of 
nontrivial perfect $t$-error correcting codes. 
%\textbf{IN FACT... REFERENCE}
However, for each radius $t$, the size of the sphere 
$S(e,t)$ is equal to $|S(e,t-1)| +
\#\{\pi \in \mathbb{S}_n \;:\; \ell(\pi) = n-t\}$. 
This means each sphere size 
calculation of radius $t$ requires calculation of sphere sizes 
for radii from $0$ through $t-1$.  Hence 
such explicit calculations are impractical for 
large values of $t$.  For values 
of $t>3,$ another method can be used to show that 
nontrivial perfect codes do not exist. The next lemma 
provides a sufficient condition to 
conclude that perfect codes do not exist.  
In the proof of the lemma, 
the notation $n\choose{t}$ denotes 
the usual combinatorial choice function.
% i.e. ${n}\choose{r}$ $:= {n!}/({(n-t)!t!}).$

%%%%%%%%%%%%%%%%% 5TH SECTION, MULTIPLER-ERROR

\begin{lemma}\label{perfect inequality}
%Let $t \in \{0,1,\dots,n/2-1\}$.
Let $t$ be a nonnegative integer such that $n \ge 2t$. 
 If the following inequality holds, then no nontrivial perfect 
 $t$-error correcting permutation codes exist in $\mathbb{S}_n:$ 
 
\begin{equation}\label{eqn2}
 F(n,t) := \frac{\left( (n-t)! \right)^2 t!}{n!(n-2t)!} > 1.
\end{equation}

 We call the above inequality the \textbf{overlapping condition}.  
\end{lemma}
\begin{proof}
%Assume $t \in \{0,1,\dots,n/2-1\}$.
Assume that $t$ is a nonnegative integer such that $n \ge 2t$. 
We proceed by contrapositive.  
Suppose $C \subset \mathbb{S}_n$ is a 
nontrivial perfect $t$-error correcting permutation code. 
We want to show that $F(n,t) \le 1.$
Since $C$ is a perfect  code, we know it is also an MPC$(n, 1, d)$ 
code with $2t+1 \;\le\; d$ and  %\le n-1$ and
 by inequality (\ref{eqn1}), $|C| \le (n-2t)!$.
 At the same time, for any $\sigma \in \mathbb{S}_n,$ 
 we have $\vert S(\sigma,t)\vert = \vert S(e,t)\vert,$ 
 which is less than or equal to ${n \choose n-t} (n!)/(n-t)!,$ 
 since any permutation $\pi \in S(e,t)$ can be obtained by 
 first choosing $n-r$ elements of $e$ to be in increasing order, 
 and then arranging the remaining $t$ elements into $\pi.$  
 Of course this method will generally result in double counting some 
 permutations in $S(e,t),$ hence the inequality.  
 Now 
 \[\vert S(\sigma,t)\vert \;\;\le\;\; {n \choose n-t} \frac{n!}{(n-t)!}
\text{ implies that } 
  \frac{(n-t)!}{\binom{n}{t}} \;\;\le\;\; \frac{n!}{\vert S(\sigma,t)\vert} 
  \;\;=\;\; |C| \;\;\le\;\; (n-2t)!.
  \]
  Moreover, $(n-t)!/\binom{n}{t} \le (n-2t)!$ if and only if 
  $F(n,t) \le 1$.
\end{proof}

%%%%%%%%%%%INSERT ISITA 3 BEGIN
Notice that the overlapping condition is never satisfied for $t = 1.$ 
However, the following proposition will imply that as 
long as $t>1,$ then the overlapping condition 
may be satisfied for sufficiently large $n.$  

\begin{proposition}\label{limit}
Let $t$ be a nonnegative integer such that $n\ge2t$.
% \in \{0,1,\dots,n/2-1\}$.  
Then 
$\underset{n\rightarrow \infty}{\lim} %\frac{\big((n-r)!\big)^2 r!}{n!(n-2r)!} 
F(n,t) = t!. $
\end{proposition}
\begin{proof}
Assume $t$ is a nonnegative integer such that 
$n \ge 2t$.  Then 
\begin{eqnarray*}
\underset{n\rightarrow \infty}{\lim} %\frac{((n-r)!)^2 r!}{n!(n-2r)!} 
F(n,t)
&=& \underset{n\rightarrow \infty}{\lim} 
    \frac{(n-t)(n-t-1)\cdots(n-2t+1)(n-2t)!(n-t)!t!}
           {(n)(n-1)\cdots(n-t+1)(n-t)!(n-2t)!} \\~\\
   &=& \underset{n\rightarrow \infty}{\lim} 
      \frac{(n-t)(n-t-1)\cdots(n-2t+1)t!}{(n)(n-1)\cdots(n-t+1)} \\~\\
    &=&  \underset{n\rightarrow \infty}{\lim} 
       \frac{(n^{t-1})t!}{n^{t-1}}
      = t!
\end{eqnarray*}
\end{proof}

The proposition above means that for any nonnegative integer $t$
less than or equal to $n/2$, 
there is some value $k$ such that for all values of $n$ 
larger than $k,$ there does not exist a perfect 
$t$-error correcting code.  The question remains 
of how large the value of $k$ must be before 
it is guaranteed that perfect $t$-error correcting codes do not exist.

\begin{table}[h]\caption{Non-feasibility of perfect $t$-error correcting codes}
\label{nonfeasible}
\begin{center}
\begin{tabular}{|c|c||c|c|}
\hline 
$t$ &  $\min n$ satisfying (\ref{eqn2}) &
 $t$ & $\min n$ satisfying (\ref{eqn2})\tabularnewline
\hline
\hline 
$1$ &  N/A & $6$ &$13$  \tabularnewline
\hline 
$2$ &  $8$& $7$ & $14$  \tabularnewline
\hline 
$3$ &  $8$& $8$ & $16$\tabularnewline
\hline 
$4$  & $10$& $9$ & $18$ \tabularnewline
\hline 
$5$  & $11$& $10$ & $20$ \tabularnewline
\hline
\end{tabular}	
\end{center}
\end{table}

Table \ref{nonfeasible} compares positive integer values $t$ 
versus $\min \{n \in \mathbb{Z}_{>0} \;: \; F(n,t) > 1 \}.$
Values were determined via numerical computer calculation.  
The table suggests that for $t >6,$ 
the minimum value of $n$ satisfying the 
overlapping condition is $n = 2t.$  If what the table appears 
to suggest is true, then in view of Proposition 
\ref{limit}, we may rule out perfect $t$-correcting 
codes for any $t >6.$  The next lemma 
formalizes what is implied in the table  
%%%%%%%%%%%%%%%INSERT ISITA 3 END
by providing parameters for which the 
overlapping condition is always satisfied. 
In combination with Lemma \ref{perfect inequality}, 
the implication is that nontrivial perfect permutation codes 
do not exist for these parameters. The remaining 
cases are also easily dealt with. 

%%%%%%%%%%
\begin{lemma}\label{big lemma}
Let $t$ be an integer greater than $6$.  Then $n \ge 2t$ implies that
the overlapping condition is satisfied.
%$\frac{((n-r)!)^2 r!}{n!(n-2r)!} > 1.$
\end{lemma}

\begin{proof}
Assume $t$ is an integer greater than $6$.
We begin the proof of the lemma by showing that if $n = 2t,$ 
then the desired inequality holds. 
We assume that $n = 2t$ and proceed by induction on $t.$ 

\setstretch{1}{
For the base case, let $t = 7.$  Then $n = 14,$ and  
%$\frac{((n-r)!)^2 r!}{n!(n-2r)!} 
$F(n,t) \;=\; ((7!)^3)/(14!) \;\approx\; 1.46 \;>\; 1.$
As the induction hypothesis, suppose it is true that 
$F(2t,t) \;=\; ((t!)^3)/((tk)!) \;>\; 1.$ 
We wish to show that the following inequality holds: 
%$\frac{((n-r)!)^2 r!}{n!(n-2r)!} 
\[F(2(t+1),t+1) \;\;\;=\;\;\; \frac{((t+1)!)^3}{(2(t+1))!} \;\;\;>\;\;\; 1.\]
\setstretch{1.112}{

}
Here 
\[
\frac{((t+1)!)^3}{(2(t+1))!} \;\;\;=\;\;\; \left(\frac{(t!)^3}{(2t)!}\right) 
\left(\frac{(t+1)^3}{(2t+1)(2t+2)}\right).
\]
 By our induction hypothesis,  
the first term of the right hand side, 
$(t!)^3/(2t)!$, 
 is greater than $1,$ so it 
suffices to show that $(t+1)^3/(2t+1)(2t+2)) \;\ge\; 1.$ 
Note here that 
\[
\frac{(t+1)^3}{(2t+1)(2t+2)} \;\;\;>\;\;\;
\frac{(t+1)^3}{(2t+2)(2t+2)} \;\;\;=\;\;\; 
\frac{(t+1)^3}{4(t+1)^2} \;\;\;=\;\;\;
\frac{t+1}{4},
\] which is greater than $1$ whenever $t > 3.$  
Of course $t > 6$ by assumption, so the desired conclusion follows.

Thus far we have technically only proven that %$\frac{((n-r)!)^2 r!}{n!(n-2r)!} 
$F(n,t) > 1$
whenever $n = 2t.$  However, it is a simple matter to show that the same 
is true whenever $n > 2t$ as well.
We begin by supposing that $F(n,t)  > 1.$  
Then 
\[
F(n+1,t) 
\;\;\;=\;\;\; \frac{((n+1-t)!)^2 t!}{(n+1)!(n+1-2t)!} 
\;\;\;=\;\;\; F(n,t)
\cdot \frac{(n+1-t)^2}{(n+1)(n+1-2t)}
\] 
is necessarily 
greater than $1$ 
whenever $((n+1-t)^2)/((n+1)(n+1-2t)) \;\ge\; 1,$ which is true 
for all values of $n$ and $t.$  
}
\end{proof}
%%%%%%%%%%%

Lemma \ref{big lemma} required that $n \ge 2t$. 
However, if $n < 2t,$ then it is impossible 
for a nontrivial perfect $t$-error correcting 
permutation code to exist.  
In fact, we may say something even stronger. 

\begin{remark}\label{nlt2t}
If $t \in \mathbb{Z}_{>0}$ such that $n \le 2t+1$, then it is impossible 
for a nontrivial perfect $t$-error correcting permutation code to exist. 
\end{remark}
To understand why Remark \ref{nlt2t} is true, consider
 two permutations 
within $\mathbb{S}_n$ of maximal Ulam distance apart. 
 The most obvious 
example of which would be the identity element $e$ and the 
only-decreasing permutation $\omega^* := [n, n-1, ..., 1].$    
Notice that $S(e,t) = \{ \pi \in \mathbb{S}_n \;: \; \ell(\pi) \ge n - t\},$ which 
means that every permutation whose longest increasing subsequence 
is at least $n-t$ is in the sphere centered at $e.$  Meanwhile, there is at least 
one permutation $\sigma \in \mathbb{S}_n$ such that 
$\ell(\sigma) = 1 + t$ and $\sigma \in S(\omega^*,t),$ since we may 
apply successive translocations to $\omega^*$ in such a way that 
the longest increasing subsequence is increased with each translocation.  
As long as $n  \le 2t + 1,$ then $n-t \le t + 1 = 1 + t$, implying that $\ell(\sigma) = 1+t 
\ge n-t$, which implies that 
$\sigma \in S(e,t)\cap S(\omega^*,t).$  Therefore the only perfect code 
possible when $n \le 2t + 1$ is a single element code, i.e. a trivial code. 
Consolidating all previous results, we are now able to prove 
Theorem \ref{nonexistence}. \\

%\begin{proof}
%\textcolor{red}{rewrite proof}
\textit{Proof of Theorem \ref{nonexistence}:} 
First, by Lemmas \ref{1-error code}, \ref{2-error code}, and \ref{3-error code}, 
there do not exist any nontrivial perfect $t$-error correcting 
permutation codes for $t \in \{1,2,3\}$.  
 Next note that
 $F(n,r)$ increases as $n$ increases,
 and thus by numerical results (see Table \ref{nonfeasible}), 
for all $t \in \{4,5,6\}$ the overlapping condition is satisfied
 whenever $n \ge 2t + 2$.  Therefore 
 by Lemma \ref{perfect inequality}, and Remark \ref{nlt2t}, 
 there are no nontrivial perfect $t$-error correcting permutation codes 
 for $t \in \{4,5,6\}$.
 Finally, by Lemmas \ref{perfect inequality}, 
 \ref{big lemma}, and Remark \ref{nlt2t}, 
 there are no nontrivial 
 perfect $r$-error correcting permutation codes for $t > 6$.
\hfill $\square$
% \end{proof}

%%%%%%%%%%%%%%%%%%%%%%%%%%%%%%%%%%%%%%%%%%%%%%%%%%%%%
%%%%%%%%%%%%%%%%%%%%%%%%%%%%%%%%%%%%%%%%%%%%%%%%%%%%%
\begin{comment}
\section{Extension to Multipermutations}\label{extension}

We now extend the discussion from permutation 
codes in the Ulam metric to 
multipermutation codes in the Ulam metric. Specifically, 
we extend the discussion to $r$-regular multipermutations 
defined and studied in \cite{Farnoudtwo}. 
The use of multipermutations has the advantage of higher potential code rates. 
The use of multipermutations also naturally reduces the relative change 
in rank between codewords, which may decrease rewriting costs,
a major concern in flash memory (\cite{Gad}, \cite{Jiangtwo} ).
First, some basic notation and definitions for multipermutations are necessary. 
Unless otherwise stated, definitions are 
based on conventions established in \cite{Farnoudtwo}. 

%%%%%%%%%%
%%%%%%%%%%%%%removed defs here 
%%%%%%%%%%
Notice that the $r$-regular Ulam distance is 
defined over equivalence classes.  
That is, the definition is the minimum Ulam distance among all 
members of $R_r(\sigma)$ and $R_r(\pi)$.  
We will discuss this distance in more detail 
in the following section.  
We finish this section by defining what 
a multipermutation code is.  

%%%%%%%%%%%%%%%%%%%%%%%%%%%%%%%%%%%%%%%%%%%%%%%%%%%%%
%%%%%%%%%%%%%%%%%%%%%%%%%%%%%%%%%%%%%%%%%%%%%%%%%%%%%
\end{comment}

%%%%%%%%%%%%%%%%UPPER BOUND ON SPHERE SIZE
\section{Multipermutation Ulam Sphere Size and Duplication Sets}\label{upper}
Thus far we have focused primarily on permutations, but 
we wish to extend the discussion to multipermutations. 
With both permutations and multipermutations, the number of possible 
messages is limited by the number of distinguishable relative 
rankings in the physical scheme.  However, multipermutations 
may significantly increase the total possible messages 
compared to ordinary permutations, as observed in \cite{Farnoudtwo}.
For example, if only $k$ different charge levels are utilized at 
a given time, then permutations of length $k$ can be stored.  
Hence, in $r$ blocks of length $k,$ one may store $(k!)^r$ 
potential messages.  On the other hand, if one uses $r$-regular 
multipermutations in the same set of blocks, then 
$(kr)!/(r!)^k$ potential messages are possible.

The $r$-regular multipermutation Ulam sphere sizes play an 
important role in understanding the
potential code size for MPC$_\circ(n,r,d)$'s. 
For example, the well-known sphere-packing bounds 
and Gilbert-Varshamov type bounds rely on 
calculating, or at least bounding sphere sizes.  
In this section we analyze how to calculate 
 $r$-regular multipermutation Ulam sphere sizes, 
providing an answer to the third of four main questions 
addressed in this paper. 
Recall that a partial answer to this third question was 
given in Theorem \ref{klambda}, but the theorem was 
applicable to the special case when $\mathbf{m}_e^r$ was 
chosen as the center. 
The next theorem provides a way to calculate 
radius $1$ spheres for any center using the 
concept of duplication sets. Notation used in the theorem 
is defined subsequently and the proof 
is given toward the end of the section.

\begin{theorem}\label{ballcalc}
Recall that $n,r \in \mathbb{Z}_{>0}$ and  $r\vert n$. 
Let $\mathbf{m}_\sigma^r \in 
\mathcal{M}_r(\mathbb{S}_n)$.
Then 
\[
|S(\mathbf{m}_\sigma^r,1)| 
= 1 + (n-1)^2  - |SD(\mathbf{m}_\sigma^r)| - |AD(\mathbf{m}_\sigma^r)|.
\]
\end{theorem}

% Propositions \ref{zero sphere} and \ref{one sphere} 
In the permutation case, the Ulam metric is known to 
be left-invariant, i.e. given $\sigma, \pi, 
\tau \in \mathbb{S}_n,$ 
we have $\mathrm{d}_\circ(\sigma, \pi) = 
\mathrm{d}_\circ(\tau\sigma, \tau\pi)$
\cite{Farnoud}.
Left-invariance implies that permutation 
sphere sizes do not depend on the choice of center. 
Unfortunately, it is easily confirmed by counterexample 
that left invariance does not generally hold for the 
$r$-regular Ulam metric. 
Moreover, it is also easily confirmed that in the multipermutation 
Ulam sphere case, the choice of center has an impact on the 
size of the sphere, even when the radius remains unchanged
 (e.g. compare Proposition \ref{one sphere} 
to Proposition \ref{omegasphere} in the next section).
Hence we wish to consider spheres with various 
center multipermutations.

To aid with calculating such sphere sizes, 
we first find it convenient to introduce 
(as our own definition) 
the following subset of the set of translocations.  
\begin{definition}[$T_n$, unique set of translocations]
%Let $n \in \mathbb{Z}_{>0}$. 
 Define $T_n := \{ \phi(i,j) \in \mathbb{S}_n \; : \; i - j \ne 1\}.$ 
 
 We call $T_n$ the \textbf{unique set of translocations}.
\end{definition}
In words, $T_n$ is the set of all translocations, except 
translocations of the form $\phi(i,i-1)$.  We exclude translocations 
of this form because they can be modeled by translocations of the form 
$\phi(i-1,i)$, and are therefore redundant.  
We claim that the set $T_n$ is precisely the set of translocations 
needed to obtain all unique permutations within the 
Ulam sphere of radius $1$ via multiplication 
(right action).  
Moreover, there is no redundancy in the set, 
meaning no smaller set of translocations 
yields the entire Ulam sphere of radius $1$ 
when multiplied with a given center permutation. 
These facts are stated in the next lemma.

\begin{lemma}
Let $\sigma \in \mathbb{S}_{n}$. Then 
$S(\sigma, 1) = \{\sigma \phi \in \mathbb{S}_n \;:\; \phi \in T_n\}$,
and $|T_n| = |S(\sigma, 1)|$. 
\end{lemma}

%%%%%%%%%%%%%%%%%%%%%%%%%%%%
\begin{proof}
Let $\sigma \in \mathbb{S}_{n}$. We will first show that 
$S(\sigma, 1) = \{\sigma \phi \in \mathbb{S}_n \;:\; \phi \in T_n\}$. 
Note that 
\begin{eqnarray*}
S(\sigma, 1) \;\;\;=\;\;\;
\{\pi \in \mathbb{S}_n \;:\; \mathrm{d}_\circ(\sigma,\pi) \le 1\} 
\;\;\;=\;\;\; 
\{\sigma \phi(i,j) \in \mathbb{S}_n \; : \; i,j \in [n]\}.
\end{eqnarray*}
It is trivial that 
\begin{eqnarray*}
T_n 
\;\;\;=\;\;\;
 \{\phi(i,j) \in \mathbb{S}_n \;:\; i - j \ne 1\} 
\;\;\;\subseteq\;\;\;
 \{\phi(i,j) \in \mathbb{S}_n \;:\; i,j \in [n]\}.
\end{eqnarray*}
Therefore 
$\{\sigma\phi \in \mathbb{S}_n \;:\; \phi \in T_n\} \subseteq 
S(\sigma,1).$ 

To see why $S(\sigma, 1) \subseteq 
\{\sigma\phi \in \mathbb{S}_n \;:\; \phi \in T_n\}$, 
consider any $\sigma\phi(i,j) \in 
 \{\sigma \phi(i,j) \in \mathbb{S}_n \; : \; i,j \in [n]\} = 
S(\sigma,1).$  
If $i - j \ne 1,$ then $\phi(i,j) \in T_n,$ and thus 
$\sigma\phi(i,j) \in \{\sigma \phi \in \mathbb{S}_n \;:\; \phi \in T_n\}.$ 
Otherwise, if $i - j = 1,$ then $\sigma \phi(i,j) = \sigma \phi(j,i)$, and 
$i - j = 1 \text{ implies } j - i = -1 \ne 1,$ so $\phi(j,i) \in T_n.$  
Hence $\sigma\phi(i,j) = \sigma\phi(j,i) \in 
\{\sigma \phi \in \mathbb{S}_n \;:\; \phi \in T_n\}.$

%%%%%%%%%%%%%%%%%%%%%%%%%%%
Next we show that $|T_n| = |S(\sigma, 1)|$.
By Proposition \ref{one sphere}, $|S(\sigma,1)| = 1+ (n-1)^2$.  
On the other hand, $|T_n| = 
\#\{\phi(i,j) \in \mathbb{S}_n \;:\; i - j \ne 1\}.$  
If $i = 1,$ then there are $n$ values $j \in [n]$ 
such that $i - j \ne 1.$  Otherwise, if $i \in [n]$ but 
$i \ne 1,$ then there are $n-1$ values $j \in [n]$ such 
that $i - j \ne 1.$  However, for all $i, j \in [n],$ 
$\phi(i,i) = \phi(j,j) = e$ so that there are $n-1$ 
redundancies.  Therefore 
$|T_n| = n + (n-1)(n-1) - (n-1) = 1 + (n-1)^2.$    
\end{proof}

Although the Ulam sphere centered at $\sigma
 \in \mathbb{S}_n$ 
of radius $1$ can be characterized by all 
permutations obtainable by applying (multiplying on the right) 
a  translocation to $\sigma,$ the previous lemma shows 
that some translocations are redundant.  That is, there are 
translocations $\phi_1 \ne \phi_2$ such that 
$\sigma \phi_1 = \sigma \phi_2.$ In the case 
of permutations, the set 
$T_n$ has no such redundancies.  If 
$\phi_1, \phi_2 \in T_n,$ then 
$\sigma\phi_1 = \sigma\phi_2$ implies $\phi_1 = \phi_2.$  
However, in the case of multipermutations, 
the set $T_n$ can generally be shrunken 
further to exclude redundancies.

Given $\mathbf{m}_\sigma^r \in \mathcal{M}_r(\mathbb{S}_n)$,
the sphere 
$S(\mathbf{m}_\sigma^r, 1) 
= \{\mathbf{m}_\pi^r \in \mathcal{M}_r(\mathbb{S}_n) 
\;:\; %\mathrm{d}_\circ(\mathbf{m}_\sigma^r,\mathbf{m}_\pi^r) \le 1\} 
%n - \ell(\mathbf{m}_\sigma^r, \mathbf{m}_\pi^r) \le 1\} 
\text{ there exist } \phi \text{ such that } \mathbf{m}_\sigma^r \cdot \phi = \mathbf{m}_\pi\} 
= \{\mathbf{m}_\sigma^r \cdot \phi \in \mathcal{M}_r(\mathbb{S}_n) 
\;:\; \phi \in T_n\}.$  
However, it is possible that there exist 
$\phi_1, \phi_2 \in T_n$ such that $\phi_1 \ne \phi_2,$ 
but $\mathbf{m}_\sigma^r \cdot \phi_1 = \mathbf{m}_\sigma^r \cdot \phi_2.$ 
In such an instance we may refer to either $\phi_1$ or $\phi_2$ 
as a \textbf{duplicate translocation} for $\mathbf{m}_\sigma^r$.  
If we remove all duplicate translocations for $\mathbf{m}_\sigma^r$ 
from $T_n$, then the resulting set will have the same 
cardinality as the $r$-regular Ulam sphere of radius $1$ centered 
at $\mathbf{m}_\sigma^r$. The next definition (our own) is a standard set of duplicate translocations. It is called standard because 
as long as $r\ne1$ it always exists and is of predictable size. 

\begin{definition}[$SD(\mathbf{m})$, standard duplication set]
Given %$n \in \mathbb{Z}_{>0}$ and 
a tuple $\mathbf{m} \in \mathbb{Z}^n$, 
define
\begin{align*}
SD(\mathbf{m}) := 
\{\phi(i,j) \in T_n\backslash\{e\} \;:\; 
\mathbf{m}(i) = \mathbf{m}(j)
\text{ or }
\mathbf{m}(i) = \mathbf{m}(i-1) 
\}
\end{align*}
We call $SD(\mathbf{m})$ the \textbf{standard
duplication set} for $\mathbf{m}$. 
%For $i \in [n]$, we also define  
%$D_i(\mathbf{m}) :=
% \{\phi(i,j) \in T_n \;:\; \phi(i,j) \in SD(\mathbf{m})\}.$
\end{definition}

If we take an $r$-regular multipermutation $\mathbf{m}_\sigma^r,$ 
then removing %$SD(\mathbf{m}_\sigma^r)$ 
the general set of duplications from $T_n$ 
equates to removing a set 
of duplicate translocations.  
These duplications come in two varieties.  
The first variety corresponds to the 
first condition of the $SD(\mathbf{m})$ definition, 
when $\mathbf{m}(i) = \mathbf{m}(j)$.
For example, if $\mathbf{m}_\sigma^2 = (1,3,2,2,3,1)$, then we have 
$\mathbf{m}_\sigma^2 \cdot \phi(1,5) = (3,2,2,3,1,1) = 
\mathbf{m}_\sigma^2 \cdot \phi(1,6),$ 
since $\mathbf{m}_\sigma^2(2) = 3 = \mathbf{m}_\sigma^2(4)$. 
This is because moving the first $1$ to the left or to the right 
of the last $1$ results in the same tuple. 
The second variety corresponds to the second condition of 
the of $SD(\mathbf{m})$ definition above, 
when $\mathbf{m}(i) = \mathbf{m}(i-1)$. 
For example, if 
$\mathbf{m}_\sigma^2 = (1,3,2,2,3,1)$ as before, then for all 
$j \in [6],$ we have 
$\mathbf{m}_\sigma^2 \cdot \phi(3,j) = \mathbf{m}_\sigma^2 \cdot \phi(4, j).$ 
This is because any translocation that deletes and inserts
 the second of the two adjacent $2$'s does not result in 
 a different tuple when compared to deleting and 
 inserting the first of the 
 two adjacent $2$'s.

\begin{lemma}\label{D}
Let %$n,r \in \mathbb{Z}_{>0},$ $r\vert n$, and 
$\mathbf{m}_\sigma^r \in \mathcal{M}_r(\mathbb{S}_n)$. 
Then 
$S(\mathbf{m}_\sigma^r, 1) = 
\{\mathbf{m}_\sigma^r\cdot \phi \in \mathcal{M}_r(\mathbb{S}_n) \; \vert\; 
\phi \in T_n\backslash SD(\mathbf{m}_\sigma^r)\}.$
\end{lemma}

\begin{proof}
%The proof is omitted because of space constraints. 
Assume %$n,r \in \mathbb{Z}_{>0},$ $r\vert n$, and 
$\mathbf{m}_\sigma^r \in \mathcal{M}_r(\mathbb{S}_n)$. 
First note that 
$S(\mathbf{m}_\sigma^r, 1) = 
\{\mathbf{m}_\sigma^r \cdot \phi \in \mathcal{M}_r(\mathbb{S}_n)
 \;:\; \phi \in T_n\}$. 
Hence it suffices to show that for all 
$\phi(i,j) \in SD(\mathbf{m}_\sigma^r),$ there exists some 

$i', j' \in [n]$ such that $\phi(i',j') \in 
T_n\backslash SD(\mathbf{m}_\sigma^r)$ and 
$\mathbf{m}_\sigma^r \cdot \phi(i,j) 
= \mathbf{m}_\sigma^r\cdot\phi(i',j').$ 
We proceed by dividing the proof into two 
main cases.  Case I is when 
$(\mathbf{m}_\sigma^r(i) \ne \mathbf{m}_\sigma^r(i-1)$ or $i=1$. 
Case II is when 
$(\mathbf{m}_\sigma^r(i) = \mathbf{m}_\sigma^r(i-1).$ 

Case I (when $(\mathbf{m}_\sigma^r(i) \ne \mathbf{m}_\sigma^r(i-1)$ or $i=1$)
can be split into two subcases: 
\begin{align*}
&& \text{Case IA: } &i < j &&\\
&& \text{Case IB: } &i > j. &&
\end{align*}

We can ignore the instance when $i = j$, since 
$\phi(i,j) \in SD(\mathbf{m}_\sigma^r)$ implies 
$i \ne j.$  
For case IA, if for all $p \in [i,j]$
(for $a,b \in \mathbb{Z}$ with 
$a<b,$ the notation $[a,b] := \{a, a+1, \dots, b\}$) we have 
$\mathbf{m}_\sigma^r(i) = \mathbf{m}_\sigma^r(p)$, then 
$\mathbf{m}_\sigma^r\cdot \phi(i,j) = \mathbf{m}_\sigma^r \cdot e.$ 
Thus setting $i' = j' = 1$ yields the desired result. 
Otherwise, if there exists $p \in [i,j]$ such that 
$\mathbf{m}_\sigma^r(i) \ne \mathbf{m}_\sigma^r(p),$ 
then let $j^* := j - \min \{k \in \mathbb{Z}_{>0} \; \vert \; 
\mathbf{m}_\sigma^r(i) \ne \mathbf{m}_\sigma^r (j - k)\}.$ 
Then $\phi(i,j^*) \in T_n\backslash SD(\mathbf{m}_\sigma^r)$ 
and $\mathbf{m}_\sigma^r\cdot\phi(i,j) 
= \mathbf{m}_\sigma^r\cdot\phi(i,j^*).$ 
Thus setting $i' = i$ and $j' = j^*$ yields the desired result. 
Case IB is similar to Case IA.  ~%but details are omitted for brevity. \\

Case II (when 
$\mathbf{m}_\sigma^r(i) = \mathbf{m}_\sigma^r(i-1)$), 
can also be divided into two subcases. 
\begin{align*}
&&\text{Case IIA: }& i < j &&\\
&&\text{Case IIB: }& i > j .&&
\end{align*}

As in Case I, we can ignore the instance 
when $i = j$. 
For Case IIA, if for all $p \in [i,j]$ we have 
$\mathbf{m}_\sigma^r(i) = \mathbf{m}_\sigma^r(p),$ 
then $\mathbf{m}_\sigma^r \cdot \phi(i,j) = \mathbf{m}_\sigma^r \cdot e,$
so setting $i = j = 1$ achieves the desired result. 
Otherwise, if there exists $p \in [i,j]$ such that 
$\mathbf{m}_\sigma^r(i) \ne \mathbf{m}_\sigma^r(p),$ 
then let $i^* := i - \min 
\{k \in \mathbb{Z}_{>0} \;:\; 
(\mathbf{m}_\sigma^r(i) \ne \mathbf{m}_\sigma^r(i - k -1)) \text{ or } (i - k =1)\}$.  
Then 
$\mathbf{m}_\sigma^r\cdot \phi(i,j) = \mathbf{m}_\sigma^r\cdot \phi(i^*, j)$ 
and either one of the following is true: (1)
$\phi(i^*,j) \notin D_{i^*}(\mathbf{m}_\sigma^r) 
 \text{ implies } \phi(i^*,j) \notin SD(\mathbf{m}_\sigma^r)$, so set 
 $i' = i^*$ and $j' = j$; or (2) 
 by Case IA there exist $i',j' \in [n]$ such that
  $\phi(i',j') \in T_n \backslash 
SD(\mathbf{m}_\sigma^r)$ and
 $\mathbf{m}_\sigma^r \cdot\phi(i',j') = \mathbf{m}_\sigma^r\cdot \phi(i^*,j) 
 = \mathbf{m}_\sigma^r \cdot\phi(i,j).$ 
 Case IIB is similar to Case IIA. %but omitted here for brevity. 
\end{proof}

While Lemma \ref{D} shows 
that $SD(\mathbf{m}_\sigma^r)$ is a set of duplicate translocations for 
$\mathbf{m}_\sigma^r,$ 
we have not shown 
that $T_n\backslash SD(\mathbf{m}_\sigma^r)$ is the 
set of minimal size having the quality that 
$S(\mathbf{m}_\sigma^r,1) = 
\{\mathbf{m}_\sigma^r \cdot \phi \in \mathcal{M}_r(\mathbb{S}_n) \;:\; \phi \in 
T_n\backslash SD(\mathbf{m}_\sigma^r)\}$.  
In fact it is not minimal. 
In some instances it is possible to remove 
further duplicate translocations to reduce the set size.  
We will define another set of duplicate translocations, but 
a few preliminary definitions are first necessary.

We say that $\mathbf{m} \in \mathbb{Z}^n$ is 
\textbf{alternating} if for all odd integers $1 \le i \le n$, 
 $\mathbf{m}(i) = \mathbf{m}(1)$ and 
for all even integers $2\le i'\le n$, 
$\mathbf{m}(i') = \mathbf{m}(2)$ but 
$\mathbf{m}(1) \ne \mathbf{m}(2)$. 
In other words, any alternating tuple is 
of the form 
$(a,b,a,b,\dots,a,b)$ or $(a,b,a,b,\dots,a)$ 
where $a,b \in \mathbb{Z}$ and 
$a \ne b$.
Any singleton is also said to be alternating. 
Now for integers 
$1 \le i\le n$ and  $0\le k \le n-i$, the \textbf{substring} 
$\mathbf{m}[i, i + k]$ of $\mathbf{m}$ is defined as 
$\mathbf{m}[i, i + k ] := 
(\mathbf{m}(i), \mathbf{m}(i+1), \dots \mathbf{m}(i+k))$. 
Given a substring $\mathbf{m}[i,j]$ of $\mathbf{m}$, 
the \textbf{length}  of $\mathbf{m}[i,j]$, denoted by $|\mathbf{m}[i,j]|$,
is defined as $|\mathbf{m}[i,j]| := j-i+1$. 
As an example, if $\mathbf{m}' : = (1,2,2,4,2,4,3,1,3)$, 
then $\mathbf{m}'[3,6] = (2,4,2,4)$ is an alternating 
substring of $\mathbf{m}'$ of length $4$.

\begin{definition}[$AD(\mathbf{m})$, alternating duplication set]
Given $\mathbf{m} \in \mathbb{Z}^n$, define
 \begin{eqnarray*}
 AD(\mathbf{m}) := 
\{ &&\hspace{-.5cm}\phi(i,j) \in T_n \backslash SD(\mathbf{m})  \; : \;
 i < j \text { and there exists } \; k \in [i,j-2] \text{ such that } \\
&& \hspace{2.8cm} (\phi(j,k) \in T_n\backslash SD(\mathbf{m}))  \text{ and }
(\mathbf{m}\cdot \phi(i,j) = \mathbf{m}\cdot\phi(j,k)) \; \; \}.
\end{eqnarray*}

We call $AD(\mathbf{m})$ the 
\textbf{alternating duplication set} for $\mathbf{m}$ 
because it is only nonempty when 
$\mathbf{m}$ contains an alternating substring 
of length at least $4$. 
For each $i \in [n],$  also define 
$AD_i(\mathbf{m}) := 
 \{\phi(i,j) \in AD(\mathbf{m}) \;:\; j \in [n]\}.$ 
 Notice that $AD(\mathbf{m})
= \overset{n}{\underset{i =1}{\bigcup}} AD_i(\mathbf{m})$. 
\end{definition}
In the example of $\mathbf{m}' : = (1,2,2,4,2,4,3,1,3)$ above, 
$\mathbf{m}' \cdot \phi(2,6) = 
\mathbf{m}' \cdot \phi(6,3)$ and 
$\phi(2,6),\phi(6,3) \in T_9 \backslash SD(\mathbf{m}')$, 
implying that  $\phi(2,6) \in AD(\mathbf{m}')$. 
In fact, it can easily be shown that 
$AD(\mathbf{m}') = \{ \phi(2,6)\}$. 
In order to simplify the discussion of the alternating 
duplication set, we find the following lemma useful. 

\begin{lemma}\label{remark2}
Let $\mathbf{m}\in \mathbb{Z}^n$ and %$n \in \mathbb{Z}_{>0}$, and 
$i \in [n]$. 
Then 
$AD_i(\mathbf{m}) \ne \varnothing$ 
if and only if \\ 
1) $\mathbf{m}(i) \ne \mathbf{m}(i-1)$ \\
2) There exists $j \in [i+1,n]$ and $k \in [i,j-2]$ such that 

i) For all $p \in [i,k-1]$, $\mathbf{m}(p) = \mathbf{m}(p+1)$ 

ii) $\mathbf{m}[k,j]$ is alternating 

iii) $|\mathbf{m}[k,j]| \ge 4$. 
\end{lemma}
\begin{proof}
Let $\mathbf{m}\in \mathbb{Z}^n$ and %$n \in \mathbb{Z}_{>0}$, and 
$i \in [n]$. We will first assume 1) and 2) in the lemma statement 
and show that $AD_i(\mathbf{m})$ is not empty. 
Suppose $\mathbf{m}(i) \ne \mathbf{m}(i-1)$, and 
that there exists $j \in [i+1,n]$ and $k \in [i,j-2]$ such 
that for all $p \in [i,k-1]$, we have 
$\mathbf{m}(p) = \mathbf{m}(p+1)$. Suppose also 
that $\mathbf{m}[k,j]$ is alternating with 
$|\mathbf{m}[k,j]| \ge 4.$ 

For ease of notation, let 
$a := \mathbf{m}(k) = \mathbf{m}(k+2)$ and 
$b := \mathbf{m}(k+1) = \mathbf{m}(k+3)$ 
so that $\mathbf{m}[k,k+3] = (a,b,a,b) \in \mathbb{Z}^4$.  
Then 
\begin{align*}
(\mathbf{m} \cdot \phi(i,k+3))[k,k+3] = & 
~ (\mathbf{m}\cdot \phi(k,k+3))[k,k+3] \\ 
= & ~ (b,a,b,a) \\
= & ~ (\mathbf{m}\cdot \phi(k+3,k))[k,k+3].
\end{align*} 
Moreover, for all $p \notin [k,k+3],$ we have 
$(\mathbf{m}\cdot\phi(i,k+3))(p) = \mathbf{m}(p) 
=  (\mathbf{m}\cdot\phi(k+3,k))(p)$. Therefore 
$\mathbf{m}\cdot\phi(i,k+3) = 
\mathbf{m}\cdot \phi(k+3,k)$.  Also notice 
that $\mathbf{m}(i) \ne \mathbf{m}(i-1)$ 
implies that $\mathbf{m}\cdot \phi(i,k+3) 
\notin SD(\mathbf{m})$.  
Hence $\phi(i,k+3) \in AD_i(\mathbf{m})$. 

We now prove the second half of the lemma. 
That is, we assume that $AD_i(\mathbf{m}) 
\ne \varnothing$ and then show that 
1) and 2) necessarily hold. 
Suppose that $AD_i(\mathbf{m})$ is nonempty. 
Then $\mathbf{m}(i) \ne \mathbf{m}(i-1)$, since 
otherwise there would not exist any 
$\phi(i,j) \in T_n \backslash SD(\mathbf{m})$. 

Let $j \in [i+1,n]$ and $k \in [i,j-2]$ such that 
$\phi(j,k) \in T_n \backslash SD(\mathbf{m}) 
\text{ and } \mathbf{m}\cdot\phi(i,j) = \mathbf{m}(j,k)$. 
Existence of such $j$, $k$, and $\phi(j,k)$ is 
guaranteed by definition of $AD_i(\mathbf{m})$ and 
the fact that $AD_i(\mathbf{m})$ 
was assumed to be nonempty. 
Then for all $p \in [i,k-1]$, we have 
$\mathbf{m}(p) = \mathbf{m}(p+1)$ and 
for all $p \in [k,j-2]$, we have 
$\mathbf{m}(p) = \mathbf{m}(p+2)$. 
Hence either $\mathbf{m}[k,j]$ is 
alternating, or else 
for all $p,q \in [k,j]$, we have 
$\mathbf{m}(p) = \mathbf{m}(q)$. 
However, the latter case is impossible, 
since it would imply that for all 
$p,q \in [i,j]$ that $\mathbf{m}(p) = \mathbf{m}(q)$, 
which would mean $\phi(j,k) \notin 
T_n \backslash SD(\mathbf{m})$, a contradiction. 
Therefore $\mathbf{m}[k,j]$ is alternating. 

It remains only to show that $|\mathbf{m}[k,j]| \ge 4$. 
Since $k \in [i,j-2],$ it must be the case that 
$|\mathbf{m}[k,j]| \ge 3$. However, if 
$|\mathbf{m}[k,j]| = 3$ (which occurs when $k = j-2$), 
then $(\mathbf{m}\cdot\phi(i,j))(j) 
= \mathbf{m}(i) = \mathbf{m}(k) \ne \mathbf{m}(k+1) = 
(\mathbf{m}\cdot\phi(j,k)(j)$, 
which implies that 
$\mathbf{m}\cdot\phi(i,j) \ne \mathbf{m}\cdot\phi(j,k)$, 
a contradiction. 
Hence $|\mathbf{m}[k,j]| \ge 4$. 
\end{proof}

One implication of Lemma \ref{remark2} is that 
there are only two possible forms for 
$\mathbf{m}[i,j]$ where $\phi(i,j) \in AD_i(\mathbf{m})$. 
The first possibility is that  
$\mathbf{m}[i,j]$ is an alternating substring 
of the form 
$(a,b,a,b,\dots,a,b)$ 
(here $a,b\in \mathbb{Z}$), 
so that $\mathbf{m}[i,j] \cdot \phi(i,j)$ 
is of the form $(b,a,b,a\dots,b,a)$. 
In this case, as long as $|\mathbf{m}[i,j]| \ge 4$, 
then setting $k=i$ implies that 
$k \in [i,j-2]$, that 
$\phi(j,k) \in T_n\backslash SD(\mathbf{m})$, and that 
$\mathbf{m}[i,j]\cdot\phi(i,j) = 
\mathbf{m}[i,j]\cdot\phi(j,k)$. 

The other possibility is that $\mathbf{m}[i,j]$ 
is of the form 
$(\underset{k}{\underbrace{a,a,a,\dots,a}},
\underset{n-k}{\underbrace{b,a,b,\dots,a,b}})$ 
(again $a,b\in \mathbb{Z}$), so that 
$\mathbf{m}[i,j]\cdot\phi(i,j)$ is of the form 
$(\underset{k-1}{\underbrace{a,\dots,a}},
\underset{n-k+1}{\underbrace{b,a,b,\dots,b,a}})$. 
Again in this case, as long as $|\mathbf{m}[i,j]| \ge 4,$ then 
$k \in [i,j-2]$ with 
$\phi(j,k) \in T_n\backslash SD(\mathbf{m})$ and 
$\mathbf{m}[i,j]\cdot\phi(i,j) = 
\mathbf{m}[i,j]\cdot\phi(j,k)$. 
To simplify the calculation of 
$|AD(\mathbf{m}_\sigma^r)|$, we wish to define a 
set of equal size that is easier to count. 
The two remarks that follow the definition are obvious, but are 
helpful in proving that the size of the new set 
is equal to the size of $AD(\mathbf{m})$.

\begin{definition}[$AD^*(\mathbf{m})$]
Given %$n \in \mathbb{Z}_{>0}$ and 
$\mathbf{m} \in \mathbb{Z}^n$,
define 
\begin{eqnarray*}
AD^*(\mathbf{m}) := \{\;  (i,j) \in [n]\times [n] &:&
(\mathbf{m}[i,j] \text{ is alternating}), \; (|\mathbf{m}[i,j]| \ge 4), 
\text{ and } (|\mathbf{m}[i,j]| \text{ is even})\;   \}.
 \end{eqnarray*}
For each $i \in [n]$, also define 
$AD_i^*(\mathbf{m}) := 
\{(i,j) \in AD^*(\mathbf{m}) \;:\; j \in [n]\}.$ 
Notice that $AD^*(\mathbf{m})
= \overset{n}{\underset{i =1}{\bigcup}} AD_i^*(\mathbf{m})$. 
\end{definition} 

\begin{remark}\label{rem2} 
If $\mathbf{m} \in \mathbb{Z}^n$ is alternating and $n$ 
is even, then 
$\mathbf{m}\cdot\phi(1,n) 
= \mathbf{m}\cdot\phi(n,1)$. 
\end{remark}

\begin{remark}\label{rem3} 
If $\mathbf{m} \in \mathbb{Z}^n$ is alternating, 
$n \ge 3$, and $n$ is odd, 
then $\mathbf{m}\cdot\phi(1,n) \ne \mathbf{m}\cdot\phi(n,1)$. 
\end{remark}

\begin{lemma}\label{E*}
Let  $\mathbf{m} \in \mathbb{Z}^n$. 
Then $|AD(\mathbf{m})| = |AD^*(\mathbf{m})|$
\end{lemma}
\begin{proof}
Let $\mathbf{m} \in \mathbb{Z}^n$.
The idea of the proof is simple. 
Each element $\phi(i,j)\in AD(\mathbf{m})$ 
involves exactly one alternating sequence 
of length greater or equal to $4$, so the set
sizes must be equal. We formalize the 
argument by showing that 
$|AD(\mathbf{m})| \le |AD^*(\mathbf{m})|$ 
and then that 
$|AD^*(\mathbf{m})| \le |AD(\mathbf{m})|$.

To see why 
$|AD(\mathbf{m})| \le |AD^*(\mathbf{m})|$, 
we define a mapping 
$map : [n] \to [n]$, which maps 
index values either to the 
beginning of the nearest alternating 
subsequence to the right, or else 
to $n$. For all $i \in [n]$, let

 \begin{align*}
map(i) :=
\begin{cases}
  i + \min\{p \in \mathbb{Z}_{\ge 0} \;:\; 
    (\mathbf{m}(i) \ne \mathbf{m}(i+p+1))
     \text{ or }   (i+p = n) \}  
\hfill \;\;\;\;\; (\text{if } \mathbf{m}(i) \ne \mathbf{m}(i-1) \text{ or } i = 1) \\
 n
\hfill  \text{(otherwise)} 
  \end{cases} 
\end{align*}

Notice by definition of $map$, if 
$i,i' \in [n]$ such that $i \ne i'$, and 
if $\mathbf{m}(i) \ne \mathbf{m}(i-1)$ or $i=1$ 
and at the same time 
$\mathbf{m}(i') \ne \mathbf{m}(i'-1)$ or $i'=1$, 
then $map(i) \ne map(i')$. 

Now for each $i \in [n]$, 
if $\mathbf{m}(i) \ne \mathbf{m}(i-1)$ or 
$i=1$, then 
$|AD_i(\mathbf{m})|$ = $|AD_{map(i)}^*(\mathbf{m})|$ 
by Lemma \ref{remark2} and the two previous remarks.
Otherwise, if
$\mathbf{m}(i) = \mathbf{m}(i-1)$, 
then 
$|AD_i(\mathbf{m})| = |AD_{map(i)}^*(\mathbf{m})| = 0$. 
Therefore 
$|AD_i(\mathbf{m})| \le |AD_i^*(\mathbf{m})|$. 
This is true for all $i \in [n]$, so 
$|AD(\mathbf{m})| \le |AD^*(\mathbf{m})|$.

The argument to show that 
$|AD^*(\mathbf{m})| \le |AD(\mathbf{m})|$ 
is similar, except it uses the following 
function $map^* : [n] \to [n]$ instead of $map$.  
For all $i \in [n]$, let

 \begin{align*}
map^*(i) :=
\begin{cases}
  i - \min\{p \in \mathbb{Z}_{\ge 0} \;:\; 
    (\mathbf{m}(i) \ne \mathbf{m}(i-p-1))
     \text{ or }   (i-p = 1) \}  
\hfill \;\;\;\;\; (\text{if } \mathbf{m}(i) \ne \mathbf{m}(i-1) \text{ or } i = n) \\
 n
\hfill  \text{(otherwise)} 
  \end{cases} 
\end{align*}

\normalsize

\end{proof}

By definition, calculating $|AD^*(\mathbf{m})|$ equates 
to calculating the number of alternating substrings 
$\mathbf{m}[i,j]$ of $\mathbf{m}$ such that the length 
of the substring is both even and longer than 4. 
We can simplify the calculation of $AD(\mathbf{m})$ 
further by establishing a relation to the following quantity. 

\begin{definition}[$\psi(n)$, $\psi(\mathbf{x})$]
%Let $n \in \mathbb{Z}_{>0}$. 
Define 
\[
\psi(n) :=  \left\lfloor 
\frac{(n-2)^2}{4}  
\right \rfloor, 
\text{ \; \; and for $\mathbf{x} \in \mathbb{Z}^*$, define \; \; }
\psi(\mathbf{x}) := \underset{i=1}{\overset{|\mathbf{x}|}{\sum}}\psi(\mathbf{x}(i)), 
\]
where $|\mathbf{x}|$ denotes the length of the tuple $\mathbf{x}$. 
\end{definition}
While we define both $\psi(n)$ and $\psi(\mathbf{x})$ here, we will not 
make use of $\psi(\mathbf{x})$ until the following section. 
The next lemma relates $\psi(n)$ to the calculation of $AD(\mathbf{m})$.

\begin{lemma}\label{oddeven}
Let $\mathbf{m}$ be an alternating string. Then % of length $n\in \mathbb{Z}_{>0}$. Then 
\begin{align*}
 |AD(\mathbf{m})| = 
\psi(|\mathbf{m}|)
\end{align*}
\end{lemma}

\begin{proof}
Assume $\mathbf{m}$  is an alternating string and let $|\mathbf{m}| = n$. 
By Lemma \ref{E*}, 
$|AD(\mathbf{m})| = |AD^*(\mathbf{m})| =
|\overset{n}{\underset{i=1}{\bigcup}} AD_i^*(\mathbf{m})|$. 
Since $\mathbf{m}$ was assumed to be alternating, 
\begin{eqnarray*}
 |\overset{n}{\underset{i =1}{\bigcup}} AD_i^*(\mathbf{m})| 
&=& \#\{(i,j) \in [n] \times [n] \;:\; |\mathbf{m}[i,j]| \ge 4 
\text{ and } |\mathbf{m}[i,j]| \text{ is even} \} \\
&=& \#\{(i,j) \in [n]\times[n] \;:\; j-i+1 \in A\},
\end{eqnarray*}
where $A$ is the set of even integers between 
$4$ and $n$, i.e. 
$A := \{a \in [4,n] \;:\; a \text{ is even}\}$. 
For each $a \in A$ , we have 
\begin{eqnarray*} 
\#\{(i,j) \in [n] \times [n] \;:\; j-i+1 = a\} 
&=&\#\{i \in [n] \;:\; i \in [1,n-a+1]\} \\
&=& n - a + 1. 
\end{eqnarray*}

Therefore 
$|AD(\mathbf{m})| = \underset{a \in A}{\sum}(n-a+1)$. 
In the case that $n$ is even, then 
\begin{eqnarray*}
\underset{a \in A}{\sum}(n-a+1) 
\;\;\;=\;\;\; \underset{i=2}{\overset{n/2}{\sum}}(n-2i+1) 
\;\;\;=\;\;\; \left(\frac{n-2}{2} \right)^2
\;\;\;=\;\;\;
\psi(n).
\end{eqnarray*}

In the case that $n$ is odd, then 
\begin{eqnarray*}
\underset{a \in A}{\sum}(n-a+1) 
\;\;\;=\;\;\; \underset{i=2}{\overset{(n-1)/2}{\sum}}(n-2i+1) 
\;\;\;=\;\;\; \left(\frac{n-3}{2} \right) \left(\frac{n-1}{2}\right)
\;\;\;=\;\;\; \
\psi(n).
\end{eqnarray*}
\end{proof}

Notice that by Lemma \ref{oddeven}, 
it suffices to calculate 
$|AD(\mathbf{m})|$ for locally maximal length 
alternating substrings of $\mathbf{m}$. 
An alternating substring $\mathbf{m}[i,j]$ 
is of \textbf{locally maximal length} if and only if 
1) $\mathbf{m}[i-1]$ is not alternating or $i=1$; and 
2) $\mathbf{m}[i,j+1]$ is not alternating or $j = n$.

%%%%%%%%%%%%%%%%%%%%%%%%%%
Finally, we define the general set of duplications, 
$D(\mathbf{m})$. The 
lemma that follows the definition also shows that 
removing the set $D(\mathbf{m}_\sigma^r)$ from $T_n$ 
removes all duplicate translocations 
associated with $\mathbf{m}_\sigma^r$. 
\begin{definition}[$D(\mathbf{m})$, duplication set]
Given $n \in \mathbb{Z}_{>0}$ and 
$\mathbf{m} \in \mathbb{Z}^n$, define
\[D(\mathbf{m}) := SD(\mathbf{m}) \cup AD(\mathbf{m}).\]
We call $D(\mathbf{m})$ the \textbf{duplication 
set} for $\mathbf{m}$. For each $i \in [n]$, we also define 
$D_i(\mathbf{m}) := 
\{ \phi(i,j) \in D(\mathbf{m}) \; : \; j \in [n]\}$.
\end{definition}
%%%%%%%%%%%%%%%%%%%%%%%%%%%

\begin{lemma}\label{D*}
Let %$n,r  \in \mathbb{Z}_{>0}$, $r \vert n$, and 
$\mathbf{m}_\sigma \in \mathcal{M}_r(\mathbb{S}_n)$ and
 $\phi_1, \phi_2 \in T_n \backslash 
D(\mathbf{m}_\sigma^r) $. 
Then $\phi_1 = \phi_2$ if and only if  
$\mathbf{m}_\sigma^r\cdot \phi_1 = \mathbf{m}_\sigma^r \cdot \phi_2.$ 
\end{lemma}

%%%%%%%%%%%%%%%%%%%%%%%%%%%%%%%begin
\begin{proof}
Assume %$n,r  \in \mathbb{Z}_{>0}$, $r \vert n$, and 
$\mathbf{m}_\sigma \in \mathcal{M}_r(\mathbb{S}_n)$ and
 $\phi_1, \phi_2 \in T_n \backslash 
D(\mathbf{m}_\sigma^r) $. 
If $\phi_1 = \phi_2$ then
$\mathbf{m}_\sigma^r \cdot \phi_1 = \mathbf{m}_\sigma^r \cdot \phi_2$ 
trivially. It remains to prove that 
$\mathbf{m}_\sigma^r \cdot \phi_1 = \mathbf{m}_\sigma^r \cdot \phi_2 
\text{ implies } \phi_1 = \phi_2.$ We proceed by contrapositive. 
Suppose that $\phi_1 \ne \phi_2.$  
We want to show that 
$\mathbf{m}_\sigma^r\cdot \phi_1 \ne \mathbf{m}_\sigma^r\cdot\phi_2.$ 
Let $\phi_1 := \phi(i_1,j_1)$ and $\phi_2 := \phi(i_2, j_2)$.  
The remainder of the proof can be split into two main cases: 
Case I is if $i_1 = i_2$ and Case II is if $i_1 \ne i_2.$

Case I (when $i_1 = i_2$), can be further divided into two subcases: 
%\small
\begin{flalign*}
&& \text{Case IA: }&
\mathbf{m}_\sigma^r(i_1) = 
\mathbf{m}_\sigma^r(i_1 -1)&&\\
&&  \text{Case IB: } &
 \mathbf{m}_\sigma^r(i_1) \ne
\mathbf{m}_\sigma^r(i_1 -1). &&
\end{flalign*}
\normalsize

Case IA is easy to prove.  We have  
$D_{i_1}(\mathbf{m}_\sigma^r) = 
D_{i_2}(\mathbf{m}_\sigma^r) = 
\{\phi(i_1,j) \in T_n\backslash \{e\} \;:\; j \in [n]\}$, 
so $\phi_1 = e = \phi_2,$ a contradiction.  
For Case IB, we can first assume without loss of 
generality that $j_1 < j_2$ and then split 
into the following smaller subcases:  
%\small
\begin{align*}
&&\text{i) }& (j_1 < i_1) \text{ and } (j_2 > i_1)&& \\
&&\text{ii) }& (j_1 < i_1) \text{ and } (j_2 \le i_1)&& \\
&&\text{iii) }& (j_1 > i_1) \text{ and } (j_2 > i_1)&& \\ 
&&\text{iv) }& (j_1 > i_1) \text{ and } (j_2 \le i_1).&&
\end{align*}
\normalsize

However, subcase iv) is unnecessary since 
it was assumed that $j_1 < j_2,$ 
so $j_1 > i_1 \text{ implies } j_2 > j_1 > i_1.$  
Subcase ii) can also be reduced to 
$(j_1 < i_1) \text{ and } (j_2 < i_1)$ since 
$j_2 \ne i_2 = i_1.$  
Each of the remaining subcases is proven 
by noting that there is some element in the 
multipermutation $\mathbf{m}_\sigma^r \cdot \phi_1$ that is 
necessarily different from $\mathbf{m}_\sigma^r\cdot \phi_2.$ 
For example, in subcase i), we have 
$\mathbf{m}_\sigma^r\cdot \phi_1(j_1) = 
\mathbf{m}_\sigma^r(i_1) \ne 
\mathbf{m}_\sigma^r(j_1) = \mathbf{m}_\sigma^r\cdot\phi_2(j_1).$ 
Subcases ii) and iii) are solved similarly. \\

Case II (when $i_1 \ne i_2$) can be divided into three subcases: 
%\small
\begin{flalign*}
&&\text{Case}& \text{ IIA: }
(\mathbf{m}_\sigma^r(i_1) = \mathbf{m}_\sigma^r(i_1 -1) 
\text{ and } 
\mathbf{m}_\sigma^r(i_2) = \mathbf{m}_\sigma^r(i_2-1)), &&\\ 
&&\text{Case}& \text{ IIB: either } && \\
&& &(\mathbf{m}_\sigma^r(i_1) = \mathbf{m}_\sigma^r(i_1 -1) 
\text{ and } 
\mathbf{m}_\sigma^r(i_2) \ne \mathbf{m}_\sigma^r(i_2-1)) &&\\
&& & \text{ \hspace{-6.8mm} or } 
(\mathbf{m}_\sigma^r(i_1) \ne \mathbf{m}_\sigma^r(i_1 -1) 
\text{ and } 
\mathbf{m}_\sigma^r(i_2) = \mathbf{m}_\sigma^r(i_2-1)), &&\\ 
&&\text{Case}& \text{ IIC: }
(\mathbf{m}_\sigma^r(i_1) \ne \mathbf{m}_\sigma^r(i_1 -1) 
\text{ and } 
\mathbf{m}_\sigma^r(i_2) \ne \mathbf{m}_\sigma^r(i_2-1)).  &&
\end{flalign*}
\normalsize

Case IIA is easily solved by mimicking the proof of Case IA.  
Case IIB is also easily solved as follows.  
First, without loss of generality, 
we assume that 
$\mathbf{m}_\sigma^r(i_1) = \mathbf{m}_\sigma^r(i_1-1) \text{ and } 
\mathbf{m}_\sigma^r(i_2) \ne \mathbf{m}_\sigma^r(i_2-1).$  
Then $D_{i_1}(\mathbf{m}_\sigma^r) = 
\{\phi(i_1,j)\in T_n \backslash \{e\} \; : \; 
j \in [n]\},$ so $\phi_1 = e.$  
Therefore we have 
$\mathbf{m}_\sigma^r\cdot \phi_1(j_2) = \mathbf{m}_\sigma^r(j_2) 
\ne \mathbf{m}_\sigma^r(i_2) = \mathbf{m}_\sigma^r\cdot\phi_2(i_2-1).$ \\

Finally, for Case IIC, without loss of generality we may 
assume that $i_1 < i_2$ and then split into the following four subcases: 
%\small
\begin{align*}
&&\text{i) } &(j_1 < i_2) \text{ and } (j_2 \ge i_2) &&\\
&&\text{ii) } &(j_1 < i_2) \text{ and } (j_2 < i_2) &&\\
&&\text{iii) } &(j_1 \ge i_2) \text{ and } (j_2 \ge i_2) &&\\
&&\text{iv) } &(j_1 \ge i_2) \text{ and } (j_2 < i_2). && 
\end{align*}
\normalsize

However, since 
$\phi(i_2,j_2) \in T_n\backslash D(\mathbf{m}_\sigma^r)$ 
implies $i_2 \ne j_2,$ 
subcases i) and iii) can be reduced to 
$(j_1 < i_2) \text{ and } (j_2 > i_2)$ and 
$(j_1 \ge i_2) \text{ and } (j_2 > i_2)$ respectively.  
For subcase i), we have 
$\mathbf{m}_\sigma^r\cdot \phi_1(j_1) = \mathbf{m}_\sigma^r(i_1) 
\ne \mathbf{m}_\sigma^r(j_1) = \mathbf{m}_\sigma^r \cdot \phi_2(j_1).$ 
Subcases ii) and iii) are solved in a similar manner.  
For subcase iv), 
if $j_1 > i_2$, then $\mathbf{m}_\sigma^r \cdot \phi_1(j_1) = 
\mathbf{m}_\sigma^r (i_1) \ne \mathbf{m}_\sigma^r(j_1) = 
\mathbf{m}_\sigma^r \cdot \phi_2(j_1).$  
Otherwise, if $j_1 = i_2$, then 
$\phi_1 = \phi(i_1, i_2)$ and 
$\phi_1 = \phi(i_2, j_2).$ 
Thus if $\mathbf{m}_\sigma^r\cdot \phi_1 = 
\mathbf{m}_\sigma^r \cdot \phi_2$ then 
$\phi_1 \in D_{i_1}(\mathbf{m}_\sigma^r),$ 
which implies that 
$\phi_1 \notin T_n \backslash D(\mathbf{m}_\sigma^r)$, 
a contradiction. 
\end{proof}
%%%%%%%%%%%%%%%%%%%%%%%%%%%%%end

Lemma \ref{D*} implies that we can 
calculate $r$-regular Ulam sphere sizes 
of radius $1$ whenever we can calculate the 
appropriate duplication set. This calculation 
can be simplified by noting that for a sequence 
$\mathbf{m} \in \mathbb{Z}^n$ that 
$SD(\mathbf{m}) \cap AD(\mathbf{m}) = \varnothing$ 
(by the definition of $AD(\mathbf{m}))$ 
and then decomposing the duplication set into these 
components. 
This idea is stated in Theorem \ref{ballcalc} 
at the beginning of this section, 
which like Theorem \ref{klambda}, 
is a partial answer the the third main question of this paper. 
We now have the machinery to prove Theorem \ref{ballcalc}.
\\~\\
%%%%%%%%%%%%%%%%%%%%%%%%%%
%\begin{proof}
\textit{proof of Theorem \ref{ballcalc}} \\
Let %$n,r \in \mathbb{Z}_{>0}$,  $r\vert n$, 
$\mathbf{m}_\sigma^r \in \mathcal{M}_r(\mathbb{S}_n)$.
By the definition of $D(\mathbf{m}_\sigma^r)$ and 
lemma \ref{D}, 
\begin{eqnarray*}
&&\{\mathbf{m}_\sigma^r \cdot \phi \in \mathcal{M}_r(\mathbb{S}_n) :
\phi \in T_n\backslash D(\mathbf{m}_\sigma^r)\} \\
&=&
\{\mathbf{m}_\sigma^r\cdot \phi \in \mathcal{M}_r(\mathbb{S}_n) \; :\; 
\phi \in T_n\backslash SD(\mathbf{m}_\sigma^r)\} \\
&=& S(\mathbf{m}_\sigma^r,1).
\end{eqnarray*}
This implies $|T_n\backslash D(\mathbf{m}_\sigma^r)| 
\;\ge\; |S(\mathbf{m}_\sigma^r,1)|.$
By lemma \ref{D*}, for $\phi_1,\phi_2 \in 
T_n\backslash D(\mathbf{m}_\sigma^r),$ 
if $\phi_1 \ne \phi_2,$ then $\mathbf{m}_\sigma^r\cdot \phi_1 
\ne \mathbf{m}_\sigma^r\cdot\phi_2.$ 
Hence we have $|T_n\backslash D(\mathbf{m}_\sigma^r)|
\;\le\; |S(\mathbf{m}_\sigma^r,1)|,$ which implies that 
$|T_n\backslash D(\mathbf{m}_\sigma^r)|
= |S(\mathbf{m}_\sigma^r,1)|.$
It remains to show that 
$|T_n\backslash D(\mathbf{m}_\sigma^r)| 
\;=\; 1 + (n-1)^2  - |SD(\mathbf{m}_\sigma^r)| - |AD(\mathbf{m}_\sigma^r)|.$ 
This is an immediate consequence of the fact that 
$|T_n| = 1 + (n-1)^2$ and 
$SD(\mathbf{m}_\sigma^r) \cap AD(\mathbf{m}_\sigma^r) = \varnothing$.
\hfill $\square$ \\
%\end{proof}
%%%%%%%%%%%%%%%%%%%%%%%%%%%%

Theorem \ref{ballcalc} reduces the calculation of 
$|S(\mathbf{m}_\sigma^r,1)|$ to calculating 
$|SD(\mathbf{m}_\sigma^r)|$ and $|AD(\mathbf{m}_\sigma^r)|$. 
It is an easy matter to calculate $|SD(\mathbf{m}_\sigma^r)|$, 
since it is exactly equal to $(n-2)$ times the number of 
$i \in [n]$ such that 
$\mathbf{m}_\sigma^r(i) = \mathbf{m}_\sigma^r(i-1)$ plus 
$(r-1)$ times the number of $i \in [n]$ such that 
$\mathbf{m}_\sigma^r(i) \ne \mathbf{m}_\sigma^r(i-1)$ or $i = 1$. 
We also showed how to calculate $|AD(\mathbf{m})|$ earlier. 
The next example is an application of Theorem \ref{ballcalc} 

\begin{example}
Suppose %$\sigma := [1,2,3,4,9,6,7,11,5,10,12,8].$  Then 
$\mathbf{m}_\sigma^3 = (1,1,1,2,3,2,3,2,4,4,3,4)$. 
There are 3 values of $i \in [12]$ 
such that $\mathbf{m}_\sigma^3(i) = \mathbf{m}_\sigma^3(i-1)$, 
which implies that $|SD(\mathbf{m}_\sigma^3| = 
(3)(12-2) + (12-3)(3-1) = 48$. Meanwhile, by 
Lemmas \ref{E*} and \ref{oddeven}, 
$|AD(\mathbf{m}_\sigma^3)| = ((5-3)/2)((5-1)/2)) = 2$. 
By Theorem \ref{ballcalc}, 
 $|S(\mathbf{m}_\sigma^3), 1| = (12-1)^2 - 48 - 2 = 71$. 
\end{example}

\section{Min/Max Spheres and Code Size Bounds}\label{minmax} 

In this section we show choices of center achieving 
 minimum and maximum $r$-regular Ulam sphere sizes
for the radius $t=1$ case. 
As an application, we also state new upper and lower 
bounds on maximal code size in 
Lemmas \ref{upperbound}, \ref{perfect bound}, and 
\ref{G-V bound} (Lemmas \ref{bound1} and \ref{bound2} 
may also be included in this list, which are bounds in the 
special case when $n/r=2$). 
These bounds represent the final main contribution of this 
paper, answering the fourth main question. 

The binary case, when $n/r = 2$, presents unique challenges 
because of the nature of its alternating duplication sets. 
In particular, the choice of center multipermutation yielding 
the maximal sphere size in the non-binary cases does not 
yield the maximal size in the binary case. Thus we divide this 
section into parts -- the first subsection treating the non-binary 
case, and the remaining two subsections treating the binary case. 

\subsection{non-binary case} 
We begin by discussing the non-binary case in this subsection. 
The non-binary case is the general case where $n/r \ne 2$. 
Tight minimum and maximum values of sphere sizes 
are explicitly given. We then discuss 
resulting bounds on code size. First let us consider 
the $r$-regular Ulam sphere of minimal size. 
The first two lemmas presented in this section apply to all cases, 
both non-binary and binary, while the remaining 
results only apply when $n/r \ne 2$.

\begin{lemma}\label{esphere}
Recall that $n,r \in \mathbb{Z}_{>0}$ and $r\vert n$. 
Let $\mathbf{m}_\sigma^r \in \mathcal{M}_r(\mathbb{S}_n)$. 
Then 
\[|S(\mathbf{m}_e^r, 1)| \;\;\le\;\; |S(\mathbf{m}_\sigma^r,1)|.\]
\end{lemma}

%%%%%%%%%%%%%%%%%%%%%%%%%%%
\begin{proof}
Assume %$n,r \in \mathbb{Z}_{>0}$, $r\vert n$, and 
$\mathbf{m}_\sigma^r \in \mathcal{M}_r(\mathbb{S}_n)$. 
In the case that $n/r = 1,$ then 
$\mathbf{m}_e^r = e$ and $\mathbf{m}_\sigma^r = \sigma$, 
so that 
$|S(\mathbf{m}_e^r, 1)| = |S(\mathbf{m}_\sigma^r,1)|$. 
Therefore we may assume that $n/r \ge 2$. 
By Theorem \ref{ballcalc}, 
$\underset{\sigma \in \mathbb{S}_n}{\min}
(|S(\mathbf{m}_\sigma^r,1)|)  
= 1 + (n-1)^2 - 
\underset{\sigma \in \mathbb{S}_n}{\max} 
(|SD(\mathbf{m}_\sigma^r)| + |AD(\mathbf{m}_\sigma^r)|). 
$
Since $n/r \ge 2,$ we know that 
$n-2 > r-1,$ which implies that for all 
$\sigma \in \mathbb{S}_n$, that 
$|SD(\mathbf{m}_\sigma^r)|$ is maximized by 
maximizing the number of integers $i \in [n]$ such that 
$\mathbf{m}_\sigma^r(i) = \mathbf{m}_\sigma^r(i-1)$. 
This is accomplished by choosing $\sigma = e,$ and 
hence for all 
$\sigma \in \mathbb{S}_n$, we have 
$|SD(\mathbf{m}_e^r)| \ge |SD(\mathbf{m}_\sigma^r)|$.  

We next will show that for any increase in the size of 
$|AD(\mathbf{m}_\sigma^r)|$ compared to 
$|AD(\mathbf{m}_e^r)|$, that $|SD(\mathbf{m}_\sigma^r)|$ 
is decreased by a larger value compared to 
$|SD(\mathbf{m}_e^r)|$, so that 
$(|SD(\mathbf{m}_\sigma^r)| + |AD(\mathbf{m}_\sigma^r)|)$ is 
maximized when $\sigma = e$. 
By Lemmas \ref{E*} and \ref{oddeven}, 
$|AD(\mathbf{m}_\sigma^r)|$ is characterized by 
the lengths of its locally maximal alternating substrings.  
For every locally maximal alternating substring 
$\mathbf{m}_\sigma^r[a,a+k-1]$ (here $a,k \in \mathbb{Z}_{>0}$) of 
$\mathbf{m}_\sigma^r$ of length $k$, 
there are at least $k-2$ fewer instances where 
$\mathbf{m}_\sigma^r = \mathbf{m}_\sigma^r(i-1)$ or $i = 1$
when compared to instances where 
$\mathbf{m}_e^r(i) = \mathbf{m}_e^r(i-1)$. 
This is because for all $i \in [a+1,a+k-1]$, 
$\mathbf{m}_\sigma^r(i) \ne \mathbf{m}_\sigma^r(i-1)$ and $i+1 \ne 1$. 
Hence for each locally maximal alternating substring 
$\mathbf{m}_\sigma^r[a, a+k-1]$, then 
$|SD(\mathbf{m}_\sigma^r)|$ is decreased by at least 
$(k-2)(n-2 - (r-1)) \ge (k-2)(r-1)$ when compared to 
$|SD(\mathbf{m}_e^r)|$. 
Meanwhile, $|AD(\mathbf{m}_\sigma^r)|$ 
is increased by the same 
locally maximal alternating substring by at most 
$(k-2)((k-2)/4)$ by Lemma \ref{oddeven}. 
However, since $k \le 2r$, we have 
$(k-2)((k-2)/4) \le (k-2)(r-1)/2$, which is 
of course less than $(k-2)(r-1)$. 
\end{proof}
%%%%%%%%%%%%%%%%%%%%%%%%%%%%%%

Lemma \ref{esphere}, along with 
Proposition \ref{one sphere} implies that 
the $r$-regular Ulam sphere size of radius $t=1$ 
is bounded (tightly) below by $(1 + (n-1)(n/r-1))$. 
This in turn implies the 
following sphere-packing type upper bound on 
any single error-correcting code.

\begin{lemma}\label{upperbound}
%Let $n,r \in \mathbb{Z}_{>0}$ and $r \vert n$. 
If $C$ is a single-error correcting
 $\mathsf{MPC}_\circ(n,r)$ code, 
then 
\[
|C|  \;\;\le\;\; \frac{n!}{(r!)^{n/r}\left(1+(n-1)(n/r-1)\right)}.
\]
\end{lemma}

%%%%%%%%%%%%%%%%%%%%%%%%%%%%%
\begin{proof}
Let $C$ be a single-error correcting 
$\mathsf{MPC}_\circ(n,r)$ code. 
A standard sphere-packing bound argument 
implies that 
$|C| \le {(n!)}/({(r!)^{n/r}
(\underset{\sigma \in \mathbb{S}_n}{\min} 
|S(\mathbf{m}_\sigma^r,1)|)}$.
The remainder of the proof follows from 
Proposition \ref{one sphere} and Lemma \ref{esphere}.
\end{proof}
%%%%%%%%%%%%%%%%%%%%%%%%%%%%%

We have seen that $|S(\mathbf{m}_\sigma^r)|$ is minimized 
when $\sigma = e$. We now discuss the choice of center 
yielding the maximal sphere size. 
Let 
$\omega \in \mathbb{S}_n$ be defined as follows: 
$\omega(i) := ((i-1)\mod (n/r))r + \lceil ir/n \rceil $ 
and $\omega := [\omega(1), \omega(2), \dots \omega(n)].$ 
With this definition, for all $i \in [n],$ we have 
$\mathbf{m}_{\omega}^r(i) = i \mod (n/r)$ 
For example, if $r = 3$ and $n=12,$ then 
$\omega = [1, 4, 7, 10, 2, 5, 8, 11, 3, 6, 9, 12]$ and 
$\mathbf{m}_{\omega}^r = (1, 2, 3, 4, 1, 2, 3, 4, 1, 2, 3, 4).$ 
We can use Theorem \ref{ballcalc} to calculate 
$|S(\mathbf{m}_\omega^r,1)|$, and then 
show that this is the 
maximal $r$-regular Ulam sphere size 
(except for the case when $n/r = 2$).

\begin{lemma}\label{omegasphere}
Suppose %$n,r \in \mathbb{Z}_{>0}$, $r \vert n$, and  
$n/r \ne 2.$  Then 
\[|S(\mathbf{m}_\sigma^r,1)| 
\;\le\; |S(\mathbf{m}_{\omega}^r,1)|
\;=\;  1+(n-1)^2 - (r-1)n.\]
\end{lemma}

%%%%%%%%%%%%%%%%%%%%%%%%%%%%%%%%%%%%
\begin{proof}
Assume %$n,r \in \mathbb{Z}_{>0}$, $r \vert n$, and  
$n/r \ne 2.$ 
First notice that if $n/r = 1$ then for any $\sigma \in \mathbb{S}_n$ 
(including $\sigma = \omega$), the sphere
$S(\mathbf{m}_{\sigma}^r,1)$
contains exactly one 
element (the tuple of the form $(1,1,\dots,1)$). 
Hence the lemma holds trivially in this instance. 
Next, assume that $n/r > 2.$  
We will first prove that $|S(\mathbf{m}_{\omega}^r,1)| 
=  1+(n-1)^2 - (r-1)n$.

Since $n/r > 2,$ it is clear that 
 $\mathbf{m}_{\omega}^r$ contains no 
alternating subsequences of length greater than $2$. 
Thus by Lemma \ref{remark2}, 
$AD(\mathbf{m}_{\omega}^r) = \varnothing$ 
and therefore by Theorem \ref{ballcalc}, 
$|S(\mathbf{m}_{\omega}^r,1)| = 
1 + (n-1)^2 - |SD(\mathbf{m}_{\omega}^r)|.$ 
Since there does not exist $i \in [n]$ such that 
$\mathbf{m}_{\omega}^r (i) = \mathbf{m}_{\omega}^r(i-1)$, 
we have $|SD(\mathbf{m}_{\omega}^r)| = (r-1)n$, 
completing the proof of the first statement in the lemma. 

We now prove that 
$|S(\mathbf{m}_\sigma^r,1)| 
\le |S(\mathbf{m}_{\omega}^r,1)|$.
Recall that 
$|SD(\mathbf{m}_\sigma^r)|$ is equal to 
$(n-2)$ times the number of 
$i \in [n]$ such that 
$\mathbf{m}_\sigma^r(i) = \mathbf{m}_\sigma^r(i-1)$ plus 
$(r-1)$ times the number of $i \in [n]$ such that 
$\mathbf{m}_\sigma^r(i) \ne \mathbf{m}_\sigma^r(i-1)$. 
But $n/r > 2$ implies that $r-1 < n-2$, which implies 
$\underset{\mathbf{m}_\pi^r \in \mathcal{M}_r(\mathbb{S}_n)}{\min}
|SD(\mathbf{m}_\pi^r,1)| = (r-1)n$. Therefore 

\begin{eqnarray*}
|S(\mathbf{m}_\sigma^r,1)| &\le& 
 1 + (n-1)^2 - 
 \underset{\mathbf{m}_\pi^r \in \mathcal{M}_r(\mathbb{S}_n)}{\min}
 |SD(\mathbf{m}_\pi^r,1)|   - 
 \underset{\mathbf{m}_\pi^r \in \mathcal{M}_r(\mathbb{S}_n)}{\min}
 |AD(\mathbf{m}_\pi^r,1)| \\
 &\le& 1 + (n-1)^2 - 
 \underset{\mathbf{m}_\pi^r \in \mathcal{M}_r(\mathbb{S}_n)}{\min}
   |SD(\mathbf{m}_\pi^r,1)|  \\
 &=&\underset{}{} 1 + (n-1)^2 - (r-1)n \\
 &=&
|S(\mathbf{m}_{\omega}^r,1)|.
\end{eqnarray*}
\normalsize

\end{proof}
%%%%%%%%%%%%%%%%%%%%%%%%%%

The upper bound of lemma \ref{omegasphere}
implies a lower bound on a perfect single-error correcting 
$\mathsf{MPC}(n,r)$.

\begin{lemma}\label{perfect bound}
%Let $n,r \in \mathbb{Z}_{>0}$, $r \vert n$, and  
Suppose $n/r \ne 2.$
If $C$ is a perfect single-error
 correcting $\mathsf{MPC}(n,r)$, then 
\[
\frac{n!}{(r!)^{n/r}  ((1+(n-1)^2) - (r-1)n)}
\;\;\;\le\;\;\;
|C|.
\]
\end{lemma}

%%%%%%%%%%%%%%%%%%%%%%%%%%%%
\begin{proof}
Assume % $n,r \in \mathbb{Z}_{>0}$, $r \vert n$, 
$n/r \ne 2$,
and that $C$ is a perfect single-error 
correcting $\mathsf{MPC}(n,r)$. 
Then $\underset{\mathbf{m}_c^r \in 
\mathcal{M}_r(C)}{\sum}
|S(\mathbf{m}_c^r,1)| = {(n!)}/({(r!)^{n/r}})$. 
This means 
\begin{align*}
\frac{n!}{(r!)^{n/r}}
\;\;\;\le \;\;\;
\left(|C| \right) \cdot
\left( \underset{\mathbf{m}_c^r \in \mathcal{M}_r(C)}{\max}
\left(|S(\mathbf{m}_c^r,1)|\right)\right),
\end{align*} 
which by Lemma \ref{omegasphere} 
implies the desired result. 
\end{proof}

%%%%%%%%%%%%%%%%%%%%%%%%%%%%

A more general lower bound is easily
obtained by applying Lemma \ref{omegasphere} with a 
standard Gilbert-Varshamov bound argument.
While the lower bound of Lemma \ref{perfect bound} 
applies only to perfect codes that are $\mathsf{MPC}(n,r,d)$ with 
$d \ge 3$, the next lemma applies 
to any $\mathsf{MPC}(n,r,d)$, which may or may not be perfect. 

\begin{lemma}\label{G-V bound}
%Let $n,r \in \mathbb{Z}_{>0}$, $r \vert n$, and  
Suppose $n/r \ne 2$, 
and let $C\subseteq \mathcal{M}_r(\mathbb{S}_n)$ be an 
$\mathsf{MPC}_\circ(n,r,d)$ code of maximal cardinality. Then 
\[
 \frac{n!}{(r!)^{n/r} (1 + (n-1)^2 - (r-1)n )^{d-1}  }
\;\;\;\le\;\;\;
|C| 
\]
\end{lemma}

\begin{proof}
Assume that %$n,r \in \mathbb{Z}_{>0}$, $r \vert n$, and  
$n/r \ne 2$, and that 
$C$ is an $\mathsf{MPC}_\circ(n,r,d)$ code
of maximal cardinality. 
For all $\mathbf{m}_\sigma^r \in \mathcal{M}_r(\mathbb{S}_n)$, 
there exists 
$c \in 
C$ such that 
$\mathrm{d}_\circ(\mathbf{m}_\sigma^r,c) \le d-1$. 
Otherwise, we could add 
$\mathbf{m}_\sigma^r \notin C$ to $C$ while maintaining 
a minimum distance of $d$, contradicting 
the assumption that $|C|$ is maximal. 

Therefore 
$\underset{\mathbf{m}_c^r \in \mathcal{M}_r(C)}
{\bigcup}S(\mathbf{m}_c^r,d-1)
= \mathcal{M}_r(\mathbf{S}_n)$.
This in turn implies that 
\begin{align*}
 \frac{n!}{(r!)^{n/r}}
\;\;\;\le \;\;\;
\underset{\mathbf{m}_c^r\in \mathcal{M}_r(C)}
{\sum}|S(\mathbf{m}_c^r,d-1)|.
\end{align*}
Of course, the right hand side of the 
above inequality is less than or equal to 
$\left(|C|\right) \cdot
\left( \underset{\mathbf{m}_c^r \in \mathcal{M}_r(C)}{\max}
|S(\mathbf{m}_c^r,d-1)|\right)$. 
Finally Lemma \ref{omegasphere} implies that 
\[
\underset{\mathbf{m}_c^r \in \mathcal{M}_r(C)}
{\max}(|S(\mathbf{m}_c^r,d-1)|)
\;\;\;\le\;\;\; 
(1+(n-1)^2 - (r-1)n)^{d-1} 
\] 
so the 
conclusion holds. 
\end{proof}

\subsection{binary case -- cut location maximizing sphere size}

In the previous subsection we were able to find center multipermutations 
whose sphere sizes were both minimal (Lemma \ref{esphere}) and maximal
(Lemma \ref{omegasphere}). These 
were used to provide bounds on the maximum code size
(Lemmas \ref{upperbound}, \ref{perfect bound}, \ref{G-V bound}). 
However, a complication arises that prevents Lemma \ref{omegasphere} 
from applying to the binary case, the case when $n/r = 2$.  
We say that $\mathbf{m}_\sigma^r \in \mathcal{M}_r(\mathbb{S}_n)$ 
is a \textbf{binary multipermutation} if and only if $n/r = 2$.
The next two subsections focus on determining the maximum sphere 
size for binary multipermutations. The current subsection addresses the question of cut location. The notion of cuts is defined in the following paragraphs. 
For the remainder of the paper we assume that $n$ is an even 
integer and that $n/r = 2$ (equivalently $r=n/2$).

Since we are assuming that $n/r = 2$, 
by definition $\mathbf{m}_\omega^r$ is 
an $n$-length alternating string, which results in the size of the 
alternating duplication set $AD(\mathbf{m}_\omega^r)$ 
increasing rapidly as $n$ increases.
This in turn results in $|S(\mathbf{m}_\omega^r,1)|$ 
 no longer being maximal 
 (in the sense of Lemma \ref{omegasphere}).  
 For example, if $n=12$, then we have $\mathbf{m}_\omega^r = 
(1,2,1,2,1,2,1,2,1,2,1,2,1,2)$, which would imply that 
$|AD(\mathbf{m}_\omega^r)| = \psi(12) 
%\lfloor((12-2)/2)^2\rfloor 
= 25$. 

To compensate for this problem, it is best to ``cut" the original 
$\mathbf{m}_\omega^r$ into some number $c$ of locally maximal 
 alternating substrings. Whenever $\mathbf{m}$ is a tuple in 
two symbols, for example when $\mathbf{m} \in \{1,2\}^n$, we use the term 
\textbf{cut} to refer to any locally maximal alternating substring of 
$\mathbf{m}$. This language applies to binary multipermutations.
Considering the example above when $n=12$, we could instead 
take the binary multipermutation $(1,2,1,2,1,2,2,1,2,1,2,1)$, which 
has two cuts of length $6$, namely $(1,2,1,2,1,2)$ and $(2,1,2,1,2,1)$ 
as opposed to a single length $12$ cut in the original $\mathbf{m}_\omega^r$.
Notice here that the standard duplication 
set increases by $5$ but the new alternating duplication set 
size is now 
$\psi(6) + \psi(6) = 8$, 
%$\lfloor((6-2)/2)^2\rfloor + \lfloor((6-2)/2)^2\rfloor = 8$, 
a decrease in $17$. 

Intuitively, these cuts should be chosen so that each 
is as similar in length as possible in 
order to minimize the total size of the alternating duplication set. 
For example, 
$(\underset{\text{1st cut}}{\underbrace{1,2,1,2,1,2}},
\underset{\text{2nd cut}}{\underbrace{2,1,2,1,2,1}})$, which 
has an alternating duplication set of size $8$ is preferable to 
$(\underset{\text{1st cut}}{\underbrace{1,2}},
\underset{\text{2nd cut}}{\underbrace{2,1,2,1,2,1,2,1,2,1}})$, 
which has an alternating duplication set of size $16$. 
This idea is proven subsequently.  
Another question concerns 
the optimal number of such cuts, since each time a cut is introduced 
 the standard duplication set size necessarily increases. This 
question is addressed in the next subsection, and it turns out that 
having approximately $\sqrt{r}$ cuts minimizes total duplications and 
thus results in the maximum sphere size. 

To start this subsection, we will show that given a multipermutation 
with a fixed number $c$ of cuts, 
the alternating duplication set is minimized when these cut
lengths are as similar in length as possible. In order to simplify the 
argument, the following two lemmas reduce the discussion to the lengths of these alternating substrings. 
%Throughout the next two subsections, we assume that 
%$c \in [n-1]$ (i.e. $c$ is a positive integer less than $n$) 
%and that $n/r =2$.
%\textbf{NO N FROM HERE}

\begin{lemma}\label{equalonestwos}
Let %$n$ be a positive even integer and 
$\mathbf{m} \in \{1,2\}^n$. 
Then there exists a binary multipermutation 
$\mathbf{m}_\sigma^r \in \mathcal{M}_r(\mathbb{S}_n)$ 
such that  $\mathbf{m}_\sigma^r = \mathbf{m}$ 
if and only if 
$\#\{ i \in [n] \;:\; \mathbf{m}(i) = 1\}
\;=\; 
\#\{ i \in [n] \;:\; \mathbf{m}(i) = 2\}$, 
i.e. 
the number of $1$'s and $2$'s of $\mathbf{m}$
are equal.
\end{lemma}
\begin{proof}
%Recall that $n$ is a positive even integer and 
Assume $\mathbf{m} \in \{1,2\}^n$.  First suppose that 
there exists a binary multipermutation 
$\mathbf{m}_\sigma^r \in \mathbb{S}_n$ such that 
$\mathbf{m}_\sigma^r = \mathbf{m}$. Then 
by the definition of binary multipermutations, 
$\#\{ i \in [n] \;:\; \mathbf{m}(i) = 1\} = r = 
\#\{ i \in [n] \;:\; \mathbf{m}(i) = 2\}$, completing the first direction. 

For the second direction of the proof, suppose that 
$\#\{ i \in [n] \;:\; \mathbf{m}(i) = 1\}
= \#\{ i \in [n] \;:\; \mathbf{m}(i) = 1\} = n/2 = r$. 
Then we can construct a binary multipermutation 
$\mathbf{m}_\sigma^r$ with the property that 
$\mathbf{m}_\sigma^r = \mathbf{m}$ as follows: 
Define $\{i_1, i_2, \dots, i_r\} := \{i \in [n] \;:\; \mathbf{m}(i) = 1\}$ and 
$\{i_{r+1}, i_{r+2}, \dots i_n\} := \{i \in [n] \;:\; \mathbf{m}(i) = 2 \}$. 
For all $j \in [n]$, set $\sigma(i_j) := j$ and define 
$\sigma := (\sigma(1), \sigma(2), \dots \sigma(n))$.  
Then $\mathbf{m}_\sigma^r = \mathbf{m}$. 
\end{proof}

\begin{lemma}\label{gsequence}
  Let %$n$ be a positive even integer, 
  $c \in [n-1]$, and 
  $(q(1), q(2), \dots, q(c)) \in \mathbb{Z}_{>0}^c$ such that 
  $\sum_{i=1}^c q(i) = n$. 
  Then there exists $i \in [c]$ such that $q(i)$ is even if and only if 
  there exists a binary multipermutation 
  $\mathbf{m}_\sigma^r \in \mathcal{M}_r(\mathbb{S}_n)$ such that 
  \[
  \mathbf{m}_\sigma^r = (
  \underset{q(i_1)} {\underbrace{\mathbf{m}_\sigma^r[a_1,b_1]}}, 
  \underset{q(i_2)} {\underbrace{\mathbf{m}_\sigma^r[a_2,b_2]}}, 
  \dots,
   \underset{q(i_c)} {\underbrace{\mathbf{m}_\sigma^r[a_c,b_c]}}
  ),
  \]
  where for all 
  $i \in [c]$, $a_i, b_i \in [n]$, and $a_i \le b_i$ such that
  $\mathbf{m}_\sigma^r[a_i,b_i]$ is a cut (locally maximal alternating substring). 
\end{lemma}

The proof of Lemma \ref{gsequence} can be found in the appendices. 
In words, the lemma states that given any tuple of positive integers 
$(q(1), q(2), \dots q(c))$ whose entries sum to $n$, 
as long as there is at least one even 
integer in the tuple, then the entries can be made to correspond to the 
lengths of the cuts of some binary multipermutation $\mathbf{m}_\sigma^r$. 
Notice that in the formulation resulting from Lemma \ref{gsequence}, 
the number of cuts $c$ in a binary multipermutation $\mathbf{m}_\sigma^r$ 
is one more than the number of repeated adjacent digits. 
In other words, 
$c 
\;=\;
\#\{  i \in [2,n]  \;:\; \mathbf{m}_\sigma^r(i) = \mathbf{m}_\sigma^r(i-1) \} +1 $.
Hence for a fixed number of cuts $c$, 
the standard duplication set size $|SD(\mathbf{m}_\sigma^r)|$  
does not depend on the lengths of individual cuts. 

On the other hand, the size of the alternating duplication set
 $|AD(\mathbf{m}_\sigma^r)|$ does depend on the lengths of the cuts. 
This means that if the number of cuts is fixed at $c$ then 
by Lemma \ref{oddeven} and Lemma \ref{gsequence}, 
finding the maximum sphere size equates 
to minimizing 
%$\sum_{i=1}^c \lfloor ((g(i)-2)/2)^2 \rfloor$, 
$\psi((q(1),q(2),\dots,q(c))$, 
where $(q(1), q(2), \dots q(c)) \in \mathbb{Z}_{>0}^c$ has 
at least one even entry and whose entries sum to $n$. 
We claim that the tuple defined next minimizes the sum in question. 

\begin{definition}[$q_c$, $rem_c$, $\mathbf{q}_c$]
Let %$n$ be a positive even integer and 
$c\in [n-1]$. Denote by 
$q_c\in \mathbb{Z}_{> 0}$ and 
$rem_c \in \mathbb{Z}_{\ge 0}$ the unique 
quotient and remainder when $n$ is divided by $c$, i.e. 
$c*q_c + rem_c = n$ where $rem_c < c$. Define 

\[
\mathbf{q}_c := 
\begin{cases} 
(q_c +1, \underset{c-2}{\underbrace{q_c, \dots q_c}}, q_c - 1) 
~ ~ ~ ~ ~ ~ ~\text{if } q_c \text{ is odd and } rem_c = 0 \\
(\underset{rem_c}{\underbrace{q_c+1, \dots q_c + 1}}, 
\underset{c-rem_c}{\underbrace{q_c, \dots q_c}}) 
\hfill \text{otherwise}
\end{cases}
\]
We also use the notation $\mathbf{q}_c = (\mathbf{q}_c(1), \dots \mathbf{q}_c(c)) \in \mathbb{Z}_{>0}^c$. 
\end{definition}

The above definition guarantees that two important conditions 
are satisfied: (1) the entries of $\mathbf{q}_c$ 
sum to $n$; and (2) there exists some $i \in [c]$ such that 
$\mathbf{q}_c(i)$ is even. These two conditions correspond with 
the conditions and statement of Lemma \ref{gsequence}.
Additionally, by definition, 
$\mathbf{q}_c$ is a weakly decreasing sequence with all 
entries being positive integers, and thus 
it is a partition of $n$.

Standard calculation (see Remark \ref{floordiff2} in Appendix D)  
indicates that if two cuts of a binary multipermutation 
differ by $2$ or more, then the size of the alternating duplication set 
associated with that multipermutation can be reduced by bringing the 
length of those two cuts closer together.  Generalizing over all the 
cuts in the multipermutation, we may minimize the alternating 
duplication set and hence maximize sphere size by choosing all cuts 
to be as similar in length as possible. Another way of saying that 
the cut sizes are as similar in length as possible is to say 
that the cut sizes are precisely the values of $\mathbf{q}_c$. 
The fact that cut sizes equaling the values of $\mathbf{q}_c$
minimizes the associated alternating duplication set size 
is stated in the next theorem.
% the first main contribution of this paper.

\begin{theorem}\label{minE}
Let $c \in [n-1]$. Then 
%Let $n$ be a positive even integer.
%Let $n,r \in \mathbb{Z}_{>0}$ and $n/r =2$. 
%  $c<n$. 
%Also let
%$\#\{ i \in [2,n] \;:\; \mathbf{m}_\sigma^r(i) = \mathbf{m}_\sigma^r(i-1) \} 
%\;=\; c-1$. Then 
\[
\underset{\mathbf{m}_\sigma^{r} \in 
\mathcal{M}_{r}^c(\mathbb{S}_n)}{\min} 
|AD(\mathbf{m}_\sigma^r)| 
\;=\; 
%\sum_{i=1}^c 
%\left\lfloor 
%\left(
%\frac{\mathbf{q}_c(i)-2}{2}
%\right)^2 
%\right\rfloor.
\psi(\mathbf{q}_c), 
\]
where 
$\mathcal{M}_{r}^c(\mathbb{S}_n) := \{\mathbf{m}_\pi^r\in \mathcal{M}_r(\mathbb{S}_n) 
\; : \; 
\#\{\mathbf{m}_\pi^r(i) = \mathbf{m}_\pi^r(i-1)\} +1 = c\}$, i.e. 
$\mathcal{M}_{r}^c(\mathbb{S}_n)$ is the set of binary multipermutations 
with exactly $c$ cuts. 
\end{theorem}

\begin{proof}
Assume $c\in[n-1]$. 
Note first that by Lemma \ref{gsequence}, there 
exists a binary multipermutation with exactly 
$c$ cuts, whose cut lengths 
correspond to $\mathbf{q_c}$. 
Now let $(a(1), a(2), \dots a(c)) \in \mathbb{Z}_{>0}^c$ such that 
$\sum_{i=1}^ca(i) = n$ and there exists $i\in[c]$ such that 
$a(i)$ is even. 
Again by Lemma \ref{gsequence}, $(a(1),a(2),\dots a(c)$ 
corresponds to the cut lengths of an arbitrary 
binary multipermutation with exactly $c$ cuts. 
Hence by Lemma \ref{oddeven} it suffices to show that 
$\psi(\mathbf{q}_c) \le \psi((a(1), a(2), \dots a(c))$. 
%By Lemma \ref{gsequence}, finding $\sigma$ that 
%minimizes $|AD(\mathbf{m}_\sigma^r)|$ is equivalent to finding 
%the minimum of $\sum_{i=1}^{c}\lfloor((a(i)-2)/2)^2\rfloor$
%over all possible $(a(1), a(2), \dots a(c)$. 
We divide the remainder of the proof into two halves corresponding to the 
the split definition of $\mathbf{q}_c$. 

First, suppose $q_c$ is odd and $rem_c = 0$ so that 
$\mathbf{q}_c = (q_c+1, q_c, \dots, q_c, q_c-1)$.   
Then since there exists $i \in [c]$ such that $a(i)$ is 
even, there must be distinct $i'$ and $j'$ in $[c]$ such 
that $a_{i'} = q_c + h_{i'}$ and $a_{j'} = q_c - h_{j'}$ 
where $h_{i'}, h_{j'} \in \mathbb{Z}_{>0}$. 
Hence, by Remark \ref{floordiff2} (see appendix \ref{appendix D}), 
\begin{align*}
%\sum_{i=1}^{c} \left\lfloor \left(\frac{q_c-2}{2}\right)^2 \right\rfloor + 1 
\psi((
\underset{c}{\underbrace{q_c, q_c, \dots, q_c}}))+1
\; \; \le \;\;
%\sum_{i=1}^c \left\lfloor \left(\frac{a(i)-2}{2}\right)^2 \right\rfloor, 
\psi((
\underset{c}{\underbrace{a(1),a(2),\dots, a(c)}})),
\end{align*}
but also by Remark \ref{floordiff2} (applied to the first and 
last entry of $\mathbf{q}_c$), 
\begin{eqnarray*}
%\sum_{i=1}^c \left\lfloor \left(\frac{\mathbf{q}_c(i)-2}{2}\right)^2 \right\rfloor
\psi(\mathbf{q}_c)
\;\;=\;\;
%\left\lfloor \left(\frac{(q_c+1)-2}{2}\right)^2 \right\rfloor +
\psi(q_c+1) + 
\psi((\underset{c-2}{\underbrace{q_c, q_c, \dots, q_c}})) + 
\psi(q_c-1)
%\sum_{i=2}^{c-1} \left\lfloor \left(\frac{q_c-2}{2}\right)^2 \right\rfloor +
%\left\lfloor \left(\frac{(q_c-1)-2}{2}\right)^2 \right\rfloor 
\;\;=\;\;
\psi((\underset{c}{\underbrace{q_c, q_c, \dots, q_c}})) + 1.
% \sum_{i=1}^{c} \left\lfloor \left(\frac{q_c-2}{2}\right)^2 \right\rfloor + 1 
\end{eqnarray*}

For the second half, suppose that $q_c$ is even or that $rem_c \ne 0$. 
Then $\mathbf{q}_c = (\underset{rem_c}{\underbrace{q_c+1, \dots, q_c+1}}, 
\underset{c-rem_c}{\underbrace{q_c, \dots, q_c}})$. This means that 
for all $i,j \in [c]$, that $|\mathbf{q}_c(i) - \mathbf{q}_c(j)| 
\;\le\;
 1$. Hence, by 
Remark \ref{floordiff2}, 
\[
%\sum_{i=1}^c \left\lfloor \left(\frac{\mathbf{q}_c(i)-2}{2}\right)^2 \right\rfloor 
\psi(\mathbf{q}_c)
\;\; \le \;\;
\psi(\underset{c}{\underbrace{a(1),a(2),\dots, a(c)}})).
%\sum_{i=1}^c \left\lfloor \left(\frac{a(i)-2}{2}\right)^2 \right\rfloor
\]
%and again by Lemma \ref{gsequence}, the conclusion holds. 
\end{proof}

We have shown that choosing cuts to be as evenly distributed 
as possible results in minimizing the alternating duplication set. 
However, as mentioned before, while increasing cuts generally 
decreases the size of the alternating duplication set, it also 
increases the size of the standard duplication set. The question of
the optimal number of cuts in a multipermutation $\mathbf{m}_\sigma^r$ 
minimizing
$|SD(\mathbf{m}_\sigma^r)| + |AD(\mathbf{m}_\sigma^r)|$ remains. 

\subsection{binary case -- number of cuts maximizing sphere size}
The previous subsection demonstrated the nature of 
cuts maximizing sphere size in the binary case 
once the number of cuts $c$ is fixed. This 
subsection focuses on determining the number of cuts 
maximizing binary multipermutation sphere size. 
Computer analysis for values of $r$ up to $10,000$ 
suggests that 
%%%%%FIRST CASE!!!!
$c \approx \sqrt{r}$ cuts minimizes the sum of 
$|SD(\mathbf{m}_\sigma^r)|$ and $|AD(\mathbf{m}_\sigma^r)|$ 
(and therefore maximizes the sphere size).
The next remark and subsequent lemmas prove that 
this is indeed the case. We therefore call $\sqrt{r}$ the 
\textbf{ideal cut value} and use the notation 
$\hat{c} \;:=\; \sqrt{r}$. In practice the actual optimal number of cuts 
is only approximately equal to $\hat{c}$ since 
$\hat{c}$ is not generally an integer. 
As in the previous subsection, recall that 
we assume $n$ is a positve even integer and that $n/r = 2$ for 
the remainder of this paper. 

\begin{remark} \label{remark51}
Let $\mathbf{m}_\sigma^r \in \mathcal{M}_r(\mathbb{S}_n)$
 be a binary multipermutation. Then 
%Let $n,r,c \in \mathbb{Z}_{>0}$, $n/r =2$, and 
%  $c<n$. If 
%$\#\{i \in [2,n] \;:\; \mathbf{m}_\sigma^r(i) \;=\; \mathbf{m}_\sigma^r(i-1)\}
%\;=\; c -1$. 
\[
|SD(\mathbf{m}_\sigma^r)| \;\;=\;\;
(c-1)(n-2) + (n-(c-1))(r-1) \;\;=\;\; c(r-1) + (n-1)(r-1),
\]
where 
$c := \#\{i \in [2,n] \;:\; 
\mathbf{m}_\sigma^r(i) \;=\; \mathbf{m}_\sigma^r(i-1)\} + 1$.
\end{remark}

Note that the remark could technically be simplified 
by rewriting $n$ as $2r$, but here and elsewhere 
$n$ is kept in favor of $2r$ to retain intuition behind 
the meaning and for ease of comparison with 
previous results in the non-binary case. 
Although the remark is obvious, its significance is that 
the only component that depends upon $c$ is $c(r-1)$. 
This means that each time the number of cuts is increased by 
$1$, the size of the standard duplication set is increased by $r-1$. 

Therefore to show that $\hat{c}$ cuts minimizes duplications, it 
is enough to show the following two facts: 
(1) if the number of cuts is greater or equal to $\hat{c}$, then 
increasing the number of cuts by one 
causes a decrease in the alternating 
duplication set by at most $r-1$; and 
(2) if the number of 
cuts is less than or equal to $\hat{c}$, then a further decrease in cuts 
by one will enlarge the alternating duplication set by at least $r-1$.
These two facts are expressed in the next two lemmas. 

\begin{lemma}\label{cplusone}
Let $c \in [n-2]$ and $\hat{c} \le c$. Then 
%$n,r,c \in \mathbb{Z}_{>0}$, $n/r =2$, and $c<n-1$.
%If $\hat{c} \le c$, then 
\begin{align}\label{eqn5}
%\overset{c}{\underset{i=1}{\sum}}
%\left\lfloor\left( 
%\frac{\mathbf{q}_{c}(i)-2}{2}
%\right)^2\right\rfloor 
\psi(\mathbf{q}_c) 
-
\psi(\mathbf{q}_{c+1})
%\overset{c+1}{\underset{i=1}{\sum}}
%\left\lfloor\left( 
%\frac{\mathbf{q}_{c+1}(i)-2}{2}
%\right)^2\right\rfloor 
\;\;\; \le \;\;\; r-1. 
\end{align}
\end{lemma}

The proof for Lemma \ref{cplusone} is in the appendices. 
The next example demonstrates how to 
construct $\mathbf{q}_{c+1}$ from $\mathbf{q}_{c}$ 
when $\hat{c} \le c < n-1$. 
This corresponds to increasing the number of cuts from 
$c$ to $c+1$. Notice that 
$q_{c+1} > c$ and that each 
cut is decreased by at most $2$, with some cuts 
decreased by only $1$. This corresponds 
to the second case in the proof of Lemma \ref{cplusone}. 
%%%%%%%%%%%%%%%%%
\begin{example}\label{example2}
Let $n = 30$ and $c = 4$. 
Notice that $\hat{c} = \sqrt{15} \approx 3.873$ so that $\hat{c} < c < n-1$.
We also have $q_{4} = 7$ and $rem_{4} = 2$ 
while $q_{5} = 6$ and $rem_5 = 0$. 
Therefore $\mathbf{q}_{4} = (8,8,7,7)$ 
and $\mathbf{q}_{5} = (6,6,6,6,6)$. 

We may visualize $\mathbf{q}_{4}$ and $\mathbf{q}_{5}$ 
respectively as the left and right 
diagrams in Figure \ref{example2fig}, 
with the $i$th row of the diagram corresponding to 
the $i$th cut, $\mathbf{q}_{4}(i)$ or $\mathbf{q}_{5}(i)$. 
The numbers in the blocks in the left diagram 
of Figure \ref{example2fig} represent the order in which 
each row would be shortened to construct the last 
cut of $\mathbf{q}_4$. 

%\vspace{-.2cm}
\begin{figure}[h]
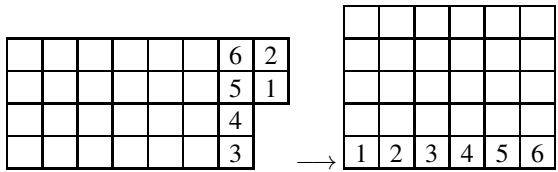

\center
\resizebox{1\linewidth}{!}{
  \begin{minipage}{\linewidth}
 \begin{align*}
\begin{Young} 
&&& & & & 6&2 \cr 
&&& & & &5 & 1  \cr 
&&& & & &  4 \cr
&&& & & &3   \cr
\end{Young} 
\longrightarrow 
\begin{Young} 
&& & & &  \cr
&& & & &  \cr
& & & & & \cr
& & & & & \cr
1&2 &3 &4 &5& 6  \cr
\end{Young} 
\end{align*}
\end{minipage} }
\caption{Constructing $\mathbf{q}_{5}$ from $\mathbf{q}_{4}$ (when $n = 30$)}
\label{example2fig}
\end{figure}

If $\mathbf{m}_\sigma^r$ is a multipermutation with four
cuts whose lengths correspond to $\mathbf{q}_4$, then applying 
Remark \ref{remark51} and Lemma \ref{oddeven}, 
$|SD(\mathbf{m}_\sigma)| + |AD(\mathbf{m}_\sigma)|  
= 492$.  By Theorem \ref{ballcalc}, this means 
$|S(\mathbf{m}_\sigma,1)| = 238$.  On the other hand, 
if $\mathbf{m}_\pi$ is a multipermutation with five cuts 
whose lengths correspond to $\mathbf{q}_5$, then 
similar methods show 
$|SD(\mathbf{m}_\pi)|+ |AD(\mathbf{m}_\pi)| = 496$, 
which implies $|S(\mathbf{m}_\pi,1)| = 234$, a smaller value. 
\end{example}

Lemma \ref{cplusone} implied that if the number of cuts 
is greater or equal to $\hat{c}$, then increasing 
cuts shrinks the overall possible sphere size. The next 
lemma is analogous. It says that if the number 
of cuts is less than or equal to $\hat{c}$, then 
reducing the number of cuts shrinks the overall possible 
sphere size. 

\begin{lemma}\label{cminusone}
%$n,r,c \in \mathbb{Z}_{>0}$, $n/r =2$, and $c<n$.
Let $c \in [n-1]$ and $c \le \hat{c}$. Then 
\begin{align}\label{eqn8}
%\overset{c-1}{\underset{i=1}{\sum}}
%\left\lfloor\left( 
%\frac{\mathbf{q}_{c-1}(i)-2}{2}
%\right)^2\right\rfloor 
\psi(\mathbf{q}_{c-1})
-
%\overset{c}{\underset{i=1}{\sum}}
%\left\lfloor\left( 
%\frac{\mathbf{q}_{c}(i)-2}{2}
%\right)^2\right\rfloor 
\psi(\mathbf{q}_c)
\;\;\;>\;\;\; r-1. 
\end{align}
\end{lemma}

The proof for Lemma \ref{cminusone} is in the appendices.
The next example demonstrates how to 
construct $\mathbf{q}_c$ from $\mathbf{q}_{c-1}$ 
when $c \le \hat{c}$. Notice that each 
cut length is decreased by at least $2$. 
%%%%%%%%%%%%%%%%%
\begin{example}\label{example1}
Let $n = 34$ and $c = 4$. 
Notice that $\hat{c} = \sqrt{17} \approx 4.123$ so that $c < \hat{c}$.
We also have $q_{3} = 11$ and $rem_{3} = 1$ 
while $q_4 = 8$ and $rem_4 = 2$. 
Therefore $\mathbf{q}_{3} = (12, 11, 11)$ 
and $\mathbf{q}_4 = (9, 9, 8, 8)$. 
We can visualize 
$\mathbf{q}_{3}$ and $\mathbf{q}_4$ respectively 
as the left and right 
diagrams in Figure \ref{example1fig}, 
with the $i$th row of each diagram corresponding to 
the $i$th cut, $\mathbf{q}_{3}(i)$   
 or $\mathbf{q}_4(i)$. 
The numbers in the blocks in the left diagram 
of Figure \ref{example1fig} represent the order in which 
each row would be shortened to construct the last 
cut of $\mathbf{q}_4$. 

\vspace{-.2cm}
\begin{figure}[h]
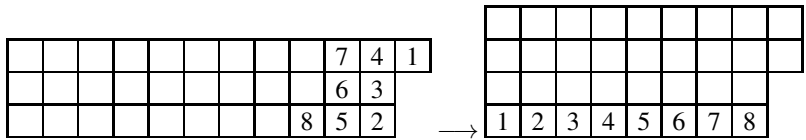

\resizebox{1\linewidth}{!}{
  \begin{minipage}{\linewidth}
\begin{align*}
\begin{Young} 
& & & & & & & & & 7&4&1 \cr 
& & & & & & & &  &6 &3 \cr 
& & & & & & & & 8&5&2\cr
\end{Young} 
\longrightarrow 
\begin{Young} 
 & & & & & & & & \cr 
 & & & & & & & & \cr 
 & & & & & & & \cr
1& 2& 3& 4& 5& 6& 7&8  \cr
\end{Young} 
\end{align*}
\end{minipage} }
\caption{Constructing $\mathbf{q}_4$ from $\mathbf{q}_{3}$ (when $n=34$)}
\label{example1fig}
\end{figure}

If $\mathbf{m}_\sigma$ is a multipermutation with three 
cuts whose lengths correspond to $\mathbf{q}_3$ above, then applying 
Remark \ref{remark51} and Lemma \ref{oddeven}, 
$|SD(\mathbf{m}_\sigma)| + |AD(\mathbf{m}_\sigma)|  
= 641$.  By Theorem \ref{ballcalc}, this means 
$|S(\mathbf{m}_\sigma,1)| = 449$.  On the other hand, 
if $\mathbf{m}_\pi$ is a multipermutation with four cuts 
whose lengths correspond to $\mathbf{q}_4$ above, then 
similar methods show 
$|SD(\mathbf{m}_\pi)|+ |AD(\mathbf{m}_\pi)| = 634$, 
which implies that $|S(\mathbf{m}_\pi,1)| = 456$, a larger value. 
\end{example}

Lemmas \ref{cplusone} and \ref{cminusone} imply 
that the number of cuts $c$ minimizing the sum 
$|SD(\mathbf{m}_\sigma^r)| + |AD(\mathbf{m}_\sigma^r)|$ 
(and thus maximizing sphere size)
observes the inequalities $\hat{c}-1 < c < \hat{c}+1$. 
This answers the question of the optimal number of cuts. 
For a particular value $r$, it is a relatively simple matter to calculate 
the exact size of the maximal Ulam multipermutation sphere. 
One simply has to determine whether 
$c = \lfloor \hat{c} \rfloor$ or $c = \lceil \hat{c}\rceil$ 
yields a smaller 
$|SD(\mathbf{m}_\sigma^r)| + |AD(\mathbf{m}_\sigma^r)|$
(here $\lceil x \rceil$ denotes the ceiling function on $x \in \mathbb{R}$, 
i.e. the least integer greater than or equal to $x$). 
Once the best choice for $c$ is ascertained, an 
application of Theorem \ref{ballcalc} will yield the 
maximum size for that particular $r$. 
The above statements are summarized in the next theorem.
% which is the second main contribution of this paper. 

\begin{theorem}\label{main2}
%~\\
%\resizebox{1\linewidth}{!}{
%  \begin{minipage}{\linewidth}
%\begin{multline}\label{eqn14}
%Let $n,r,c \in \mathbb{Z}_{>0}$, $n/r =2$, and $c<n$. Then 
\begin{align*}
%&
\underset{\mathbf{m}_\sigma^r \in 
\mathcal{M}_r(\mathbb{S}_n)}{\max}
|S(\mathbf{m}_\sigma^r,1)| %& %\\
 \; \; = \; \; 
 %&
 1 + (n-1)^2 
- \underset{c\in \{\lfloor \hat{c} \rfloor, \lceil \hat{c}\; \rceil \}}{\min}
\Bigg(
c(r-1)+(n-1)(r-1)  %\left.
+ 
%\sum_{i=1}^{c}
%\left\lfloor
%\left(\frac{\mathbf{q}_c(i)-2}{2}\right)^2
%\right\rfloor
\psi(\mathbf{q}_c)
\Bigg).
& \numberthis \label{eqn14}
\end{align*}
%\end{multline}
%  \end{minipage} } 
\end{theorem}

\begin{proof}
The proof is an immediate consequence of Theorem \ref{ballcalc}, 
Theorem \ref{minE}, 
Remark  \ref{remark51},  and Lemmas \ref{cplusone} and \ref{cminusone},
\end{proof}

Once again we retain $n$ instead of $2r$ in the 
theorem statement for intuition purposes. 
Theorem \ref{main2} indicates that the number of 
cuts $c$ maximizing sphere size should be either 
$\lfloor \hat{c} \rfloor$ or $\lceil \hat{c} \rceil$, but 
it does not state when $\lfloor \hat{c} \rfloor$ or 
$\lceil \hat{c} \rceil$ is optimal. It turns out that 
whichever is closer to the true value of 
$\hat{c}$ will yield the maximal sphere size. 
That is, if $\hat{c} - \lfloor \hat{c} \rfloor \;\le\;  
\lceil\hat{c}\rceil - \hat{c}$, then 
$\lfloor \hat{c} \rfloor$ (appropriately chosen) cuts will 
yield the maximal sphere size and visa versa. 
Stated another way, if 
$\hat{c} \;\le\; \lfloor \hat{c} \rfloor + 0.5$, then 
$ \lfloor \hat{c} \rfloor$ cuts provides the largest possible sphere size, 
but if $\hat{c} \;\ > \; \lfloor \hat{c} \rfloor + 0.5$, then 
$\lceil \hat{c} \rceil$ cuts provides the largest possible sphere size. 
To prove these facts, the next lemma is helpful.

\begin{lemma}\label{r equivalence}
Recall $r \in \mathbb{Z}_{>0}$ and $\hat{c}=\sqrt{r}$. 
We have $r \;\le\; \lfloor \hat{c} \rfloor^2 + \lfloor \hat{c} \rfloor 
\text{ if and only if }
\hat{c} \;\le\;  \lfloor \hat{c} \rfloor + 0.5$.
\end{lemma}
\begin{proof}
We begin by showing that 
if $\hat{c} \;\le\; \lfloor \hat{c} \rfloor + 0.5$, then 
$ r \;\le\; \lfloor \hat{c} \rfloor^2 + \lfloor \hat{c} \rfloor$.
Assume that $\hat{c} \;\le\; \lfloor \hat{c} \rfloor + 0.5$. 
Squaring both sides, we have 
$
r \;\le\; \lfloor \hat{c} \rfloor^2 + \lfloor \hat{c} \rfloor + 0.25,
$
which implies that $
r \;\; \le \;\; \lfloor \hat{c} \rfloor^2 + \lfloor \hat{c} \rfloor.$

Next we will show that if 
$r \;\le\; \lfloor \hat{c} \rfloor ^2 + \lfloor \hat{c} \rfloor$, 
then 
$\hat{c} \;\le\; \lfloor \hat{c} \rfloor + 0.5$. 
We proceed by contrapositive. 
Suppose that 
$\hat{c} \;> \; \lfloor \hat{c} \rfloor + 0.5$. 
Squaring both sides, we have 
$
r \; > \; \lfloor \hat{c} \rfloor^2 + \lfloor \hat{c} \rfloor + 0.25, 
$
which implies that $r \;>\; \lfloor \hat{c} \rfloor^2 + \lfloor \hat{c} \rfloor$.
\end{proof}

Lemma \ref{r equivalence} means that if $\hat{c}$ 
is closer to $\lfloor \hat{c} \rfloor$ than it is to 
$\lceil \hat{c} \rceil$, then the inequality 
$r \le \lfloor \hat{c} \rfloor^2 + \lfloor \hat{c} \rfloor$
is satisfied. Otherwise, if $\hat{c}$ is closer to 
$\lceil \hat{c} \rceil$, then the opposite inequality, 
$r > \lfloor \hat{c} \rfloor^2 + \lfloor \hat{c} \rfloor$, is 
satisfied. Besides being useful to prove the following two 
lemmas, Lemma \ref{r equivalence} in conjunction with the 
next two lemmas allows us to easily determine the number of 
cuts that will yield the maximal sphere size for each $r$. 
This is explained after the statement of the next lemma. 

\begin{lemma}\label{floorrootr}
$\hat{c} \;\le\; \lfloor \hat{c} \rfloor + 0.5$ if and only if 
\begin{align} \label{eqn15}
%\overset{\lfloor \hat{c} \rfloor}{\underset{i=1}{\sum}}
%\left\lfloor
%\left(\frac{\mathbf{q}_{\lfloor \hat{c} \rfloor}(i)-2}{2}\right)^2
%\right\rfloor 
\psi(\mathbf{q}_{\lfloor \hat{c} \rfloor})
- 
%\overset{\lceil \hat{c} \rceil}{\underset{i=1}{\sum}}
%\left\lfloor
%\left(\frac{\mathbf{q}_{\lceil \hat{c} \rceil}(i)-2}{2}\right)^2
%\right\rfloor 
\psi(\mathbf{q}_{\lceil \hat{c} \rceil})
\;\;\; \le \;\;\; 
r-1
\end{align}
\end{lemma} 
The full proof can be found in the appendices. 
The crux of the argument is the same as 
that of Lemmas \ref{cplusone} and \ref{cminusone}.

%%%%%%%%%%%%%%%%%%%%%%%%%%%%%%%%%%%%%%%%%%%%%%%%%%%%%%%%%%%%%%%%%%%%%%%%
\begin{comment}
\begin{lemma}\label{ceilrootr}
Let $n,r,c \in \mathbb{Z}_{>0}$, $n/r = 2$, and $c<n$. 
If $\hat{c} \;\ >\; \lfloor \hat{c} \rfloor + 0.5$, then 
\begin{align}\label{eqn16}
\overset{\lfloor \hat{c} \rfloor}{\underset{i=1}{\sum}}
\left\lfloor
\left(\frac{\mathbf{q}_{\lfloor \hat{c} \rfloor}(i)-2}{2}\right)^2
\right\rfloor 
- 
\overset{\lceil \hat{c} \rceil}{\underset{i=1}{\sum}}
\left\lfloor
\left(\frac{\mathbf{q}_{\lceil \hat{c} \rceil}(i)-2}{2}\right)^2
\right\rfloor 
\;\;\; > \;\;\; 
r-1
\end{align}
\end{lemma} 
Again the full proof can be found in the appendices, 
but the crux of the argument is the same 
as that of Lemma \ref{cminusone}.\\ 
\end{comment}
%%%%%%%%%%%%%%%%%%%%%%%%%%%%%%%%%%%%%%%%%%%%%%%%%%%%%%%%%%%%%%%%%%%%%%%%

Lemma \ref{floorrootr} characterizes 
precisely when $\lfloor \hat{c} \rfloor$ 
or $\lceil \hat{c} \rceil$ cuts is optimal for maximizing sphere size. 
The lemmas imply, as mentioned previously, that whichever of 
$\lfloor \hat{c} \rfloor$ and 
$\lceil \hat{c} \rceil$ is closer to $\hat{c}$ 
is the optimal cut value. However, also as mentioned 
previously, Lemma \ref{r equivalence} allows us to 
easily determine which is optimal by simply looking at $r$. 
Notice that $\lfloor \hat{c} \rfloor^2 + \lfloor \hat{c} \rfloor$ is 
exactly half way between 
$\lfloor \hat{c} \rfloor^2$ and 
$\lfloor \hat{c} \rfloor^2 + 2 \lfloor \hat{c} \rfloor = 
\lceil \hat{c} \rceil -1$. Hence, Lemmas \ref{r equivalence} and 
\ref{floorrootr} imply that 
for $r$ closer to $\lfloor \hat{c} \rfloor^2$, that 
$\lfloor \hat{c} \rfloor$ cuts are better, but for $r$ closer 
to $\lceil \hat{c} \rceil^2$, that $\lceil \hat{c} \rceil$ cuts are better. 
For example, if $r=11$, then $\lfloor \hat{c} \rfloor ^2 = 9$ and 
$\lceil \hat{c} \rceil^2 = 16$. Since $11$ is closer to $9$, we know that 
$\lfloor \hat{c} \rfloor = 3$ cuts is optimal. Moreover, if 
$r \in \{9,10,11,12\}$, then $\lfloor \hat{c} \rfloor = 3$ 
cuts is optimal, but if $r \in \{13,14,15,16\}$, 
then $\lceil \hat{c} \rceil = 4$ cuts is optimal. 
\\

Returning to Theorem \ref{main2}, we can also 
easily obtain an upper bound on maximum sphere 
size. This is shown in the next lemma and corollary. 

\begin{lemma}\label{continuous} ~\\
%Let $c < n$. Then \\
%%%%%%%%%%%%%%%%
Let %$n,r,c \in \mathbb{Z}_{>0}$, $n/r =2$, and $c < n$. 
$c \in [n-1]$. Then 
\begin{align*}
%\overset{c}{\underset{i=1}{\sum}}
%\left\lfloor
%\left(\frac{\mathbf{q}_c(i)-2}{2}\right)^2
%\right\rfloor 
\psi(\mathbf{q}_c)
\;\;\;\ge\;\;\;
 c\left( \left( 
\frac{r}{c}-1
\right)^2 
-\frac{1}{4}\right).
\end{align*}
  %%%%%%%%%%%%%%
\end{lemma}

\begin{proof}
Suppose %$n,r,c \in \mathbb{Z}_{>0}$, $n/r =2$, and $c < n$, 
$c \in [n-1]$  and let 
$\mathbf{a} := (a(1), a(2), \dots, a(c)) \in \mathbb{R}^c_{>0}$ such that 
$\sum_{i=1}^{c}a(i) = n$. 
Note first that \\
%%%%%%%%%%%%%%
\begin{align}\label{eqn17}
\sum_{i=1}^c\left(
\frac{a(i)-2}{2} 
\right)^2 
\;\;\;=\;\;\;
\sum_{i=1}^c\left(
\frac{a(i)^2}{4} - a(i) + 1
\right) 
\;\;\;=\;\;\;
\frac{1}{4}\sum_{i=1}^c a(i)^2 
- n + c. 
\end{align}\\ 
%%%%%%%%%%%%
Note also that by applying the Cauchy-Schwarz 
inequality to $\mathbf{a}$ and $\mathbf{1} \in 
\mathbb{R}^c$, the all-$1$ vector of 
length $c$, we obtain \\
%%%%%%%%%%%
\begin{align}\label{eqn18}
\sum_{i=1}^c a(i)^2
\;\;\;\ge\;\;\; 
\frac{\sum_{i=1}^c a(i)^2}{c}
\;\;\;=\;\;\; 
\frac{n^2}{c} 
\;\;\;=\;\;\; 
\sum_{i=1}^c\left(
\frac{n}{c}
\right)^2 
  \end{align} \\
%%%%%%%%%%%%%%%%%%%%
Equation (\ref{eqn17}) and inequality (\ref{eqn18}) 
imply that choosing $\mathbf{a} = (n/c, n/c, \dots, n/c)$ 
minimizes the sum on the far left of 
Equation (\ref{eqn17}). 
The minimum of the left side of Equation (\ref{eqn17}) is less than or equal to 
$\sum_{i=1}^c((\mathbf{q}_c(i)-2)/2)^2$, 
with equality only holding when $n/c \in \mathbb{Z}$. 
Thus an application of  Remark \ref{floorcalc} completes the proof. 
\end{proof}

\begin{cor} \label{spherebound}
%%%%%%%%%%
%Let $n,r \in \mathbb{Z}_{>0}$, $n/r =2$, 
%and $\sigma \in \mathbb{S}_n$. 
Let $\mathbf{m}_\sigma^r$ be a binary multipermutation. 
Also define 
\\
\begin{align*}
U(r) \; \; := \; \;
1 + (n-1)^2 
 - \; 
\left(\rule{0cm}{.7cm}\right.
 (\hat{c}-1)(r-1) 
\; + \; (n-1)(r-1)  
 + \big(\hat{c}-1\big)
\left(
\bigg( 
\frac{r}{\hat{c}+1}-1
\bigg)^2 
-\frac{1}{4}
\right) 
\left)\rule{0cm}{.7cm}\right..
\end{align*} \\
Then 
\begin{align*}%\label{eqn19}
|S(\mathbf{m}_\sigma^r, 1)| 
\;\;\;<\;\;\; U(r). 
\end{align*}
%%%%%%%%
\end{cor}

\begin{proof}
The proof follows from Theorem \ref{main2}, 
Lemma \ref{continuous}, and the fact that 
$
\hat{c}-1 < \lfloor\hat{c}\rfloor \le \lceil\hat{c}\rceil < \hat{c}+1.
$
\end{proof}

The following table compares values from 
Corollary \ref{spherebound}
versus the size of the actual largest multipermutation
sphere for given values of $r$. 
The actual values of largest sphere sizes were calculated 
using Theorem \ref{main2}.

\begin{table}[h!]
\caption{Maximum sphere size verses bounded value}
\label{}
\begin{center}
\begin{tabular}{|c|c|c|c|}
\hline
r & Max sphere size (\ref{eqn14})& Inequality (\ref{eqn17})& ratio \tabularnewline
\hline 
$10$ & $148$ & $\sim168$  & $\sim .8819$\tabularnewline
\hline 
$100$ & $18,101$ & $\sim18,423$& $\sim .9825$  \tabularnewline
\hline 
$1000$ & $1,937,753$ & $\sim1,941,489$  &$\sim .9981$ \tabularnewline
\hline 
\end{tabular}
\end{center}
\end{table}

As the table suggests, the estimated value of the 
maximum sphere size obtained by applying 
Corollary \ref{spherebound} is asymptotically good. 
By asymptotically good we mean that the
 ratio between the true maximum sphere size, 
 $\underset{ \mathbf{m}_\sigma^r \in 
 \mathcal{M}_r(\mathbb{S}_n)}{\max} 
 |S(\mathbf{m}_\sigma^r,1)|$, 
 and the upper bound value, $U(r)$ from Corollary \ref{spherebound},  
approaches $1$ as $r$ approaches infinity. 
This can be confirmed by observing that  if 
\begin{align*}
L(r)  \;\; := \; \;
1 + (n-1)^2 
 - \; 
\left(\rule{0cm}{.7cm}\right.
 (\hat{c}+1)(r-1) 
\; + \; (n-1)(r-1)  
 + \big(\hat{c}+1\big)
\bigg( 
\frac{r}{\hat{c}-1}-\frac{1}{2}
\bigg)^2 
\left)\rule{0cm}{.7cm}\right.,
\end{align*} 
then 
$L(r) \;<\; \underset{ \mathbf{m}_\sigma^r \in 
\mathcal{M}_r(\mathbb{S}_n)}{\max} 
|S(\mathbf{m}_\sigma^r,1)|$. 
After making this observation, 
the Squeeze Theorem can then be applied with 
 $(\underset{ \mathbf{m}_\sigma^r \in 
\mathcal{M}_r(\mathbb{S}_n)}{\max} 
 |S(\mathbf{m}_\sigma^r,1)|)/U(r)$ 
 being squeezed between
$L(r)/U(r)$ and $U(r)/L(r)$. 
As before, in the definitions of both 
$U(r)$ and $L(r)$, we keep $n$ in favor of $2r$.

Finally, Corollary \ref{spherebound} 
can be applied to establish a new lower bound
on perfect single-error correcting MPC$(n,r)$'s 
in the binary case. 
%Specifically, if $C$ is a perfect 
%single-error correcting MPC$(n,r)$, then
%it must satisfy
%$|C| \ge n!/((r!)^2 (\max|S(\mathbf{m}_\sigma^r,\mathrm{d}_\circ^r=1)|)$.
%Corollary \ref{spherebound} 
It can also be applied to establish 
a new Gilbert Varshamov type lower bound. 
%Specifically, if 
%$C$ is an MPC$(n,r,d)$, then it must satisfy 
%$|C| \ge n!/((r!)^2 (\max|S(\mathbf{m}_\sigma^r,\mathrm{d}_\circ^r=1)|)^{d-1})$.
These two bounds are stated as the last two lemmas. 

\begin{lemma}\label{bound1} 
%Let $n,r \in \mathbb{Z}_{>0}$, $n/r =2$, and 
Let $C$ be a perfect 
single-error correcting MPC$(n,r)$. Also let 
$U(r)$ be defined as in Corollary \ref{spherebound}. Then 
%\small
\[
\frac{n!}{(r!)^2(U(r))} \;\; \le \;\; |C|.
\]
\normalsize
\end{lemma} 
\begin{proofoutline}
The proof follows from Corollary \ref{spherebound}.
\hfill $\square$
\end{proofoutline}

\begin{lemma}\label{bound2}
Let %$n,r \in \mathbb{Z}_{>0}$, $n/r =2$, and let
$C$ be an MPC$_\circ(n,r,d)$. Also let 
$U(r)$ be defined as in Corollary \ref{spherebound}. Then 
%\small
\[
\frac{n!}{(r!)^2(U(r))^{d-1}} \;\;\le\;\; |C|
\]
\normalsize
\end{lemma}
\begin{proofoutline}
The proof follows from Corollary \ref{spherebound}
and a standard Gilbert-Varshamov argument 
(see the proof of Lemma \ref{G-V bound}
for such an argument). 
\hfill $\square$
\end{proofoutline}

%%%%%%%%%%%%%%%%%%%%%%
%%%%%%%%%%%%%%%%%%%%%%
%%%%%%%%%%%%%%%%%%%%%%
\begin{comment}
It is worth noting that whenever $r\ge 4$, the previous 
lemma implies that the size of the largest non-binary 
sphere, i.e. $|S(\mathbf{m}_\omega^r, 1)|$ 
when $n/r \ne 2$, is larger than the maximum sphere size 
in the binary case. This means that the lower bounds 
on maximal code size of Lemmas \ref{perfect bound} and 
\ref{G-V bound} still hold whenever $r$ is at least $4$. 
%%%%%%%%%%%%%%%%%%%%%%
%%%%%%%%%%%%%%%%%%%%%%
%%%%%%%%%%%%%%%%%%%%%%
\end{comment}

%TAKE OUT, MOVE, OR REWRITE THE NEXT PARAGRAPH?
%In this section we began considering the problem of 
%calculating the size of $r$-regular Ulam spheres of 
%radius $1$.  We showed (Proposition \ref{ballcalc}) that for a given 
%$r$-regular multipermutation $\mathbf{m}_\sigma^r,$
%the problem essentially reduces to 
%calculating $AD(\mathbf{m}_\sigma^r)$ and $SD(\mathbf{m}_\sigma^r)$. 
%It remains also to consider arbitrary radii.  

%%%%%%%%%%%%%%%%CONCLUSION
\section{Conclusion} \label{conclusion}
This paper first considered and answered two questions. 
The first question concerned Ulam sphere sizes and the 
second concerned the possibility of perfect codes. 
It was shown that Ulam sphere sizes can be calculated 
explicitly for reasonably small radii using an application of the 
RSK-correspondence. 
It was then shown, partially using the
afforementioned sphere-calculation 
method, that nontrivial perfect Ulam permutation codes 
do not exist. These new results are summarized in Tables 
\ref{permspheresizes} and 
\ref{permcodelimits}, found in the introduction. 

Following the discussion of permutation codes, the 
multipermutation code case was considered next, 
and two more questions were addressed.
The third question of calculating $r$-regular Ulam spheres 
was addressed for the cases when the center is $\mathbf{m}_e^r$ 
or when the radius $t=1$. 
This lead to new upper and lower bounds 
on maximal code size, providing an answer to the fourth question. 
These new results are summarized in the Tables 
\ref{multipermspheresizes} and 
\ref{multipermcodesizes}, found in the introduction. 

Many remaining problems remain. 
One problem is to find a method 
for calculating $r$-regular Ulam spheres 
for more general parameters. Our current work 
began to show how to calculate sizes for 
any radius when the center is $\mathbf{m}_e^r$ (using 
Young tableaux) or for any center when 
the radius is 1, but not for general parameters. 
While we proved the nonexistence of nontrivial 
perfect Ulam permutation codes, it is unknown 
whether or not perfect multipermutation codes exist. 
A general formula for 
any center or radii, or at least 
bounds on general sphere sizes, would help 
in understanding bounds on the size of 
multipermutation Ulam codes 
and the possibility of perfect multipermutation codes.

% or


\appendices

\section{}\label{appendix A}
%\appendix
\textit{Proof of Remark \ref{n-l}:} \\
%\begin{proof}
Let %$n,r \in \mathbb{Z}_{>0},$ $r \vert n$, and
 $\mathbf{m}_\sigma^r, \mathbf{m}_\pi^r \in 
 \mathcal{M}_r(\mathbb{S}_n)$.
We will first show that 
$\mathrm{d}_\circ(\mathbf{m}_\sigma^r,\mathbf{m}_\pi^r) \ge n - 
\ell(\mathbf{m}_\sigma^r, \mathbf{m}_\pi^r)$.
By definition of $\mathrm{d}_\circ(\mathbf{m}_\sigma^r,\mathbf{m}_\pi^r),$ there 
exist $\sigma' \in R_r(\sigma)$ and 
$\pi' \in R_r(\pi)$ such that $\mathrm{d}_\circ(\mathbf{m}_\sigma^r,\mathbf{m}_\pi^r) 
= \mathrm{d}_\circ(\sigma', \pi') = n - \ell(\sigma', \pi').$ 
Hence if 
for all $\sigma' \in R_r(\sigma)$ and $\pi' \in R_r(\pi)$ 
we have $\ell(\sigma', \pi') \le 
\ell(\mathbf{m}_\sigma^r, \mathbf{m}_\pi^r)$, 
then $\mathrm{d}_\circ(\mathbf{m}_\sigma^r,\mathbf{m}_\pi^r) 
%\ge n - \ell(\sigma,\pi) 
\ge n - \ell(\mathbf{m}_\sigma^r, \mathbf{m}_\pi^r)$
(subtracting a larger value from $n$ 
results in a smaller overall value).
Therefore it suffices to show that 
that for all $\sigma' \in R_r(\sigma)$ and 
$\pi' \in R_r(\pi),$ that $\ell(\sigma', \pi') \le 
\ell(\mathbf{m}_\sigma^r, \mathbf{m}_\pi^r)$.  
This is simple to prove 
because if two permutations have a common subsequence, 
then their corresponding $r$-regular multipermutations 
will have a related common subsequence. 
Let $\sigma' \in R_r(\sigma)$, $\pi' \in R_r(\pi)$, and 
$\ell(\sigma', \pi') = k.$  Then there exist indexes 
$1 \le i_1 < i_2 < \dots < i_k \le n$ and 
$1 \le j_1 < j_2 < \dots < j_k \le n$ such that  
for all $p \in [k],$ $\sigma'(i_p) = \pi'(j_p).$  
Of course, whenever $\sigma'(i) = \pi'(j)$, then 
$\mathbf{m}_{\sigma'}^r(i) = \mathbf{m}_{\pi'}^r(j)$.
Therefore $\ell(\sigma', \pi') =k \le 
\ell(\mathbf{m}_{\sigma'}^r, \mathbf{m}_{\pi'}^r)
= \ell(\mathbf{m}_\sigma^r, \mathbf{m}_\pi^r).$
%\end{proof}
%\begin{lemma}\label{n-lB}

Next, we will show that $\mathrm{d}_\circ(\mathbf{m}_\sigma^r,\mathbf{m}_\pi^r) \le
n - \ell(\mathbf{m}_\sigma^r, \mathbf{m}_\pi^r).$
%\end{lemma}
%\begin{proof}
Note that
\begin{align*}
\mathrm{d}_\circ(\mathbf{m}_\sigma^r,\mathbf{m}_\pi^r) ~
&= ~ \underset{\sigma' \in R_r(\sigma), \pi' \in R_r(\pi)}{\min}
\mathrm{d}_\circ(\sigma', \pi') \\ 
&= ~ \underset{\sigma' \in R_r(\sigma), \pi' \in R_r(\pi)}{\min} 
(n - \ell(\sigma', \pi')) \\
&= ~ n - \underset{\sigma' \in R_r(\sigma), \pi' \in R_r(\pi)}{\max}
\ell(\sigma', \pi').
\end{align*}
Here if $\underset{\sigma' \in R_r(\sigma), \pi' \in R_r(\pi)}{\max} 
\ell (\sigma', \pi') \ge
\ell(\mathbf{m}_\sigma^r, \mathbf{m}_\pi^r)$, 
then 
$\mathrm{d}_\circ(\mathbf{m}_\sigma^r,\mathbf{m}_\pi^r) \le 
n - \ell(\mathbf{m}_\sigma^r, \mathbf{m}_\pi^r)$ 
(subtracting a smaller value from $n$ results in 
a larger overall value). 
It is enough to show that 
 there exist $\sigma' \in R_r(\sigma)$ and 
$\pi' \in R_r(\pi)$ such that 
$\ell(\sigma', \pi') \ge 
\ell(\mathbf{m}_\sigma^r, \mathbf{m}_\pi^r)$. 
To prove this fact, we take a longest common subsequence 
of $\mathbf{m}_\sigma^r$ and $\mathbf{m}_\pi^r$ and then 
carefully choose $\sigma' \in R_r(\sigma)$ and 
$\pi' \in R_r(\pi)$ to have an equally long common subsequence.  
The next paragraph describes how this can be done. 

Let $\ell(\mathbf{m}_\sigma^r, \mathbf{m}_\pi^r) = k$ and  
let $(1 \le i_1 < i_2 < \dots < i_k \le n)$ and 
$(1 \le j_1 < j_2 < \dots < j_k \le n)$ be integer sequences 
such that for all $p \in [k],$ 
$\mathbf{m}_\sigma^r(i_p) = \mathbf{m}_\pi^r(j_p).$ 
The existence of such sequences is guaranteed 
by the definition of $\ell(\mathbf{m}_\sigma^r, \mathbf{m}_\pi^r).$ 
Now for all $p \in [k],$ define 
$\sigma'(i_p)$ to be the smallest integer $l \in [n]$ 
such that $\mathbf{m}_\sigma(l) = \mathbf{m}_\sigma(i_p)$ 
and if $q \in [k]$ with $q < p,$ then 
$\mathbf{m}_\sigma^r(i_q) = \mathbf{m}_\pi^r(i_p)$ implies 
$\sigma'(i_q) < \sigma'(i_p) = l.$  
For all $p \in [k],$ define $\pi(j_p)$ similarly.  
Then for all $p \in [k],$ $\sigma'(i_p) = \pi'(j_p).$
The remaining terms of $\sigma'$ and $\pi'$ may 
easily be chosen in such a manner that 
$\sigma' \in R_r(\sigma)$ and $\pi' \in R_r(\pi).$
Thus there exist $\sigma' \in R_r(\sigma)$ and 
$\pi' \in R_r(\pi)$ such that 
$\ell(\sigma',\pi') \ge 
\ell(\mathbf{m}_\sigma^r, \mathbf{m}_\pi^r)$. 
\hfill $\square$ 
%\end{proof}

\section{}\label{appendix B}
\textit{Proof of Remark \ref{translocations}:} \\
%\begin{proof}
Suppose %$n,r \in \mathbb{Z}_{>0}$, $r \vert n$, and 
 $\mathbf{m}_\sigma^r, \mathbf{m}_\pi^r \in 
 \mathcal{M}_r(\mathbb{S}_n)$.
There exists a translocation $\phi \in \mathbb{S}_n$ 
such that $\ell(\mathbf{m}_\sigma^r \cdot\phi, \mathbf{m}_\pi^r) 
= \ell(\mathbf{m}_\sigma^r, \mathbf{m}_\pi^r) + 1$, 
since it is always possible to arrange one element with a 
single translocation.  This then implies that 
%\begin{align*}
$\min \{k \in \mathbb{Z} \;:\; \text{ there exists } $
$(\phi_1, \dots, \phi_k) \; \text{ such that }\; 
\mathbf{m}_\sigma^r \cdot \phi_1 \cdots \phi_k = \mathbf{m}_\pi^r\} 
\le n - \ell(\mathbf{m}_\sigma^r, \mathbf{m}_\pi^r) = 
\mathrm{d}_\circ(\mathbf{m}_\sigma^r,\mathbf{m}_\pi^r).$
%\end{align*}
At the same time, given 
$\ell (\mathbf{m}_\sigma^r, \mathbf{m}_\pi^r) \le n,$ 
then for all translocations $\phi \in \mathbb{S}_n,$ 
we have that 
$\ell (\mathbf{m}_\sigma^r \cdot \phi, \mathbf{m}_\pi^r) 
\le \ell (\mathbf{m}_\sigma^r, \mathbf{m}_\pi^r) + 1$, 
since a single translocation can only arrange one 
element at a time.  Therefore by Remark \ref{n-l}, 
%\begin{align*}
$\min \{k \in \mathbb{Z} \;:\; \text{ there exists } (\phi_1, \dots, \phi_k) \text{ s.t } 
\mathbf{m}_\sigma^r \cdot \phi_1 \cdots \phi_k = \mathbf{m}_\pi^r\} 
\; \ge \; n - \ell(\mathbf{m}_\sigma^r, \mathbf{m}_\pi^r) = 
\mathrm{d}_\circ(\mathbf{m}_\sigma^r,\mathbf{m}_\pi^r)$.
%\end{align*}
%\end{proof}
\hfill $\square$

\section{}\label{appendix C}
\textit{Proof of Lemma \ref{gsequence}:} \\
Recall that $n$ is an even integer. 
Assume %$n$ is a positive even integer, 
$c \in [n-1]$ and 
$\mathbf{q} := (q(1), q(2), \dots, q(c))\in \mathbb{Z}_{>0}^c$ such that 
  $\sum_{i=1}^c q(i) = n$. 
  For the first direction, suppose there exists some $i \in [c]$ such that 
  $q(i)$ is even. Since $n$ is even, the number of odd 
  values in $\mathbf{q}$ is even, i.e. 
  $\#\{ i \in [c] \;:\; q(i) \text{ is odd} \}=2k$ for 
some nonnegative integer $k$. We will now construct 
an $\mathbf{m} \in \{1,2\}^n$ with an equal number of 
$1$'s and $2$'s, whose cuts correspond to $\mathbf{q}$. 
We begin by defining two sets: first 
$\{q^*(1),q^*(2), \dots,q^*(2k)\} := \{q(i) \;:\; q(i) \text{ is odd}\}$ 
and then 
$\{q^*(2k+1), q^*(2k+2), \dots q^*(c)\} := \{q(i) \;:\; q(i) \text{ is even}\}$.  
Then define $\mathbf{m}$ as follows: 
\begin{eqnarray*}
\mathbf{m} &:=& (
\underset{g^*(1)}{\underbrace{\mathbf{m}[a_1,b_1]}}, 
\underset{g^*(2)}{\underbrace{\mathbf{m}[a_2,b_2]}}, 
\dots,
\underset{g^*(k)}{\underbrace{\mathbf{m}[a_k,b_k]}}, \;
\underset{g^*(2k+1)}{\underbrace{\mathbf{m}[a_{k+1},b_{k+1}]}},\\
& &
\underset{g^*(k+1)}{\underbrace{\mathbf{m}[a_{k+2},b_{k+2}]}}, 
\underset{g^*(k+2)}{\underbrace{\mathbf{m}[a_{k+3},b_{k+3}]}}, 
\dots,
\underset{g^*(2k)}{\underbrace{\mathbf{m}[a_{2k+1},b_{2k+1}]}}, \;
\underset{g^*(2k+2)}{\underbrace{\mathbf{m}[a_{2k+2},b_{2k+2}]}}
\dots,
\underset{g^*(c)}{\underbrace{\mathbf{m}[a_c,b_c]}})
\end{eqnarray*}
where $\mathbf{m}(1) = 1$ and for all $j \in [c]$, $\mathbf{m}[a_j,b_j]$ is a cut. 
  
The idea here is simple. By the definition of $\mathbf{m}$, 
each of the first $k$ cuts begin and end with $1$ since they are 
all odd length cuts, and thus 
each will has one more $1$ than $2$. 
The $(k+1)$th cut, $\mathbf{m}[a_{k+1},b_{k+1}]$,
 is taken to be of even length, which 
reverses the order of the subsequent $k$ cuts. 
Hence the $k$ cuts from $\mathbf{m}[a_{k+2},b_{k+2}]$ 
through $\mathbf{m}[a_{2k+1},b_{2k+1}]$ each begin and 
end with $2$, so each will have one more $2$ than $1$. 
The remaining cuts from $\mathbf{m}[a_{k+2},b_{k+2}]$
through $\mathbf{m}[a_c,b_c]$ are even, which implies 
that the number of $1$'s and $2$'s in each of these 
cuts is equal.  Hence, we may write the following: 
\begin{eqnarray*}
\#\{i \in [n] \;:\; \mathbf{m}(i) = 1\} &=& 
\; \; \; \; \#\{ i \in [a_1,b_k] \;:\; \mathbf{m}(i) = 1\}  \\
&&+ \; \#\{i \in [a_{k+2},b_{2k+1}] \;:\; \mathbf{m}(i) = 1\}  \\ 
&&+ \; \#\{i \in [a_{k+1},b_{k+1}] \cup [a_{2k+2},b_c] \;:\; \mathbf{m}(i) = 1\} \\  \\
&=& 
\; \; \; \; \#\{ i \in [a_1,b_k] \;:\; \mathbf{m}(i) = 2\} + k \\
&&+ \; \#\{i \in [a_{k+2},b_{2k+1}] \;:\; \mathbf{m}(i) = 2\} -k \\ 
&&+ \; \#\{i \in [a_{k+1},b_{k+1}] \cup [a_{2k+2},b_c] \;:\; \mathbf{m}(i) = 2\} \\ \\
&=& 
\; \; \; \; \#\{i \in [n] \;:\; \mathbf{m}(i) = 2\}.
\end{eqnarray*}
Applying Lemma \ref{equalonestwos} completes the first direction. 

For the second direction, let 
$\mathbf{m}_\sigma \in \mathcal{M}_r(\mathbb{S}_n)$ 
be a binary multipermutation 
such that 
\[
 \mathbf{m}_\sigma^r = 
 (
{\mathbf{m}_\sigma^r[a_1,b_1]}, 
{\mathbf{m}_\sigma^r[a_2,b_2]}, 
  \dots 
{\mathbf{m}_\sigma^r[a_c,b_c]}
  ),
\]
with each $\mathbf{m}_\sigma^r[a_j,b_j]$ a 
cut whose length 
corresponds to an element of $\mathbf{q}$. 
We want to show then that there exists some 
$i \in [c]$ such that $\mathbf{q}(i)$ is even. 
We proceed by contradiction. 

Suppose that there does not exist an $i \in [c]$ such that 
$\mathbf{q}(i)$ is even. Then for all $j \in [c]$,
$\mathbf{m}_\sigma^r[a_j,b_j]$ is odd. This then implies that 
for each cut 
$\mathbf{m}_\sigma^r[a_j,b_j]$, that 
$\mathbf{m}_\sigma^r(a_j) = \mathbf{m}_\sigma^r(b_j)$. 
In other words, an odd-length cut
necessarily begins and ends with the same element. 
However, to end one cut and begin another, 
it is necessary to repeat a digit. 
This implies that for each $j \in [c]$, and cut
$\mathbf{m}_\sigma^r[a_j,b_j]$, 
\[
\#\{i \in [a_j,b_j] \;:\; \mathbf{m}_\sigma^r(i) = \mathbf{m}_\sigma^r(1)\}
= \#\{ i \in [a_j,b_j] \;:\; \mathbf{m}_\sigma^r(i) \ne \mathbf{m}_\sigma^r(1)\} + 1,
\]
which in turn implies that 
the total number of elements of $\mathbf{m}_\sigma^r$ that equal 
to $\mathbf{m}_\sigma^r(1)$ is exactly $c$ more than the number of 
elements not equal to $\mathbf{m}_\sigma^r(1)$, 
contradicting Lemma \ref{equalonestwos}. 
\hfill $\square$

\section{}\label{appendix D}
\begin{remark}\label{floorcalc}
For any integer $a \in \mathbb{Z},$ it is easily verified that \\
1) If $a$ is even, then $\lfloor (a/2)^2\rfloor = (a/2)^2$\\
2) If $a$ is odd, then $\lfloor (a/2)^2 \rfloor = (a/2)^2 - 1/4$. 
\end{remark}

\begin{remark}\label{floordiff}
Let $a \in \mathbb{Z}$. 
Then the following is a direct consequence of the previous remark.\\
1) If $a$ is even, then 
$\lfloor (a/2)^2\rfloor  - 
\lfloor ((a-1)/2)^2\rfloor = (a/2) $. \\
2) If $a$ is odd, then 
$\lfloor (a/2)^2\rfloor  - 
\lfloor ((a-1)/2)^2\rfloor = (a/2) - 1/2$.
\end{remark}

\begin{remark}\label{floordiff2}
Let $a,b \in \mathbb{Z}_{\ge1}$ such that $a -b \ge 2$. Then 
\begin{eqnarray}\label{eqn3}
\left\lfloor \left(  \frac{a-2}{2}   \right)^2 \right\rfloor + 
\left\lfloor \left(   \frac{b-2}{2}  \right)^2 \right\rfloor
 &\ge&
\left\lfloor \left(  \frac{(a-1)-2}{2}   \right)^2 \right\rfloor + 
\left\lfloor \left(   \frac{(b+1)-2}{2}  \right)^2 \right\rfloor,
 \end{eqnarray}
\end{remark}
with equality holding only if $a-b = 2$ and both 
$a$ and $b$ are odd. 

\begin{proof}
Assume that $a,b \in \mathbb{Z}_{\ge1}$ and $a-b\ge2$. 
Then inequality (\ref{eqn3}) holds if and only if the 
following inequality also holds.
\begin{eqnarray}\label{eqn4}
\left\lfloor \left(  \frac{a-2}{2}   \right)^2 \right\rfloor -
\left\lfloor \left(  \frac{(a-1)-2}{2}   \right)^2 \right\rfloor 
-
\left(
\left\lfloor \left(   \frac{(b+1)-2}{2}  \right)^2 \right\rfloor
-
\left\lfloor \left(   \frac{b-2}{2}  \right)^2 \right\rfloor
\right) \ge 0.
\end{eqnarray}
Applying Remark \ref{floordiff} to the four cases when 
$a$ is either even or odd and $b$ is either even or odd, 
a routine calculation shows the following:  
If both $a$ and $b$ are even, then the left side of 
inequality (\ref{eqn4}) equals $(a-b)/2 > 0$. If $a$ is even 
and $b$ is odd or if $a$ is odd and $b$ is even, then the 
left side of inequality (\ref{eqn4}) equals $(a-b-1)/2 >0$. 
Finally, if both $a$ and $b$ are odd, then the left 
side of inequality (\ref{eqn4}) equals $(a-b-2)/2 \ge 0$. 
\end{proof}

\section{}\label{appendix E} 
\textit{Proof of Lemma \ref{cplusone}:}  \\
%\begin{proof}
Assume $c \in [n-2]$ 
and $\hat{c} \le c$.
We will split the proof into two cases, 
when $q_{c+1} \le c$ and when $q_{c+1} > c$. 
Let us first suppose $q_{c+1} \le c$. This corresponds 
roughly to the case where $\hat{c} \le c \le \sqrt{n}$. 
In this instance, for all $i \in [2,c]$, we have 
$\mathbf{q}_c(i) - \mathbf{q}_{c+1}(i) \le 1$ and 
$\mathbf{q}_c(1) - \mathbf{q}_{c+1}(1) \le 2$. 
This is because $\mathbf{q}_{c+1}$ can be 
constructed from $\mathbf{q}_c$ by shortening each 
cut $\mathbf{q}_c$ in order to create the $(c+1)$st cut. 
Since $q_{c+1} \le c$, 
each cut will decrease by at most one, 
except for one exceptional case when 
$q_c$ is an odd integer and $q_{c+1} = c$, in 
which case $\mathbf{q}_c(1) - \mathbf{q}_{c+1}(1) = 2$. 
%Figures \ref{excplusone1} and \ref{excplusone2} 
%explain the construction of $\mathbf{q}_{c+1}$ 
%from $\mathbf{q}_c$. 

Since for all $i \in [2,c]$, 
$\mathbf{q}_c(i) - \mathbf{q}_{c+1}(i) \le 1$ and 
$\mathbf{q}_c(1) - \mathbf{q}_{c+1}(1) \le 2$, 
the left hand side of inequality (\ref{eqn8}) 
is less than or equal to 
%%%%%%%%%%%%%%%%%%
\begin{align}\label{eqn6}
%\overset{c}{\underset{i=1}{\sum}}
%\left\lfloor\left( 
%\frac{\mathbf{q}_{c}(i)-2}{2}
%\right)^2\right\rfloor 
\psi(\mathbf{q}_c) 
\;-\;
%\left\lfloor\left(
%\frac{\mathbf{q}_c(1)-4}{2}
%\right)^2\right\rfloor
\psi(\mathbf{q}_c(1)-2)
\;-\;
%\overset{c}{\underset{i=2}{\sum}}
%\left\lfloor\left( 
%\frac{\mathbf{q}_{c}(i)-3}{2}
%\right)^2\right\rfloor 
\psi(\mathbf{q}_c(2)-1,\mathbf{q}_c(3)-1,\dots, \mathbf{q}_c(c)-1)
\;-\;
%\left\lfloor\left(
%\frac{\mathbf{q}_{c+1}(c+1)-2}{2}
%\right)^2\right\rfloor. 
\psi(\mathbf{q}_{c+1}(c+1)).
\end{align}
By adding and subtracting 
%$\lfloor\left(
%(\mathbf{q}_{c}(1)-3) / 2
%\right)^2\rfloor$ 
$\psi(\mathbf{q}_c(1)-1)$ 
from expression (\ref{eqn12})
and also disregarding the last term, 
$
%- \lfloor\left(
%({\mathbf{q}_{c+1}(c+1)-2})/{2}
%\right)^2\rfloor
-\psi(\mathbf{q}_{c+1}(c+1))
$,
after some rearrangement 
expression (\ref{eqn12}) 
is less than or equal to 
\begin{align*}
%\overset{c}{\underset{i=1}{\sum}}
%\left\lfloor\left( 
%\frac{\mathbf{q}_{c}(i)-2}{2}
%\right)^2\right\rfloor 
\psi(\mathbf{q}_c)
-
%\overset{c}{\underset{i=1}{\sum}}
%\left\lfloor\left( 
%\frac{\mathbf{q}_{c}(i)-3}{2}
%\right)^2\right\rfloor 
\psi(\mathbf{q}_c-1)
+
%\left\lfloor\left(
%\frac{\mathbf{q}_c(1)-3}{2}
%\right)^2\right\rfloor
\psi(\mathbf{q}_c(1)-1)
-
%\left\lfloor\left(
%\frac{\mathbf{q}_c(1)-4}{2}
%\right)^2\right\rfloor,
\psi(\mathbf{q}_c(1)-2),
\end{align*}
which by Remark \ref{floordiff}
is less than or equal to 
\begin{align*}
\overset{c}{\underset{i=1}{\sum}}
\left(\frac{\mathbf{q}_c(i)}{2} -1\right)
+
\left(
\frac{\mathbf{q}_c(1)-1}{2}-1
\right)
&\; \; \; \le \; \; \; 
\left(\frac{n}{2} - c \right) 
+ 
\left(
\frac{(q_c+1)-1}{2}-1
\right) \\~\\
&\; \; \; = \; \; \; \;  
r - c + \frac{q_c}{2} -1, 
\end{align*}
where this last expression is less than or equal 
to $r-1$ since $q_{c+1} \le c$ implies that 
$q_c \;\le\; q_{c+1} +1 \;\le\; 2c$. 
%$q_c = (n-{rem_c})/c \le n/c \le (2r)/\hat{c} = 2\hat{c} \le 2c$. 
This concludes the case where $q_{c+1} \le c$. \\

Next assume that $q_{c+1} > c$. This corresponds roughly 
to the case where $\sqrt{n} < c < n-1$. 
Note that 
\[
q_{c+1} = \frac{n-rem_{c+1}}{c+1} \;\;<\;\;
\frac{n}{c} \;\;\le\;\; \frac{2r}{\hat{c}} 
\;\;=\;\; 2\hat{c} \le 2c.
\]
Since $q_{c+1}$ is strictly less than $2c$, for all 
$i \in [c]$ we have $\mathbf{q}_c(i) - \mathbf{q}_{c+1}(i) 
\le 2$. Moreover, because $c < q_{c+1} < 2c$, 
the number of $i \in [c]$ such that 
$\mathbf{q}_{c}(i) - \mathbf{q}_{c+1}(i) = 2$ 
is $q_{c+1} - c$. This means then that the 
number of $i \in [c]$ such that 
$\mathbf{q}_c(i) - \mathbf{q}_{c+1}(i) = 1$ 
is equal to $c - (q_{c+1}-c) $. 
Similarly to before, the reasoning for these set sizes 
comes from constructing $\mathbf{q}_{c+1}$ from 
$\mathbf{q}_{c}$ and considering how much each cut 
length is decreased to construct the final cut, 
$\mathbf{q}_{c+1}(c+1)$. 
Example (\ref{example2}) helps here to aid comprehension. 

It should be noted that there is one exceptional case, 
when $q_{c+1}$ is an odd integer and $rem_{c+1} = 0$. In 
this instance, $\mathbf{q}_{c+1}(c+1) = q_{c+1}-1$, which 
means the number of $i \in [c]$ such that 
$\mathbf{q}_c(i) - \mathbf{q}_{c+1}(i) = 2$ is decreased by 
one, while the number of $i \in [c]$ such that 
$\mathbf{q}_{c}(i) - \mathbf{q}_{c+1}(i) = 1$ 
is increased by one. It is easily 
shown that the final effect on the size of the left hand side of 
inequality (\ref{eqn8}) is a decrease, so it is enough to 
prove the inequality in the typical case,  
when 
$\#\{i \in [c] \;:\; \mathbf{q}_{c}(i) - \mathbf{q}_{c+1}(i) = 2\}
= q_{c+1} - c$ and 
$\#\{i \in [c] \;:\; \mathbf{q}_{c}(i) - \mathbf{q}_{c+1}(i) = 1\}
= c - (q_{c+1} - c)$. These set sizes imply that the left 
hand side of inequality (\ref{eqn11}) is less than or equal to 
%%%%%%%%%%%%%%%%
\begin{align*}%\label{eqn7}
%\overset{c}{\underset{i=1}{\sum}}
%\left\lfloor\left( 
%\frac{\mathbf{q}_{c}(i)-2}{2}
%\right)^2\right\rfloor 
\psi(\mathbf{q}_c)
&\;-\;\;
%\overset{q_{c+1}-c}{\underset{i=1}{\sum}}
%\left\lfloor\left( 
%\frac{\mathbf{q}_{c}(i)-4}{2}
%\right)^2\right\rfloor 
\psi(\mathbf{q}_c(1)-2, \mathbf{q}_c(2)-2, \dots, \mathbf{q}_c(q_{c+1}-c)-2)\\
&\;-\;\;
%\overset{c}{\underset{i=q_{c+1}-c+1}{\sum}}
%\left\lfloor\left( 
%\frac{\mathbf{q}_{c}(i)-3}{2}
%\right)^2\right\rfloor 
\psi(\mathbf{q}_c(q_{c+1}-c+1)-1, 
  \mathbf{q}_c(q_{c+1}-c+2)-1, \dots,
  \mathbf{q}_c(c)-1 )
)
\;\;-\;\;
%\left\lfloor\left(
%\frac{\mathbf{q}_{c+1}(c+1)-2}{2}
%\right)^2\right\rfloor. 
\psi(\mathbf{q}_{c+1}(c+1)).
\end{align*}
By adding and subtracting 
$\sum_{i=1}^{q_{c+1}-c}
\lfloor \left(
(\mathbf{q}_c(i)-3)/2
\right)^2\rfloor$, 
after some rearrangement 
expression (\ref{eqn13}) can be rewritten as 
%%%%%%%
\begin{align*}
&
%\overset{c}{\underset{i=1}{\sum}}
%\left\lfloor\left( 
%\frac{\mathbf{q}_{c}(i)-2}{2}
%\right)^2\right\rfloor
\psi(\mathbf{q}_c)
\;\;\;+\;\;\;
%\overset{q_{c+1}-c}{\underset{i=1}{\sum}}
%\left\lfloor\left( 
%\frac{\mathbf{q}_{c}(i)-3}{2}
%\right)^2\right\rfloor
\psi(\mathbf{q}_c(1)-1, \mathbf{q}_c(2)-1, \dots, 
\mathbf{q}_c(q_{c+1}-c)-1) 
\;\;\;-\;\;\;
%\left\lfloor\left(
%\frac{\mathbf{q}_{c+1}(c+1)-2}{2}
%\right)^2\right\rfloor \\
\psi(\mathbf{q}_{c+1}(c+1)) \\
%%%%%
-\;\;\;
&
%\overset{c}{\underset{i=1}{\sum}}
%\left\lfloor\left( 
%\frac{\mathbf{q}_{c}(i)-3}{2}
%\right)^2\right\rfloor 
\psi(\mathbf{q}_c-1)
\;\;\;-\;\;\;
%\overset{q_{c+1}-c}{\underset{i=1}{\sum}}
%\left\lfloor\left( 
%\frac{\mathbf{q}_{c}(i)-4}{2}
%\right)^2\right\rfloor,
\psi(\mathbf{q}_c(1)-2, \mathbf{q}_c(2) -2, \dots, 
\mathbf{q}_c(q_{c+1}-c)-2)
\end{align*}
%%%%%%%%%%%
which by Remark \ref{floordiff}
is less than or equal to 
%%%%%
\begin{align*}
&\overset{c}{\underset{i=1}{\sum}}
\left(\frac{\mathbf{q}_c(i)}{2} -1\right)
\;\;\;+\;\;\; %&+& \;\;\;
\overset{q_{c+1}-c}{\underset{i=1}{\sum}}
\left( 
\frac{\mathbf{q}_{c}(i)-1}{2} - 1
\right)
\;\;\;-\;\;\;%&-& \;\;\;
%\left\lfloor\left(
%\frac{\mathbf{q}_{c+1}(c+1)-2}{2}
%\right)^2\right\rfloor 
\psi(\mathbf{q}_{c+1}(c+1))
& \\~\\
\le \;\;\;&
\left(\frac{n}{2}- c\right)
\;\;\;+\;\;\;%&+&\;\;\;
\left(q_{c+1}-c\right)\left(\frac{(q_c+1)-1}{2} - 1\right)
\;\;\;-\;\;\;%&-&\;\;\;
\left(\left(\frac{q_{c+1}-2}{2}\right)^2 - \frac{1}{4}\right)& \\~\\
\le\;\;\; &
(r -c) 
\;\;\;+\;\;\;%&+& \;\;\;
(q_c - c - 1)\left( \frac{q_c}{2} - 1\right) 
\;\;\;-\;\;\;%&-& \;\;\; 
\left(\left(\frac{q_c-3}{2}\right)^2 - \frac{1}{4})\right),& 
\end{align*}
which reduces to 
$r + {q_c^2}/{4} - {(cq_c)}/{2} - 1$. Because of the fact that 
$q_{c+1} < 2c$ implies $q_c \;\le\; q_{c+1} + 1 \;\le\; 2c$,
this final expression 
is guaranteed to be less than or equal to $r-1$, completing the proof. 
%\end{proof}

\section{}\label{appendix F} 
\textit{Proof of Lemma \ref{cminusone}:}  \\
%\begin{proof}
Assume 
$c \in [n-1]$ and 
$c \le \hat{c}$. 
The left hand side of inequality 
(\ref{eqn8}) depends on the difference 
between $\mathbf{q}_{c-1}(i)$ and 
$\mathbf{q}_c(i)$ for each $i \in [c-1]$. 
We wish to show that these differences are 
sufficiently large to cause inequality (\ref{eqn8}) to be satisfied. 
Note that 
\[
q_c \;\;=\;\; \frac{n-rem_c}{c} \;\;>\;\; \frac{n-c}{c} \;\;=\;\; 
\frac{n}{c}-1 \;\;>\;\; \frac{n}{c^2}(c-1).
\]
Since $c \le \hat{c}$, we have $n/c^2 = 2r/c^2 \ge 2$, 
which implies that $q_c > 2(c-1)$. 
Therefore, for all $i \in [c-1]$, we have 
$\mathbf{q}_{c-1}(i) - \mathbf{q}_c(i) \ge 2$. 
This is because $\mathbf{q}_c$ can be constructed 
from $\mathbf{q}_{c-1}$ by shortening each cut of 
$\mathbf{q}_{c-1}$ in order to create the 
$c$th cut of $\mathbf{q}_c$, whose length 
is at least $q_c-1$. 
Example \ref{example1} helps comprehension here.

Next, let $k := \#\{i \in [c-1] \;:\; \mathbf{q}_{c-1}(i) 
- \mathbf{q}_c(i) >2\}$. In other words, $k$ is the number 
of cuts in $\mathbf{q}_{c-1}$ that are decreased by more 
than $2$ in the construction of $\mathbf{q}_c$ from $\mathbf{q}_{c-1}$. 
Notice that if $\mathbf{q}_{c-1}(i) - \mathbf{q}_{c}(i) = 2$ then 
by Remark \ref{floorcalc}, 
\[
%\left\lfloor 
%\left(\frac{\mathbf{q}_{c-1}(i)-2}{2}\right)^2
%\right\rfloor 
\psi(\mathbf{q}_{c-1})
-
%\left\lfloor 
%\left(\frac{\mathbf{q}_{c}(i)-2}{2}\right)^2
%\right\rfloor 
\psi(\mathbf{q}_c)
\;\;\; = \;\;\;
\left(\frac{\mathbf{q}_{c-1}(i)-2}{2}\right)^2
-
\left(\frac{\mathbf{q}_{c}(i)-2}{2}\right)^2.
\] 
Moreover, Remark \ref{floorcalc} also implies that in all 
other instances, 
\[
%\left\lfloor 
%\left(\frac{\mathbf{q}_{c-1}(i)-2}{2}\right)^2
%\right\rfloor 
\psi(\mathbf{q}_{c-1}
-
%\left\lfloor 
%\left(\frac{\mathbf{q}_{c}(i)-2}{2}\right)^2
%\right\rfloor 
\psi(\mathbf{q}_c)
\;\;\; \ge \;\;\;
\left(\frac{\mathbf{q}_{c-1}(i)-2}{2}\right)^2
-
\left(\frac{\mathbf{q}_{c}(i)-2}{2}\right)^2 - \frac{1}{4}.
\] 
Hence the left hand side of inequality (\ref{eqn8})
is greater than or equal to 
\begin{align}\label{eqn9}
\overset{c-1}{\underset{i=1}{\sum}}
\left( 
\frac{\mathbf{q}_{c-1}(i)-2}{2}
\right)^2
-
\overset{c-1}{\underset{i=1}{\sum}}
\left( 
\frac{(\mathbf{q}_{c}(i)-2}{2}
\right)^2
- \frac{k}{4}
- 
\left(
\frac{\mathbf{q}_c(c)-2}{2}
\right)^2.
\end{align}

At this point, we split the remainder of the proof into two 
possibilities, the general case where $\mathbf{q}_c(c) = q_c$ 
and the exceptional case where $\mathbf{q}_c(c) = q_c-1$, 
which only occurs when $n/c$ is an odd integer. 
We will treat the general case first and then end with some 
comments about the exceptional case. 
Since we are first assuming that $\mathbf{q}_c(c) = q_c$, 
expression (\ref{eqn9}) is equal to
\begin{align*}
&\overset{c-1}{\underset{i=1}{\sum}}
\left( 
\frac{\mathbf{q}_{c-1}(i)-2}{2}
\right)^2
-
\overset{c-1}{\underset{i=1}{\sum}}
\left( 
\frac{(\mathbf{q}_{c}(i)-2}{2}
\right)^2
- \frac{k}{4}
- 
\left(
\frac{q_c-2}{2}
\right)^2& \\ ~\\
= \;\;\;&
\overset{c-1}{\underset{i=1}{\sum}}
\left( 
\frac{\mathbf{q}_{c-1}^2(i)}{4} - \mathbf{q}_{c-1} + 1
\right)
- 
\overset{c-1}{\underset{i=1}{\sum}}
\left( 
\frac{\mathbf{q}_{c}^2(i)}{4} - \mathbf{q}_{c} + 1
\right)
- \frac{k}{4}
-
\frac{{g}_c^2}{4} + q_c - 1
& \\~\\
= \;\;\;&
\overset{c-1}{\underset{i=1}{\sum}}
\left(
\frac{\mathbf{q}_{c-1}^2(i)}{4}
\right)
-n 
+ (c-1)
-
\left[
\overset{c-1}{\underset{i=1}{\sum}}
\left(
\frac{\mathbf{q}_{c}^2(i)}{4}
\right)
- (n-q_c) + (c-1) 
\right]
- \frac{k}{4} 
- \frac{q_c^2}{4} + q_c -1
& \\~\\
=\;\;\; &
\frac{1}{4} 
\overset{c-1}{\underset{i=1}{\sum}}
\left(
\mathbf{q}_{c-1}^2(i) - \mathbf{q}_c^2(i) 
\right)
- \frac{k}{4} - \frac{q_c^2}{4} - 1. \numberthis \label{eqn10}
&
\end{align*}

From here, we focus on the summation 
$\sum_{i=1}^{c-1}
(\mathbf{q}_{c-1}^2(i) - 
\mathbf{q}_{c}^2(i))$
to prove that the overall expression is sufficiently large. 
The summation can be viewed as the sum of all shaded areas 
in Figure \ref{figure1}. 
In the figure, squares of area $\mathbf{q}_{c-1}^2(i)$ 
(with $i \in [c-1]$), are placed along the 
diagonal of an $n$-by-$n$ square. 
Within the bottom left corner of each of these squares
is placed another square of area $\mathbf{q}_{c}^2(i)$
(again $i \in [c-1]$). 
By carefully examining the total area of all shaded 
regions in the figure, we can lower bound 
$\sum_{i=1}^{c-1}
(\mathbf{q}_{c-1}^2(i) - 
\mathbf{q}_{c}^2(i))$
to satisfy the desired inequality. 

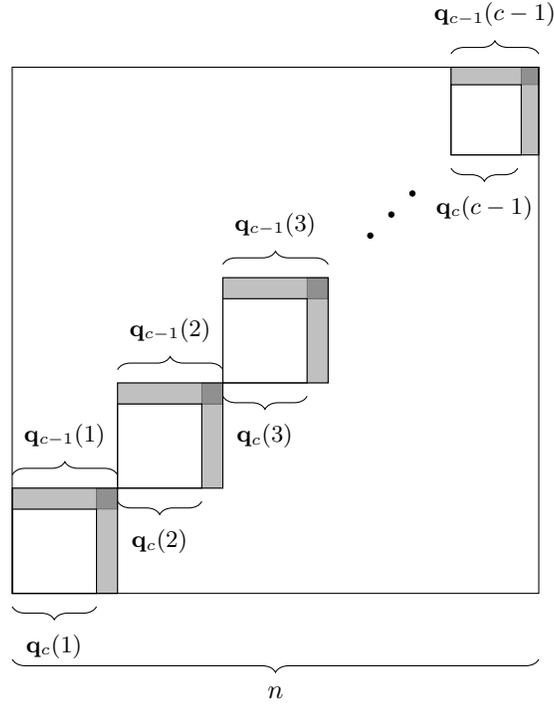
\begin{figure}[h]
\begin{center}
\begin{tikzpicture}[scale = .7]
\foreach \FullWidth/\FullHeight in {10/10} { 
% Try change these numbers. 
%They determine the overall size of the figure.
	\draw (0, 0) -- ++(\FullWidth, 0) 
	  -- ++ (0, \FullHeight)
	  -- ++({-\FullWidth}, 0) 
	  -- cycle;
	\foreach \BoxWidth/\BoxHeight/\SideWidth/\SideHeight in 
	{{\FullWidth/5}/{10/5}/0.2/0.2} { 
% Try change these four numbers that determine the boundaries of the box. 
%The last two numbers should be between 0 and 1.
	\foreach \BoxNumber in {1, 2,3 } {		
		\draw ({(\BoxNumber-1)*\BoxWidth}, {(\BoxNumber-1)*\BoxHeight}) 
		-- ++(\BoxWidth, 0) -- ++(0, \BoxHeight) -- ++({-\BoxWidth}, 0) -- cycle;
		\draw [fill = black, opacity = 0.25] ({\BoxNumber*\BoxWidth},
		 {(\BoxNumber-1)*\BoxHeight}) -- ++({-\SideWidth*\BoxWidth}, 0) 
		 -- ++(0, \BoxHeight) -- ++({\SideWidth*\BoxWidth}, 0) -- cycle; 
		 % Try play with "fill" and "opacity" (between 0 and 1)
		\draw [fill = black, opacity = 0.25] ({(\BoxNumber-1)*\BoxWidth}, 
		{\BoxNumber*\BoxHeight}) 
		  -- ++(0, {-\SideHeight*\BoxHeight}) 
		  -- ++(\BoxWidth, 0) 
		  -- ++(0, {\SideHeight*\BoxHeight}) -- cycle;	
		  };
		\foreach \BoxNumber in {1, 2,3} { % Label for first three boxes
			\draw[decoration = {brace, mirror, amplitude = 5pt, raise = 5pt}, decorate] 
			({\BoxNumber*\BoxWidth}, {\BoxNumber*\BoxHeight}) -- node[above = 12pt] 
			{\small{\(\mathbf{q}_{c-1}({\BoxNumber})\)}} ++({-\BoxWidth}, 0); 
		};
		\foreach \BoxNumber in {1,2,3} { %lower label
			\draw[decoration = {brace, amplitude = 5pt, raise = 5pt}, decorate] 
			({\BoxNumber*\BoxWidth-.4}, {(\BoxNumber-1)*\BoxHeight}) 
			-- node[below = 12pt] 
			{\small{\(\mathbf{q}_{c}({\BoxNumber})\)}} ++({-\BoxWidth+.4}, 0); 
		};
	};
	
	\draw (0, 0) -- ++(1.6, 0) -- ++(0, 1.6) -- ++({-1.6}, 0) -- cycle;
	\draw (2, 2) -- ++(1.6, 0) -- ++(0, 1.6) -- ++({-1.6}, 0) -- cycle;
	\draw (4, 4) -- ++(1.6, 0) -- ++(0, 1.6) -- ++({-1.6}, 0) -- cycle;
	\draw (10 - 10/6, 10-10/6) --++ (4/3,0) --++(0,4/3) --++({-4/3},0) --cycle;
	
\draw [fill = black] (7.6,7.6) circle (.05 cm);
\draw [fill = black] (7.2,7.2) circle (.05 cm);
\draw [fill = black] (6.8, 6.8) circle (.05 cm);

%%MY DUPLICATION FOR LAST SMALLER BOX
	\foreach \BoxWidth/\BoxHeight/\SideWidth/\SideHeight in 
	{{\FullWidth/6}/{10/6}/0.2/0.2} { 
% Try change these four numbers that determine the boundaries of the box. 
%The last two numbers should be between 0 and 1.
	\foreach \BoxNumber in {6 } {
		\draw ({(\BoxNumber-1)*\BoxWidth}, {(\BoxNumber-1)*\BoxHeight}) 
		-- ++(\BoxWidth, 0) -- ++(0, \BoxHeight) -- ++({-\BoxWidth}, 0) -- cycle;
		\draw [fill = black, opacity = 0.25] ({\BoxNumber*\BoxWidth},
		 {(\BoxNumber-1)*\BoxHeight}) -- ++({-\SideWidth*\BoxWidth}, 0) 
		 -- ++(0, \BoxHeight) -- ++({\SideWidth*\BoxWidth}, 0) -- cycle; 
		 % Try play with "fill" and "opacity" (between 0 and 1)
		\draw [fill = black, opacity = 0.25] ({(\BoxNumber-1)*\BoxWidth}, 
		{\BoxNumber*\BoxHeight}) 
		  -- ++(0, {-\SideHeight*\BoxHeight}) 
		  -- ++(\BoxWidth, 0) 
		  -- ++(0, {\SideHeight*\BoxHeight}) -- cycle;	
		  };
		\foreach \BoxNumber in {6} { %upper label
			\draw[decoration = {brace, mirror, amplitude = 5pt, raise = 5pt}, decorate] 
			({\BoxNumber*\BoxWidth}, {\BoxNumber*\BoxHeight}) -- node[above = 12pt] 
			{\small{\(\mathbf{q}_{c-1}({c - 1})\)}} ++({-\BoxWidth}, 0); 
		};
		\foreach \BoxNumber in {6} { %lower label
			\draw[decoration = {brace, amplitude = 5pt, raise = 5pt}, decorate] 
			({\BoxNumber*\BoxWidth-.4}, {(\BoxNumber-1)*\BoxHeight}) 
			-- node[below = 12pt] 
			{\small{\(\mathbf{q}_{c}({c - 1})\)}} ++({-\BoxWidth+.4}, 0); 
		};
%		\foreach \TickX in {1, 2, 3, 4, 5} { % Tick marks on bottom
%			\draw ({\TickX*\BoxWidth}, 0) -- ++(0, 0.25);
%		};
	};
	\draw[decoration = {brace, mirror, amplitude = 5pt, raise = 5pt}, 
	decorate] (0, -1) -- node[below = 12pt] {\(n\)} ++(\FullWidth, 0); 
	% Bottom braced label
};
\end{tikzpicture}
\end{center}
\caption{Diagram of 
$\sum_{i=1}^{c-1} \left( 
\mathbf{q}_{c-1}^2(i) - \mathbf{q}_c^2(i)
\right) $ }
\label{figure1}
\end{figure}

\begin{figure}[h]
\begin{center}
\begin{tikzpicture}[scale = .7]
\foreach \BoxWidth/\BoxHeight/\SideWidth/\SideHeight in {5/5/0.2/0.2} { 
% Try change these four numbers that determine the boundaries of the box. 
%The last two numbers should be between 0 and 1.
	\draw (0, 0) -- ++(\BoxWidth, 0) 
	  -- ++(0, \BoxHeight)
	   -- ++({-\BoxWidth}, 0) -- cycle;
	\draw (0, 0) -- ++(4, 0) -- ++(0, 4) -- ++({-4}, 0) -- cycle;
	\draw [fill = black, opacity = 0.25] (\BoxWidth, 0) -- ++
	  ({-\SideWidth*\BoxWidth}, 0) -- ++
	  (0, \BoxHeight) -- ++({\SideWidth*\BoxWidth}, 0) 
	  -- cycle; % Try play with "fill" and "opacity" (between 0 and 1)
	  
	\draw [fill = black, opacity = 0.25] (0, \BoxHeight) 
	  -- ++(0, {-\SideHeight*\BoxHeight}) 
	  -- ++(\BoxWidth, 0) 
	  -- ++(0, {\SideHeight*\BoxHeight}) -- cycle;
	
	\draw[decoration = {brace, mirror, amplitude = 5pt, raise = 5pt}, decorate] 
	  (0, 0) -- node[below = 12pt] {\small{\(\mathbf{q}_c(i)\)}}
	  ++({(1-\SideWidth)*\BoxWidth}, 0); 
	  % "amplitude" is how steep the curls of the braces are.
	  % "raise", "below", and "right" offsets the brace and labels
	\draw[decoration = 
	{brace, amplitude = 5pt, raise = 5pt}, decorate] 
	(0, 5) -- node[above = 12pt] {\small{\(\mathbf{q}_{c-1}(i)\)}} 
	  ++(\BoxWidth, 0);
	\draw[decoration = {brace, mirror, amplitude = 5pt, raise = 5pt}, decorate] 
	  (\BoxWidth, 0) -- node[right = 12pt] {\(\mathbf{q}_c(i) \geq{g}_c\)} 
	  ++(0, {(1-\SideHeight)*\BoxHeight});
	\draw[decoration = {brace, mirror, amplitude = 5pt, raise = 5pt}, decorate] 
	  (\BoxWidth, {(1-\SideHeight)*\BoxHeight}) -- node[right = 12pt] {\(\geq 2\)} 
	    ++(0, {\SideHeight*\BoxHeight});
};
\end{tikzpicture}
\end{center}
\caption{Diagram of $\mathbf{q}_{c-1}^2(i) - \mathbf{q}_c^2(i)$}
\label{figure2}
\end{figure}
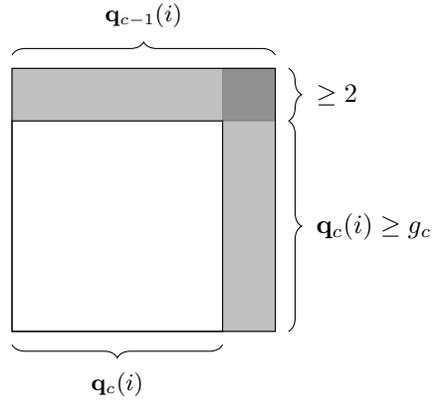

Figure \ref{figure2} depicts  
the difference $\mathbf{q}_{c-1}^2(i) - \mathbf{q}_c^2(i)$ 
for an individual $i \in [c-1]$. This is a closer view 
of one of the individual squares along the main 
diagonal in Figure \ref{figure1}. 
From Figure \ref{figure2}, we can observe that 
the value $(\mathbf{q}_{c-1}^2(i) - 
\mathbf{q}_{c}^2(i))$  can be visualized 
geometrically as the combined areas of 
two types of shapes - the rectangles 
shaded light gray and the square shaded dark gray. 
First, note that there are two identical rectangles (shaded 
light grey) whose dimensions are 
$\mathbf{q}_{c-1}(i) - \mathbf{q}_c(i)$ by $\mathbf{q_c}$. 
We know that for all $i \in [c-1]$, that 
$\mathbf{q}_c(i) \ge q_c$, 
and that 
$\sum_{i=1}^{c-1}
(\mathbf{q}_{c-1}(i)-\mathbf{q}_c(i)) 
= \mathbf{q}_c(c) = q_c$. 
%There is an exceptional case when 
%$\mathbf{q}_c(c)$ could be $q_c -1$, but this 
%is offset with $\mathbf{q}_c(1)$ being $q_c+1$. 
Hence, the sum of the area of all lightly shaded 
rectangles in Figure \ref{figure1} is at least 
$2q_c^2$. 

Next, we will focus on the square shaped region 
(shaded dark grey) in Figure \ref{figure2}. 
The dimensions of this square are
$(\mathbf{q}_{c-1}(i)-\mathbf{q}_c(i))$ by 
$(\mathbf{q}_{c-1}(i)-\mathbf{q}_c(i))$. 
Therefore 
\[
\sum_{i=1}^{c-1} 
\left(
\mathbf{q}_{c-1}^2(i) - \mathbf{q}_c^2(i) 
\right)
\;\;\;\ge \;\;\;
2q_c^2 + \sum_{i=1}^{c-1}(\mathbf{q}_{c-1}(i)-\mathbf{q}_c(i))^2.
\]
We saw previously that 
$\mathbf{q}_{c-1}(i)-\mathbf{q}_c(i) \ge 2$, and 
again using the fact that 
$\sum_{i=1}^{c-1}
(\mathbf{q}_{c-1}(i)-\mathbf{q}_c(i)) 
= \mathbf{q}_c(c) = q_c$, 
the total area of all the sum of all dark grey shaded square 
regions in Figure \ref{figure1} is at least 
$2q_c$. Moreover, each time the 
difference between $\mathbf{q}_{c-1}(i)$ and 
$\mathbf{q}_c(i)$ is greater than $2$, this means 
an overall increase of $(\mathbf{q}_{c-1}(i) - \mathbf{q}_c(i))^2$ 
by at least $3$. Hence $\sum_{i=1}^{c-1}(\mathbf{q}_{c-1}(i) - \mathbf{q}_{c}(i))^2
\ge 2q_c + 3k$, which implies that expression 
(\ref{eqn10}) is greater than or equal to 
\begin{align}\label{eqn11}
\frac{q_c^2}{2} + \frac{q_c}{2} + \frac{3k}{4} 
- \frac{q_c^2}{4} - \frac{k}{4} - 1 
\;\;\;=\;\;\; \frac{q_c^2}{4} +\frac{q_c}{2} + \frac{k}{2} -1.
\end{align}

To complete the proof in the general case, recall that 
$q_c > n/c-1$. Thus, by 
replacing $q_c$ with $n/c-1$, 
we have that the right hand side of 
equation (\ref{eqn11}), afer some basic reduction, 
is greater than  
\begin{align}\label{eqn12}
\frac{r^2}{c^2} - \frac{1}{4} + \frac{k}{2} - 1. 
\end{align}
In this final expression, since $c \le \hat{c}$, 
we have $r^2/c^2 \ge r$. Also, since 
$q_c$ was strictly greater than $2(c-1)$, we
know that $k \ge 1$, completing the proof in the general case. \\

For the exceptional case, when $\mathbf{q}_c(c) = q_c-1$, 
we can follow the same argument as in the general case 
with slight modification. In this instance the last term in 
expression (\ref{eqn10}) is reduced since $\mathbf{q}_c(c) = q_c-1$ 
rather than $q_c$, resulting in a larger overall value. 
Using this fact, we can then show that whenever 
$\mathbf{q}_c(c) = q_c-1$, expression (\ref{eqn10}) is greater or equal to 
\[
\frac{1}{4} 
\overset{c-1}{\underset{i=1}{\sum}}
\left(
\mathbf{q}_{c-1}^2(i) - \mathbf{q}_c^2(i) 
\right)
- \frac{k}{4} - \frac{q_c^2}{4} + \frac{q_c}{2} - \frac{5}{4}. \numberthis \label{eqn13}
\]
By a similar argument to the general case, we can also show that 
if $\mathbf{q}_c(c) = q_c-1$, then 
\[
\sum_{i=1}^{c-1} 
\left(
\mathbf{q}_{c-1}^2(i) - \mathbf{q}_c^2(i) 
\right)
\;\;\; \ge \;\;\;
2q_c^2 +2. 
\]
This fact, along with the fact that $q_c > n/c-1$, 
implies that expression (\ref{eqn10}), 
and therefore the left hand side of inequality (\ref{eqn8}), is greater than 
\[
\frac{r^2}{c^2} + \frac{k}{2} - 1 \;\;\;\ge\;\;\; r -1.
\]
%\end{proof}

\section{}\label{appendix G}
\textit{Proof of Lemma \ref{floorrootr}:}  \\
We first show that 
$\hat{c} \;\le\; \lfloor \hat{c} \rfloor + 0.5$ implies
\begin{align*} 
%\overset{\lfloor \hat{c} \rfloor}{\underset{i=1}{\sum}}
%\left\lfloor
%\left(\frac{\mathbf{q}_{\lfloor \hat{c} \rfloor}(i)-2}{2}\right)^2
%\right\rfloor 
\psi(\mathbf{q}_{\lfloor\hat{c}\rfloor})
- 
%\overset{\lceil \hat{c} \rceil}{\underset{i=1}{\sum}}
%\left\lfloor
%\left(\frac{\mathbf{q}_{\lceil \hat{c} \rceil}(i)-2}{2}\right)^2
%\right\rfloor 
\psi(\mathbf{q}_{\lceil \hat{c} \rceil}) 
\;\;\; \le \;\;\; 
r-1.
\end{align*}

Let $\hat{c} \;\le\;\lfloor\hat{c}\rfloor + 0.5$. Then 
by Lemma \ref{r equivalence}, we know that 
$r \;\le\; \lfloor \hat{c} \rfloor^2 + \lfloor \hat{c} \rfloor$.
From here, we will split the proof into two possibilities: (1)
where
$r \; < \; \lfloor \hat{c} \rfloor^2 + \lfloor \hat{c} \rfloor$;
and (2) where 
$r \;=\; \lfloor \hat{c} \rfloor^2 + \lfloor \hat{c} \rfloor$.

First, suppose that 
$r \;<\; \lfloor \hat{c} \rfloor^2 + \lfloor \hat{c} \rfloor$.
In the event that $\hat{c} = \lfloor \hat{c} \rfloor$, then 
$\lfloor \hat{c} \rfloor = \lceil \hat{c} \rceil$, which implies that 
the left hand side of inequality (\ref{eqn15}) is equal to $0$, 
so that the conclusion holds trivially. Thus we will assume that 
$\lfloor \hat{c} \rfloor < \hat{c}$, which implies that 
$\lceil \hat{c} \rceil = \lfloor \hat{c} \rfloor + 1$. 
From here, for ease of notation, and in order to see the 
connection to Lemma \ref{cplusone} more clearly, 
let $c : = \lfloor \hat{c} \rfloor$ 
so that $c+1 = \lceil \hat{c} \rceil$. 

Next, note that 
\[
q_{c+1} \;=\; \frac{n-rem_{c+1}}{c+1} 
\;\le\;
\frac{n}{c+1} 
\;=\;
\frac{2r}{c+1}, 
\]
and that 
\[
\frac{2r}{c+1} \;<\; 2c  
\text{ \; if and only if \; } 
r \;<\; c(c+1) \;= \; \lfloor \hat{c}\rfloor^2 + \lfloor \hat{c}\rfloor, 
\]
which is true by assumption. 
Therefore $q_{c+1} < 2c$. 
From this point, the proof for the case where 
$r \;<\; \lfloor \hat{c} \rfloor^2 + \lfloor \hat{c} \rfloor$ 
is the same as that of Lemma \ref{cplusone}. 
\\

We now consider the second possibility. Assume that 
$r \;=\; \lfloor \hat{c} \rfloor^2 + \lfloor \hat{c} \rfloor$. 
We claim that this implies that the left hand side of 
inequality (\ref{eqn15}) is exactly equal to $r-1$. 
This also has the implication that $\lfloor \hat{c} \rfloor$ 
and $\lceil \hat{c} \rceil$ both yield the same maximum sphere size. 
To see why the claim is true, note first that 
\[
r \;=\; \lfloor \hat{c} \rfloor^2 + \lfloor\hat{c} \rfloor 
\;=\;
\lfloor \hat{c} \rfloor(\lfloor \hat{c} \rfloor + 1) 
\;=\;
\lfloor \hat{c} \rfloor \cdot \lceil\hat{c}\rceil. 
\]
This implies that 
\[
q_{\lfloor\hat{c}\rfloor} \;= \;\frac{n}{\lfloor\hat{c}\rfloor}
\;=\; 2\lceil \hat{c} \rceil 
\text{ \; and that \;}
q_{\lceil\hat{c}\rceil} \;=\; \frac{n}{\lceil\hat{c}\rceil} 
\;=\; 2\lfloor \hat{c}\rfloor.
\]
Therefore we have 
\begin{align*}
%\overset{\lfloor \hat{c} \rfloor}{\underset{i=1}{\sum}}
%\left\lfloor
%\left(\frac{\mathbf{q}_{\lfloor \hat{c} \rfloor}(i)-2}{2}\right)^2
%\right\rfloor
\psi(\mathbf{q}_{\lfloor \hat{c} \rfloor}) 
\;\;=\;\; &
%\overset{\lfloor \hat{c} \rfloor}{\underset{i=1}{\sum}}
%\left\lfloor
%\left(\frac{2\lceil\hat{c}\rceil -2}{2}\right)^2
%\right\rfloor \\
\psi( (
\underset{\lfloor \hat{c} \rfloor}{\underbrace{
2\lceil \hat{c} \rceil, 2\lceil \hat{c} \rceil, \dots 2\lceil \hat{c} \rceil}}
))
%\;\;=\;\; &
%\overset{\lfloor \hat{c} \rfloor}{\underset{i=1}{\sum}}
%\left(\frac{2\lceil\hat{c}\rceil -2}{2}\right)^2 
\;\;=\;\; 
\lfloor \hat{c}\rfloor \left( \lfloor \hat{c} \rfloor\right)^2 
\;\;=\;\; 
\lfloor \hat{c} \rfloor ^3, \numberthis \label{eqn19}
\end{align*}
and similarly, 
\begin{align*} 
%\overset{\lceil \hat{c} \rceil}{\underset{i=1}{\sum}}
%\left\lfloor
%\left(\frac{\mathbf{q}_{\lceil \hat{c} \rceil}(i)-2}{2}\right)^2
%\right\rfloor 
\psi(\mathbf{q}_{\lceil \hat{c} \rceil})
\;\;=\;\; &
%\overset{\lceil \hat{c} \rceil}{\underset{i=1}{\sum}}
%\left\lfloor
%\left(\frac{2\lfloor \hat{c} \rfloor -2}{2}\right)^2
%\right\rfloor  
\psi( (
\underset{\lceil \hat{c} \rceil}{\underbrace{
2\lfloor \hat{c} \rfloor, 2\lfloor \hat{c} \rfloor, \dots 2\lfloor \hat{c} \rfloor}}
))
\;\;=\;\; 
%\overset{\lceil \hat{c} \rceil}{\underset{i=1}{\sum}}
%\left(\frac{2\lfloor \hat{c} \rfloor -2}{2}\right)^2  \\
%\;\;=\;\; &
\left(\lfloor \hat{c} \rfloor + 1\right)\left(\lfloor \hat{c}\rfloor - 1\right)^2
\;\;=\;\;
\lfloor \hat{c} \rfloor ^3 - \lfloor \hat{c} \rfloor^2 - \lfloor \hat{c} \rfloor + 1.
\numberthis \label{eqn20}
\end{align*}
Finally, subtracting (\ref{eqn19}) and (\ref{eqn20}), 
we have that the left hand side of inequality 
(\ref{eqn15}) is equal to 
$\lfloor \hat{c} \rfloor^2 + \lfloor \hat{c} \rfloor - 1$, 
which is equal to $r-1$ by the assumption that 
$r = \lfloor \hat{c} \rfloor^2 + \lfloor \hat{c} \rfloor$. 
This completes the first half of the proof. 
 \\~\\~\\
We next show that 
$\hat{c} \;\ >\; \lfloor \hat{c} \rfloor + 0.5$ implies 
\begin{align}\label{eqn16}
%\overset{\lfloor \hat{c} \rfloor}{\underset{i=1}{\sum}}
%\left\lfloor
%\left(\frac{\mathbf{q}_{\lfloor \hat{c} \rfloor}(i)-2}{2}\right)^2
%\right\rfloor 
\psi(\mathbf{q}_{\lfloor \hat{c} \rfloor})
- 
%\overset{\lceil \hat{c} \rceil}{\underset{i=1}{\sum}}
%\left\lfloor
%\left(\frac{\mathbf{q}_{\lceil \hat{c} \rceil}(i)-2}{2}\right)^2
%\right\rfloor 
\psi(\mathbf{q}_{\lceil \hat{c} \rceil})
\;\;\; > \;\;\; 
r-1.
\end{align}

Let $\hat{c} \;>\; \lfloor \hat{c} \rfloor + 0.5.$ 
For ease of notation and to see the 
connection to Lemma \ref{cminusone}, 
let $c := \lceil \hat{c} \rceil$ so that 
$c-1 = \lfloor \hat{c} \rfloor$.
By
Lemma \ref{r equivalence}, we have 
\[
r \;\;> \;\; (c-1)^2 + (c-1)
\;\;=\;\; (c-1)(c)
\text{ \; which implies that \; }
\frac{r}{c} \;\;>\;\;c-1.
\]
 Note that 
\[
q_c \;\;=\;\; \frac{n-rem_c}{c} 
\;\;>\;\;
\frac{n}{c} -1 
\;\;=\;\;
\frac{2r}{c} -1 
\;\; > \; \;
2(c-1) -1, 
\text{ \; which implies that \;} 
q_c \ge 2(c-1). 
\]
At the same time,  note that 
\[
q_c \;\;=\;\; \frac{n-rem_c}{c}
 \;\; \le \;\;
 \frac{n}{c}
 \;\;=\;\; 
 \frac{2r}{c}  
 \;\; < \;\; 
 \frac{2r}{\hat{c}-1}, 
\]
which is easily shown to be 
strictly less than $3(\hat{c}-1)$ 
as long as $r \ge 30$, and clearly 
$3(\hat{c}-1)$ is strictly less than $3(c-1)$. 
It is also easily verified numerically that 
inequality (\ref{eqn16}) is satisfied for $r <30$ 
(assuming that $\hat{c} \;>\; \lfloor\hat{c}\rfloor +0.5)$.  
Hence, for the remainder of the proof, 
we shall assume that $r \ge 30$, which means that 
\[
2(c-1) \;\;\le\;\; q_c \;\;<\;\;3(c-1).
\]
From here we split into two cases, 
(1) when $q_c = 2(c-1)$ exactly;  and 
(2) when $q_c > 2(c-1)$. 

First, assume that $q_c = 2(c-1)$. 
Then following similar logic to the proof of 
Lemma \ref{cminusone}, 
we know that for all $i \in [c-1]$, that 
$\mathbf{q}_{c-1}(i) - \mathbf{q}_c(i) \;=\;2$. 
Technically this is assuming that we are not in the 
special case when $n/(c-1)$ is an odd integer, in 
which case $\mathbf{q}_{c-1}(1) -\mathbf{q}_c(1) = 3$ 
and $\mathbf{q}_{c-1}(c-1) - \mathbf{q}_c(c-1) = 1$. 
However, this would result in an overall increase of the 
left hand side of inequality (\ref{eqn16}), so it 
is enough to consider the general case when $n/(c-1)$ 
is not an odd integer. 

Since $\mathbf{q}_{c-1}(i) - \mathbf{q}_c(i) = 2$ for each 
$i \in [c-1]$, then the left hand side of inequality (\ref{eqn16}) 
is equal to 
\begin{align*}
%\overset{c-1}{\underset{i=1}{\sum}}
%\left\lfloor
%\left( 
%\frac{\mathbf{q}_{c-1}(i)-2}{2}
%\right)^2
%\right\rfloor
\psi(\mathbf{q}_{c-1})
-
%\overset{c-1}{\underset{i=1}{\sum}}
%\left\lfloor
%\left( 
%%\frac{(\mathbf{q}_{c-1}(i)-4}{2}
%\right)^2
%\right\rfloor
\psi(\mathbf{q}_{c-1}-2)
- 
%\left\lfloor
%\left(
%\frac{\mathbf{q}_c(c)-2}{2}
%\right)^2
%\right\rfloor.
\psi(\mathbf{q}_c(c))
\end{align*}
By Remark \ref{floordiff} and 
the fact that $\mathbf{q}_c(c) = q_c = 2(c-1)$, the above expression 
is equal to 
\begin{align*}
&\overset{c-1}{\underset{i=1}{\sum}}
\left( 
\frac{\mathbf{q}_{c-1}(i)-2}{2}
\right)^2
-
\overset{c-1}{\underset{i=1}{\sum}}
\left( 
\frac{(\mathbf{q}_{c-1}(i)-4}{2}
\right)^2
- 
\left(
\frac{(2(c-1)-2}{2}
\right)^2 \\
\;\;=\;\; &
\overset{c-1}{\underset{i=1}{\sum}}
\left( 
\frac{\mathbf{q}^2_{c-1}(i)}{4}
- \mathbf{q}_{c-1}(i) + 1
\right)
- 
\overset{c-1}{\underset{i=1}{\sum}}
\left( 
\frac{\mathbf{q}^2_{c-1}(i)}{4}
- 2\mathbf{q}_{c-1}(i) + 4
\right)
- \left( c-1 \right)^2 \\
\;\;=\;\;&
\overset{c-1}{\underset{i=1}{\sum}}
\left( 
\mathbf{q}_{c-1}(i) - 3 
\right)
- \left( c^2 - 4c+4\right)^2  \\
\;\;=\;\; &n - 3(c-1) - c^2 + 4c - 4 \\
\;\;=\;\; 
& 2r + c -c^2 -1  \\
\;\; = \;\; &r-1 + r-(c-1)c. \numberthis \label{eqn21}
\end{align*}
We saw earlier that $r > (c-1)c$, 
so expression \ref{eqn21} is greater than $r-1$. 
This completes the proof for the case when 
$q_c = 2(c-1)$. \\

Next suppose that $q_c > 2(c-1)$. Let 
$k := \# \{ i \in [c-1] \;:\; \mathbf{q}_{c-1}(i) - \mathbf{q}_c(i) \;=\;3\}.$ 
In other words, $k$ is the number of cuts in $\mathbf{q}_{c-1}$
that are decreased by $3$ in the construction of $\mathbf{q}_c$ 
from $\mathbf{q}_{c-1}$. Recall that 
$q_c < 3(c-1)$, so for all $i \in [c-1]$, we have 
$\mathbf{q}_{c-1}(i) - \mathbf{q}_c(i) \le 3$. 
Since we are also assuming that $q_c > 2(c-1)$, we know 
that $q_c = 2(c-1) + k$. This is because $q_c$ is equal to  
$2$ times the number of cuts in $\mathbf{q}_{c-1}$ decreased 
by $2$, plus $3$ times the number of cuts in $\mathbf{q}_{c-1}$ 
decreased by $3$ in the construction of $\mathbf{q}_c$. 

Following the same reasoning as in the proof of 
Lemma \ref{cminusone}, we can then show that 
the left hand side of inequality \ref{eqn16} is greater 
or equal to 
\begin{align}\label{eqn22}
\frac{q_c^2}{2} + \frac{q_c}{2} + \frac{k}{2} - 1.
\end{align}
Expression \ref{eqn22} is the same expression obtained as 
the right hand side of equation \ref{eqn11} in the proof 
of Lemma \ref{cminusone}. At this stage, however, we recall 
the fact that under the current assumptions, $q_c = 2(c-1)+k$. 
Substituting $2(c-1)+k$ for $q_c$ and simplifying, 
after some rearranging we
obtain that expression \ref{eqn22} is equal to 
\[
c^2 - 1 + c(k-1) + \frac{k^2}{4},
\]
which is greater than 
$r-1 + c(k-1) + k^2/4$ since by the assumption that 
$\hat{c} > \lfloor \hat{c} \rfloor + 0.5$, we know 
$ c \;=\; \lceil \hat{c} \rceil \;>\; \hat{c}$. 
Finally, this last expression is greater than $r-1$ as long as 
$k \ge 1$, which we know is true since 
$q_c \;>\; 2(c-1)$. 
\hfill $\square$

\section*{Acknowledgment}
The author would like to thank M. Hagiwara for 
his valuable comments and discussion. 
This paper is partially supported by 
KAKENHI 16K12391 and KAKENHI 26289116.

\ifCLASSOPTIONcaptionsoff
  \newpage
\fi

\end{document}